\documentclass[aos,preprint]{imsart}

\RequirePackage{amsthm,amsmath,amsfonts,amssymb}
\usepackage{amsmath,amssymb,graphics,dsfont,bm,xcolor,float}
\usepackage{graphics}
\usepackage{color}
\usepackage{amsfonts}
\usepackage{amsmath}
\usepackage{lmodern}
\usepackage{subfig}
\RequirePackage[numbers]{natbib}
\usepackage{subfig}
\usepackage[demo]{graphicx}
\usepackage{enumitem}

\startlocaldefs
\newtheorem{Ass}{Assumption}
\newtheorem{lemma}{Lemma}
\newtheorem{theorem}{Theorem}

\endlocaldefs

\begin{document}

\begin{frontmatter}
\title{Wilks' theorem for semiparametric regressions  with weakly dependent data}
\runtitle{Wilks' theorem for semiparametric regressions}

\begin{aug}
\author[A]{\fnms{Marie} \snm{Du Roy de Chaumaray}\ead[label=e1,mark]{marie.du-roy-de-chaumaray@ensai.fr}},
\author[A]{\fnms{Matthieu} \snm{Marbac}\ead[label=e2,mark]{matthieu.marbac-lourdelle@ensai.fr}}
\and
\author[A]{\fnms{Valentin} \snm{Patilea}\ead[label=e3,mark]{valentin.patilea@ensai.fr}}
\address[A]{Univ. Rennes, Ensai, CNRS, CREST - UMR 9194, F-35000 Rennes, France, \printead{e1,e2,e3}}
\end{aug}

\begin{abstract}
\\ The empirical likelihood inference is extended to a class of semiparametric   models for stationary, weakly dependent series. A partially linear single-index regression is used for the conditional mean of the series given its past, and the present and past values of a vector of covariates. A parametric model for the conditional variance of the series is added to capture further nonlinear effects. 
We propose 
suitable moment equations which characterize the mean and variance model. We derive an empirical log-likelihood ratio which includes nonparametric estimators of several functions, and we show that this ratio behaves asymptotically as if the functions were given. \\
\end{abstract}

\begin{keyword}[class=MSC2020]
\kwd[Primary ]{62M10}
\kwd[; secondary ]{62F03, 62G20}
\end{keyword}

\begin{keyword}
\kwd{Kernel smoothing}
\kwd{$\alpha-$mixing}
\kwd{Nonlinear time series}
\kwd{Nuisance function}
 \kwd{Parametric inference} 
 \kwd{Pivotal statistic}
\end{keyword}

\end{frontmatter}
\section{Introduction}

We aim modeling and doing inference for one-dimensional  time series $(Y_i)$ given a vector-valued time series $(V_i)$ and the past values of $Y_i$ and $V_i$,  $i\in\mathbb{Z}$.  For this purpose we propose flexible semiparametric models for conditional mean and  conditional variance of $Y_i$. Formally, let $(Z_i)$ be a strictly stationary and strongly mixing  sequence of random vectors with $Z_i=(V_i^\top,\varepsilon_i)^\top\in\mathbb{R}^{d_X+d_W}\times\mathbb{R}$ where $V_i=(X_i^\top,W_i^\top)^\top\in\mathbb{R}^{d_X}\times\mathbb{R}^{d_W}$. Let $(\mathcal F_i)$ 
 be its natural filtration. For any positive integer $r$, we denote the $r$ lagged values of $Z_i$ by $Z_{i}^{\{r\}}=(V_{i-1}^\top, Y_{i-1}, \ldots, V_{i-r}^\top, Y_{i-r})^\top$.

 Let us consider the  semiparametric model defined  by
\begin{equation}\label{eq:modelbase}
Y_i = \mu(V_i;\gamma,m) + \varepsilon_i \;\textrm{ with }\; \mu(V_i;\gamma,m)= l(X_i; \gamma_1) + m(W_i^\top \gamma_2),
\end{equation}
where
\begin{equation}\label{eq:modelbase_b}
\mathbb E[\varepsilon_i\mid V_i,\mathcal F_{i-1}]=0,\quad 
\end{equation}
and
\begin{equation}\label{eq:modelbase_c}
  \mathbb E[\varepsilon_i^2\mid V_i,\mathcal F_{i-1}]=\sigma^2(V_i,Z_{i}^{\{r\}};\beta),
\end{equation}
$\gamma=(\gamma_1^\top,\gamma_2^\top)^\top\!$,   $\theta=(\gamma^\top,\beta^\top)^\top$ and $m(\cdot)$ is an infinite dimensional parameter. Thus $\theta$  gathers the finite dimensional parameters, and our interest will focus on this vector,   while $m(\cdot)$ is considered as a nuisance parameter.  The value of $r$,  as well as the real-valued functions $l(\cdot)$ and $\sigma^2(\cdot)$, are given. Moreover, the functions we consider for $\sigma^2(\cdot)$ do not require to know the infinite dimensional parameter $m(\cdot)$. Let $\theta_0$ and $m_0(\cdot)$ denote the true values of the finite and infinite-dimensional parameters of the model, respectively. The vector $V_i$ may include common random variables and/or lagged values of $Y_i$, as well as exogenous covariates. We call a model defined by 
\eqref{eq:modelbase}-\eqref{eq:modelbase_c} a CHPLSIM which stands for \emph{Conditional Heteroscedastic Partially Linear Single-Index Model}. The methodology we will propose in the sequel allows us to replace  \eqref{eq:modelbase_c} by a higher order moment equation, or to add higher order moments to  \eqref{eq:modelbase_c}. For the sake of simplicity we keep \eqref{eq:modelbase_c} and we will only mention such possible extensions in the conclusion section.  

CHPLSIM is related to the model proposed by \citet{Lian_Carroll_2015} in the case of independent observations following the same distribution.
Our model covers a wide class of models for weakly dependent and independent data. First, with $l(X_i; \gamma_1) = X_i^\top \gamma_1$,  CHPLSIM includes  the partially linear single-index model (PLSIM) \citep{Carroll97} in which the errors $\varepsilon_i$ are independent and identically distributed (i.i.d.) variables and $V_i$ are independent covariates. Such semiparametric models were originally used to overcome the curse of dimensionality inherent to nonparametric regression on $W_i$ by making use of a single-index $W_i^\top \gamma_2$. The PLSIM includes the partially linear models with a single variable in the nonparametric part. Our non-i.i.d. framework allows for heteroscedasticity in the errors of PLSIM, with the conditional variance of the errors possibly depending of both the covariates and the lagged errors values.  For instance, it allows  martingale difference errors, as considered by  \citet{chen2008empirical} and \citet{fan2010empirical}. 
  \citet{Xia99} considered a model defined by \eqref{eq:modelbase} for strongly mixing stationary time series, with identity function $l(\cdot)$, $X_i=W_i$ and $W_i$ admitting a density. Their study focuses on the estimation of the parameters in the conditional mean function using kernel smoothing, without investigating the conditional variance, as allows condition \eqref{eq:modelbase_c}. In the same type of model, using local linear smoothing, \citet{XIA20061162} allowed for $X_i$ not necessarily equal to $W_i$ and, at the price of a trimming, relaxed the condition of a density for $W_i$ to a density for the index $W_i^\top\gamma_2$. More recently, using orthogonal series expansions, \citet{Dong2016} extended the model defined by \eqref{eq:modelbase} to the case where $X_i=W_i$  is a multi-dimensional integrated process.

Model \eqref{eq:modelbase}-\eqref{eq:modelbase_b} is also related to and extends a large class of location-scale type models called conditionnal heteroscedastic autoregressive nonlinear (CHARN) models  \citep{hardle1998nonparametric, kanai2010estimating}.   CHARN models include many well-known models 
widely used with application areas  as different as foreign exchange rates \citep{Bossaerts96} or brain and muscular wave analysis \citep{Kato06}. For general nonlinear autoregressive processes, we refer to the book of \citet{Tong90} for the basic definitions as well as numerous applications on real data sets. More generally, nonparametric techniques for nonlinear AR processes can be found in the review of \citet{Hardle_NLAR}.  CHPLSIM allows for a semiparametric specification of the conditional mean and for exogenous covariates. 

We are  interested in inference on the finite dimensional parameter $\theta$  constituted of finite-dimensional parameters from both the conditional mean and the conditional variance functions. When the interest focuses  on the parameters of the conditional mean,  it suffices to consider equations \eqref{eq:modelbase}-\eqref{eq:modelbase_b} with a fully nonparametric conditional variance $\sigma^2(\cdot)$. However, in the time series context, modeling the variance can be important, for instance for forecasting purposes. For our inference purpose, we propose a semiparametric empirical likelihood approach with infinite-dimensional nuisance parameters. Empirical likelihood (EL), introduced by \citet{owen1988empirical,owen2001empirical}, is a general inference approach for models specified by moment conditions. 
Under the assumption of independence between observations, empirical likelihood has been used for inference on finite dimensional parameters  into regression models and unconditional moment equations. See \citet{qin1994Annals}; see also the review of \citet{chen2009}.

Under  i.i.d. data assumption, \citet{wang1999empirical,wang2003empirical} and \citet{lu2009empirical} study the conditions implying that the  empirical likelihood log-ratio (ELR)  still converges to a chi-squared distribution for the partially linear model. Due to the curse of the dimensionality, the performances of the nonparametric estimators decrease dramatically with the number of variables.  \citet{xue2006empirical} and \citet{zhu2006empirical}  show that, if the density of the index is bounded away from zero, the ELR converges to a chi-squared distribution and thus permits parameter testing, for single-index model and PLSIM respectively (see also \citet{ZHU2010850}). 

The aim of this paper is to propose a novel general semiparametric regression framework for EL inference which allows for dependent data. 
Some related  cases have been considered in the literature. For instance,  the ELR with longitudinal data has been considered by \citet{xue2007empirical}, for the partially linear model,  and by \citet{li2010empirical}, for PLSIM. In their framework, the convergence of the ELR is guaranteed by the independence between individuals for which a finite bounded number of repeated observations are available. Empirical likelihood has also been used for specific models in times series (see the review of \citet{nordman2014review}\color{black} ; see also \citet{CHANG2015283}\color{black}). Most of the methods developed in this context are based on a blockwise version of empirical likelihood, first introduced by \citet{Kitamura97}. 
A large amount of generalizations have been proposed in the literature depending on the type of dependency. 
We refer to \citet{nordman2014review} for an overview of those techniques of blocking. However, in such an approach, one has to tune additional 
parameters such as the number, the length or the overlapping of the blocks, which might be a complex task.

Our contribution is the extension of the EL inference approach to the case of CHPLSIM defined by \eqref{eq:modelbase}-\eqref{eq:modelbase_c}, for weakly dependent data. This extension is realized without imposing the density of the index bounded away from zero, as it is usually assumed in the literature in the case of i.i.d. data. See, for instance,  \citet{zhu2006empirical}, \citet{ZHU2010850} and  \citet{Lian_Carroll_2015}. Such a very convenient, though quite stringent, condition implies a bounded support for the index, a restriction which makes practically no sense in a general time series framework. 
To obtain our results, a preliminary crucial step before using EL consists in building a fixed number of suitable  unconditional moment equations equivalent to  conditional moment equations defining the regression model.  \color{black}  By the definition of these unconditional moment equations, our approach will not require a blocking data technique. \color{black} 
Then, we follow the lines of \citet{qin1994Annals}, with the difference of the presence of infinite-dimensional nuisance parameters. We show that the nonparametric estimation of the nuisance parameters does not affect the asymptotics and the ELR still converges to a chi-squared distribution.   The negligibility of the nonparametric estimation effect is obtained under mild conditions on the smoothing parameter. 
\citet{CHANG2015283} studied  the EL inference for unconditional moment equations under strongly mixing conditions, with the number of moment equations allowed to increase with the sample size. Since conditional moment equations models could be approximated by models defined by a large number of unconditional moment equations, in principle, \citet{CHANG2015283} could also consider semiparametric models. However, the practical effectiveness of their approach  remains an uninvestigated issue.

In Section~\ref{sec:CME} we consider the profiling approach for  the nuisance parameter $m(\cdot)$ and the identification issue for the finite-dimensional parameters. Next, we  establish the equivalence between our model equations and suitable unconditional moment estimating equations for a martingale difference sequence in Section~\ref{sec:UCME}. The number of unconditional equations is given by the dimension of the vector of identifiable parameters in the (CH)PLSIM. 
Section~\ref{sec:EL} presents the  ELR  and the Wilks' Theorem in our context.  
Section~\ref{sec:num} illustrates the   methodology  by numerical experiments and an application using daily pollution data inspired by the study of \citet{Lian_Carroll_2015}. 
 Section~\ref{sec:concl} contains some additional discussion. The proofs and mathematical details are presented in Appendix. Some technical details, additional simulation
results and  real data analysis results are collected in online Supplementary Material. 


\section{Conditional moment equations} \label{sec:CME}
\subsection{The model}
Let 
$$g_\mu(Z_i;\gamma,m)=Y_i- \mu(V_i;\gamma,m)  ,$$
with $\mu(\cdot)$ defined in \eqref{eq:modelbase}.
The partially linear single index model (PLSIM) is defined by conditional moment equation
\begin{equation}\label{model:easy}
\mathbb E [g_\mu(Z_i;\gamma,m)\mid V_i, \mathcal F_{i-1}]=0\Longleftrightarrow
\gamma=\gamma_0  \textrm{ and } m = m_0.
\end{equation}
In such case, the conditional variance of the residuals has to be finite but does not necessarily have a parametric form. 

The conditionally heteroscedastic  partially linear single index model (CHPLSIM) is defined by two conditional moment equations. For this case, we assume that the second-order conditional moment of the residuals has a semiparametric form. More precisely, the model is defined by the following conditional moment equations
\begin{equation}\label{model:full}
\left\{\begin{array}{rl}
\mathbb E [g_\mu(Z_i;\gamma,m)\mid V_i, \mathcal F_{i-1}]=0 \\
\mathbb E [g_\sigma(Z_i,Z_{i}^{\{r\}};\theta,m)\mid V_i, \mathcal F_{i-1}]=0
\end{array}\right.
\Longleftrightarrow
\theta=\theta_0 \textrm{ and } m = m_0,
\end{equation}
where 
\begin{equation}\label{g_sig}
g_\sigma(Z_i,Z_{i}^{\{r\}};\theta,m) = g_\mu^2(Z_i;\gamma,m) - \sigma^2(V_i,Z_{i}^{\{r\}};\beta),
\end{equation}
with $\sigma^2 (\cdot)$ defined in \eqref{eq:modelbase_c}.

\subsection{Profiling nuisance parameter} \label{sec:profiling}

The model defined by \eqref{eq:modelbase}-\eqref{eq:modelbase_b} requires a methodology for estimating $\theta$ and $m$, with $m$ being in a function space. A common approach, that avoids a simultaneous search involving an infinite-dimensional parameter, is the profiling \citep{severini1992,liang2010}, which defines
\begin{equation*}
m_\gamma (t) = \mathbb{E}[Y_i-l(X_i;  \gamma_1) \mid W_i^\top \gamma_2=t],\qquad t\in\mathbb{R}. 
\end{equation*}
As usually with such approach, in the following it will be assumed that  
\begin{equation}\label{simeq1}
m_{\gamma_{0}}(W_i^\top \gamma_{0,2})=m_0(W_i^\top \gamma_{0,2}).
\end{equation}    
Hence, one expects that, for each $x,w$, the value $\gamma_0$ realizes  the minimum of 
$$\gamma\mapsto \mathbb{E}[\{ Y_i- l(x_i; \gamma_1)  - m_\gamma(w_i^\top \gamma_{2}) \}^2\mid X_i=x_i,W_i=w_i,\mathcal{F}_{i-1}].$$ 
However, even if $m_\gamma(\cdot)$ is well defined for any $\gamma=(\gamma_1^\top,\gamma_2^\top)^\top\in\Gamma\subset \mathbb R^{d_1}\times \mathbb R^{d_W}$, in general the value $\gamma_0$ could not be the unique parameter value with this minimum property. More precisely, in general the true value of the vector $\gamma_2$ is not identifiable and only its direction could be consistently estimated. The standard remedies to this identifiability issue are detailed in the following.

\subsection{Identifiability of the finite-dimensional  parameters} \label{sec:ident}

Concerning the identification of $\gamma _1 \in  \mathbb R^{d_1}$, a minimal requirement is that as soon as $l(X_i;\gamma_1) = l(X_i;\gamma_1^\prime) $ a.s., then necessarily $\gamma_1=\gamma_1^\prime$. For instance, when $l(X_i;\gamma_1) = X_i^\top \gamma_1$, and thus $d_1=d_X$, then necessarily  $\mathbb{E}(X_iX_i^\top)$  invertible. The nonparametric part $m_\gamma(\cdot)$  induces some more constraints. It could absorb any intercept in the model equation. Thus, in particular, when  $l(X_i;\gamma_1) = X_i^\top \gamma_1$, the vectors  $X_i$ and $W_i$ should not contain constant components. 

There are two common approaches to restrict $\gamma_{2}$ for identification purposes: either fix one component equal to 1 \citep{Ma2013}, or set the norm of $\gamma_2$ equal to 1 and the sign of one of its components  \citep{zhu2006empirical}. Without loss of generality, we choose the first component of $\gamma_2$ to impose the constraints of value or sign. When the value of the first component is fixed,  the parameter $\gamma_2$ could be redefined as $\gamma_2=(1,\widetilde \gamma_2^\top)^\top$ where $\widetilde \gamma_2 \in \mathbb{R}^{d_W -1}$. The Jacobian matrix of this reparametrization of $\gamma_2$ is the $d_W\times (d_W -1)$ matrix 
\begin{equation}\label{Jacobi1}
\mathbf{J}_2 (\gamma_2) = \frac{\partial  \gamma_2}{\partial \widetilde \gamma_2} =  \begin{pmatrix}
    \mathbf{0}_{1\times (d_W-1)} \\
    \mathbf{I}_{d_W-1}
  \end{pmatrix},
\end{equation}
where here $\mathbf{0}_{1\times (d_W-1)}$ denotes the null  $1\times (d_W-1)-$matrix, while  $\mathbf{I}_{d_W-1}$ is the $(d_W-1)\times (d_W-1)$ identity matrix. With the second identification approach mentioned above,  the reparametrization is 
$$
\gamma_2 = \left( \sqrt{1 - \| \widetilde \gamma_2\|^2}, \; \widetilde \gamma_2^{\;\top} \right)^\top,
$$ 
where now 
$\widetilde \gamma_2\subset \{ z \in \mathbb{R}^{d_W -1}: \|z\|\leq 1\}.$ The Jacobian matrix of this reparametrization using the normalization of $\gamma_2$ is the $d_W\times (d_W -1)$ matrix 
\begin{equation}\label{Jacobi2}
\mathbf{J}_2 (\gamma_2) = \frac{\partial  \gamma_2}{\partial \widetilde \gamma_2} =  \begin{pmatrix}
   -  \{1 - \| \widetilde \gamma_2\|^2\}^{-1/2} \; \widetilde \gamma_2^{\;\top}  \\
    \mathbf{I}_{d_W-1}
  \end{pmatrix}.
\end{equation}
Hereafter, when we refer to the true value of the finite-dimensional parameter, we implicitly assume that one of these two approaches for identifying $\gamma_2$ was chosen. 

\section{Unconditional moment estimating equations} \label{sec:UCME}
This section presents unconditional moment equations which permit parameter inference by using empirical likelihood. \color{black}  The way these equations are constructed will have two important consequences: blocking data is unnecessary and the nonparametric estimation of the infinite-dimensional parameter does not break the chi-squared limit of the ELR statistics. \color{black} For ease of explanation, we start by introducing an unconditional moment equation which is equivalent to the conditional moment equation of the PLSIM defined in \eqref{model:easy}. Then, we introduce an unconditional moment equation which is equivalent to the conditional moment equation of the CHPLSIM defined in \eqref{model:full}.

\subsection{Partially linear single-index model} \label{sec:PLSIMmoments}
For the PLSIM, it is quite standard \citep{zhu2006empirical} to consider the following unconditional \color{black}  moment \color{black} equation 
\begin{equation} \label{eq:equivalentmu}
\mathbb{E}[g_\mu(Z_i;\gamma,m_\gamma) \widetilde\nabla_\gamma g_\mu(Z_i;\gamma,m_\gamma)]=0,  
\end{equation}
where  $\gamma=(\gamma_1^\top,\gamma_2^\top)^\top \in  \mathbb R^{d_\gamma}$, $d_\gamma = d_1+d_W$, and 
$$
\widetilde\nabla_\gamma g_\mu(Z_i;\gamma,m_\gamma) = 
\mathbf{J}(\gamma) \nabla_\gamma g_\mu(Z_i;\gamma,m_\gamma) \in \mathbb R^{d_\gamma-1} ,
$$
with  $\mathbf{J}(\gamma) $ the $(d_\gamma-1)\times d_\gamma$ Jacobian matrix of the reparametrization chosen to guarantee the identification of the finite-dimensional  parameter and $\nabla_\gamma$ (resp. $\nabla_{\gamma_1}$) the column matrix-valued  operator of the first order partial derivatives with respect to the components of $\gamma \in \mathbb R^{d_\gamma}$ (resp. $\gamma_1\in  \mathbb R^{d_1}$). 
In our context, 
\begin{multline*}
\nabla_\gamma g_\mu(Z_i;\gamma, m_\gamma) = - \begin{bmatrix}
\nabla_{\gamma_1} l(X_i;\gamma_1)- \mathbb E [\nabla_{\gamma_1} l(X_i;\gamma_1) \mid W_i^\top \gamma_2]\\
m'(W_i^\top\gamma_{2}) \left(W_i - \mathbb E [W_i\mid W_i^\top \gamma_2]\right)\\
\end{bmatrix} \\ \textrm{and} \quad
\mathbf{J}(\gamma) 
=  \begin{pmatrix}
   \mathbf{I}_{d_1} & \mathbf{0}_{d_1\times (d_W-1) }\\
   \mathbf{0}_{ d_W\times d_1} & \mathbf{J}_2 (\gamma_{2})\\
      \end{pmatrix},
\end{multline*}
with $m^\prime(\cdot)$ the derivative of $m(\cdot)$ and $\mathbf{J}_2(\gamma_{2})$ the Jacobian matrix of the parametrization of $\gamma_2$, that is either the matrix defined in  \eqref{Jacobi1} or the one defined in  \eqref{Jacobi2}.

The following lemma proposes new unconditional moment equation  by introducing a positive weight function $\omega(V_i)$ in \eqref{eq:equivalentmu}. 
Showing  the equivalence between the conditional moment equation \eqref{model:easy} and our new unconditional moment equation, we deduce that the latter equation could be used for EL inference. 

\begin{lemma} \label{lem:equivalence}
Let  $\omega(\cdot)$ be a positive function of $V_i=(X_i^\top,W_i^\top)^\top$ and $H_\mu(\gamma)$ be the Hessian matrix of the map $\gamma\mapsto \mathbb E[\mathbb E^2[g_\mu(Z_i;\gamma,m_\gamma)\mid V_i,\mathcal F_{i-1}]$ $\omega(V_i)]$. 	Assume that conditions \eqref{model:easy} and \eqref{simeq1} hold true and  $H_\mu(\gamma)$ is definite positive. Then 
\begin{equation}
  \mathbb{E}[g_\mu(Z_i;\gamma,m_\gamma) \widetilde\nabla_\gamma g_\mu(Z_i;\gamma,m_\gamma)\omega(V_i)]=0 
   \quad \Leftrightarrow \quad  \gamma = \gamma_0. 
\end{equation}
\end{lemma}

\quad

For the PLSIM, we consider $\omega(V_i)= \eta^4_{\gamma,f}(W_i^\top\gamma_2)$ where  $\eta_{\gamma,f} (W_i^\top\gamma_2)$ is the density of the index $W_i^\top\gamma_2$, which is assumed to exist. This choice of the weights $\omega(V_i)$ allows to cancel all the 
terms $\eta_{\gamma,f} (W_i^\top\gamma_2)$ appearing in the denominators, and thus to keep them away from zero. Thus, for the control of the small values in the denominators,  it is no longer needed to assume that the density of the index is bounded away from zero.  This assumption, often imposed in the semiparametric literature, is quite unrealistic for bounded vectors $W_i$ and could not even hold when the  $W_i$'s are unbounded. Imposing bounded $W_i$ in a time series framework where $W_i$ could include lagged values of  $Y_i$ would be too restrictive.

Thus, we consider that the parameters are defined by the unconditional moment equations
\begin{equation}
\mathbb{E}[ \Psi(Z_i;\gamma,\eta_{\gamma})]=0 , \label{eq:uncondmomemtfacile}
\end{equation}
where $\Psi(Z_i;\gamma,\eta_{\gamma})=g_\mu(Z_i;\gamma,m_{\gamma}) \tilde\nabla_\gamma g_\mu(Z_i;\gamma,m_{\gamma})\eta^4_{\gamma,f}(W_i^\top\gamma_2)\in\mathbb R^{d_\gamma - 1}$. 
Thus, we have
\begin{multline}\label{def_Phi1}
\Psi(Z_i;\gamma,\eta_{\gamma})=\left(\{Y_i -  l(X_i;\gamma_1)\} \eta_{\gamma,f}(W_i^\top\gamma_2)  - \eta_{\gamma,m}(W_i^\top \gamma_2)\right)\\ \times  \mathbf{J}(\gamma) \begin{bmatrix}
\eta_{\gamma,f}^2(W_i^\top\gamma_2)  \left(\nabla_{\gamma_1} l(X_i;\gamma_1) \eta_{\gamma,f}(W_i^\top\gamma_2)  - \eta_{\gamma,X}(W_i^\top \gamma_2) \right) \\
\eta_{\gamma,m'}(W_i^\top\gamma_2)\left(W_i\eta_{\gamma,f}(W_i^\top\gamma_2)  - \eta_{\gamma,W}(W_i^\top\gamma_2)  \right)
\end{bmatrix},
\end{multline}
where the vector
$\eta_{\gamma}=(\eta_{\gamma,m},\eta_{\gamma,m'},\eta_{\gamma,X},\eta_{\gamma,W},\eta_{\gamma,f})^\top$ groups all the non-parametric elements and, using the stationarity of the process, is given for any $ t\in\mathbb R$ by 
\begin{align*}
\eta_{\gamma,m}(t)&=m_{\gamma}(t)\eta_{\gamma,f}(t)=\mathbb{E}[Y_i-l(X_i;\gamma_1) \mid W_i^\top\gamma_2=t]\eta_{\gamma,f}(t),\\ \eta_{\gamma,m'}(t) &=  \eta_{\gamma,f}^2(t)  \frac{\partial}{\partial t}  m_{\gamma}(t) = \eta_{\gamma,f}^2(t) \frac{\partial}{\partial t} \mathbb{E}[Y_i- l(X_i;\gamma_1) \mid W_i^\top\gamma_2=t], \\
\eta_{\gamma,X}(t)&=\mathbb E[\nabla_{\gamma_1} l(X_i;\gamma_1) \mid W_i^\top \gamma_2 = t]\eta_{\gamma,f}(t),\quad\\ \quad\eta_{\gamma,W}(t)&=\mathbb E[W_i \mid W_i^\top \gamma_2 = t]\eta_{\gamma,f}(t).
\end{align*} 

\subsection{Conditionally heteroscedastic partially linear single-index model}

For the  CHPLSIM we have to construct an unconditional moment equation to take into account the conditional variance condition in \eqref{eq:modelbase_c}. In this case, the finite-dimensional parameters are  $\theta=(\gamma^\top,\beta^\top)^\top \in \mathbb R^{d_\theta}$ with $d_\theta = d_\gamma+d_\beta$. Given the definition \eqref{g_sig}, we have 
$$
 \nabla_\beta g_\sigma(Z_i,Z_{i}^{\{r\}};\theta,m) = - \nabla_\beta \sigma^2(V_i,Z_{i}^{\{r\}};\beta) \in \mathbb R^{d_\beta}.
$$
The following lemma provides the unconditional moment equations for EL inference in CHPLSIM. The proof is similar to the proof of Lemma~\ref{lem:equivalence} and is thus omitted.

\begin{lemma} \label{lem:equivalence2}
Let  $\omega_1(\cdot)$ and $\omega_2(\cdot)$ be positive functions of $V_i$.
\color{black}  Let $H_\mu(\gamma)$ and $H_\sigma(\beta)$ be the Hessian matrices of the maps $\gamma\mapsto \mathbb E[\mathbb E^2[g_\mu(Z_i;\gamma,m_\gamma)\mid V_i,\mathcal F_{i-1}]\omega_1(V_i)]$ and   $\beta\mapsto \mathbb E[\mathbb E^2[g_\sigma(Z_i,Z_{i}^{\{r\}},\theta,m)\mid V_i,\mathcal F_{i-1}]\omega_2(V_i)].$ \color{black} 
Assume that conditions \eqref{model:full} and \eqref{simeq1} hold true and  $H_\mu(\gamma)$  and $H_\sigma(\beta)$ are definite positive. 
Then
\begin{align*}
\left\{ \begin{array}{rl}
\mathbb{E}[g_\mu(Z_i;\gamma,m_\gamma) \widetilde\nabla_\gamma g_\mu(Z_i;\gamma,m_\gamma)\omega_1(V_i)]=0\\
\mathbb{E}[g_\sigma(Z_i,Z_{i}^{\{r\}};\theta,m_\gamma) \nabla_\beta \sigma^2(V_i,Z_{i}^{\{r\}};\beta)\omega_2(V_i)]=0
\end{array}
\right.
\quad \Leftrightarrow \quad \theta = \theta_0.
\end{align*}
\end{lemma}

To cancel all the denominators induced by the non-parametric estimator, we take $\omega_1(V_i )= \eta^4_{\gamma,f}(W_i^\top\gamma_2)$ and $\omega_2(V_i )= \eta^2_{\gamma,f}(W_i^\top\gamma_2)$. 
Thus, we consider that the parameters are defined by the unconditional moment equations
\begin{equation}
\mathbb{E}[ \Psi(Z_i,Z_{i}^{\{r\}};\theta,\eta_{\gamma})]=0 , \label{eq:uncondmomemt}
\end{equation}
where  $\eta_{\gamma}$ is defined as in section \ref{sec:UCME} and $\Psi(Z_i,Z_{i}^{\{r\}};\theta,\eta_{\gamma})\in\mathbb R^{d_\theta-1 }$ with
\begin{equation}\label{def_Phi2}
\Psi(Z_i,Z_{i}^{\{r\}};\theta,\eta_{\gamma})= \begin{pmatrix}
g_\mu(Z_i;\gamma,m_{\gamma}) \widetilde\nabla_\gamma g_\mu(Z_i;\gamma,m_{\gamma})\eta_{\gamma,f}^4(W_i^\top\gamma_2) \\
g_\sigma(Z_i,Z_{i}^{\{r\}};\theta,m_\gamma ) \nabla_\beta \sigma^2(V_i,Z_{i}^{\{r\}};\beta)\eta_{\gamma,f}^2(W_i^\top\gamma_2)
\end{pmatrix}.
\end{equation}

\section{Parameter inference with weakly dependent data}\label{sec:EL}

\subsection{General framework of empirical likelihood} \label{sec:GF}
In the sequel, for EL inference in the CHPLSIM we use condition \eqref{eq:uncondmomemt}, while for EL inference in the PLSIM we use condition \eqref{eq:uncondmomemtfacile}. With a slight abuse of notation, in the sequel we use the notation $\Psi(Z_i,Z_{i}^{\{r\}};\theta,\eta_{\gamma})$, with some given integer $r\geq 0$, for both  PLSIM and CHPLSIM conditions. By definition, the case $r=0$ corresponds to the case where $\Psi(Z_i,Z_{i}^{\{r\}};\theta,\eta_{\gamma})$ does not depend on the lagged values of $Z_i$. This is the case for PLSIM, but this situation could also occur in CHPLSIM.

\color{black}  By construction, we have the following important property in the context of dependent data. 

\begin{lemma} \label{lem:mart_diff}
The estimating function $\Psi(\cdot,\cdot;\cdot,\cdot)$ satisfies the following property~:
\begin{equation}\label{eq:martdiff}
\forall i\neq j \quad \mathbb E \left[ \Psi(Z_i,Z_{i}^{\{r\}};\theta_0,\eta_0)\Psi(Z_j,Z_{j}^{\{r\}};\theta_0,\eta_0)^\top\right ] = 0.
\end{equation}
\end{lemma}
This result is a direct consequence of the fact that $ E \left[ \Psi(Z_i,Z_{i}^{\{r\}};\theta_0,\eta_0)\mid V_i, \mathcal{F}_{i-1}\right ] = 0$. As pointed out by a reviewer, 
this property  indicates that, using our estimating function, one can consistently estimate the so-called long-run covariance matrix of the vector-valued sequence $\Psi(Z_1,Z_{1}^{\{r\}};\theta_0,\eta_0),\ldots,\Psi(Z_n,Z_{n}^{\{r\}};\theta_0,\eta_0)$  by the standard sample covariance matrix. Therefore, blocking data is unnecessary in our framework, which is the one of a martingale difference sequence with respect to the filtration $\sigma(V_i, \mathcal{F}_{i-1})$. See also \citet{Kitamura97}, page 2092, and  \citet{CHANG2015283} page 287.

\color{black}

If $\eta_\gamma$ is given, the empirical likelihood, obtained with the unconditional moment conditions we propose for the (CH)PLSIM, is defined by
$$
L(\theta,\eta_\gamma) = \max_{\pi_1,\ldots,\pi_n}\prod_{i=1}^n \pi_i(\theta,\eta_\gamma),
$$
where $\sum_{i=1}^n \pi_i(\theta,\eta_\gamma)\Psi(Z_i,Z_{i}^{\{r\}};\theta,\eta_\gamma)=0$, $\pi_i(\theta,\eta_\gamma) \geq 0$, $\sum_{i=1}^n\pi_i(\theta,\eta_\gamma)=1$. Thus, we have
$$
\pi_i(\theta,\eta_\gamma) = \frac{1}{n} \frac{1}{1 + \lambda(\theta,\eta_\gamma)^\top \Psi(Z_i,Z_{i}^{\{r\}};\theta,\eta_\gamma)},
$$
where $\lambda(\theta,\eta_\gamma)\in\mathbb{R}^{d_1+d_W-1}$ are the Lagrange multipliers which permit to satisfy the empirical counterpart of the restriction \eqref{eq:uncondmomemt}, that is  $$\sum_{i=1}^n \pi_i(\theta,\eta_\gamma)\Psi(Z_i,Z_{i}^{\{r\}};\theta,\eta_\gamma)=0.$$
The empirical log-likelihood ratio is then defined by
$$\ell_n(\theta,\eta_\gamma) = \sum_{i=1}^n \ln (1 + \lambda(\theta,\eta_\gamma)^\top \Psi(Z_i,Z_{i}^{\{r\}};\theta,\eta_\gamma)).$$
As the infinite-dimensional parameter $\eta_\gamma$ is unknown, nonparametric estimation using kernel smoothing is used instead.
Thus, we propose to consider
\begin{equation}\label{eq:ELLR}
\ell_n(\theta,\widehat \eta_\gamma) = \sum_{i=1}^n \ln \left(1 + \lambda(\theta,\widehat \eta_\gamma)^\top \Psi(Z_i,Z_{i}^{\{r\}};\theta,\widehat \eta_\gamma)\right),
\end{equation}
where 
\begin{equation} \label{eq:eta_chap}
\widehat \eta_\gamma = (\widehat \eta_{\gamma,m}, \widehat \eta_{\gamma,m'},\widehat \eta_{\gamma,X},\widehat \eta_{\gamma,W}, \widehat \eta_{\gamma,f})^\top,
\end{equation}
with, for any $t \in \mathbb{R}$,
\begin{align*}
\widehat\eta_{\gamma,f}(t) &= \frac{1}{nh} \sum_{i=1}^n K\left(\frac{W_i^\top \gamma_2 - t}{h}\right),\\
\widehat\eta_{\gamma,m}(t) &= \frac{1}{nh} \sum_{i=1}^n \{Y_i - l(X_i;\gamma_1)\} K\left(\frac{W_i^\top \gamma_2 - t}{h}\right),\\
\widehat\eta_{\gamma,X}(t) &= \frac{1}{nh} \sum_{i=1}^n \nabla_{\gamma_1} l(X_i;\gamma_1) K\left(\frac{W_i^\top \gamma_2 - t}{h}\right),\\
\widehat\eta_{\gamma,W}(t) &= \frac{1}{nh} \sum_{i=1}^n W_i K\left(\frac{W_i^\top \gamma_2 - t}{h}\right),
\end{align*}
and
\begin{multline*}
\widehat \eta_{\gamma,m'}(t) = \frac{1}{nh^2} \left[  \widehat\eta_{\gamma,f}(t)\sum_{i=1}^n \{Y_i - l(X_i;\gamma_1)\} K'\left(\frac{W_i^\top \gamma_2 - t}{h}\right)\right. \\ - \left. \widehat\eta_{\gamma,m}(t) \sum_{i=1}^n K'\left(\frac{W_i^\top \gamma_2 - t}{h}\right)  \right],
\end{multline*}
$K^\prime (\cdot)$ is the derivative of the univariate kernel $K(\cdot)$ and  $h$ is the bandwidth. 

\subsection{Assumptions}
We will consider weakly dependent data which satisfy strong mixing conditions. We refer the reader to the book of \citet{Rio2000} and to the  survey of \citet{Bradley05} for the basic properties as well as the asymptotic behavior of weakly dependent processes. We will focus our attention on $\alpha$-mixing sequences. We use the following measure of dependence between two $\sigma$-fields $\mathcal{A}$ and $\mathcal{B}$:
$$\alpha(\mathcal{A}, \mathcal{B}) = \sup_{A \in \mathcal{A}, B \in \mathcal{B}} \left| \mathbb{P}(A \cap B) - \mathbb{P}(A)\mathbb{P}(B)\right|.$$
We recall that a sequence $(Z_i)_{i\in\mathbb Z}$ is said to be $\alpha$-mixing (or strongly mixing) if $\alpha_m = \sup_{j \in \mathbb{Z}} \alpha ( \mathcal{F}_{-\infty}^j, \mathcal{F}_{j+m}^{\infty})$ goes to zero as $m$ tends to infinity, where for any $-\infty \leq j \leq l \leq \infty$, $\mathcal{F}_j^l = \sigma(Z_i, j \leq i \leq l)$. Let
$$U_i = ( l(X_i;\gamma_{0,1}) ,  \nabla_{\gamma_1} l(X_i;\gamma_{0,1}) ^\top, W_i^\top,\varepsilon_i)^\top.$$

\begin{Ass}\label{ass:ass1}

\begin{enumerate}[label=(\roman*)]
\item \label{ass:process} The process $(Z_i)_{i\in\mathbb Z}$,  $Z_i=(X_i^\top,W_i^\top,\varepsilon_i)^\top\in\mathbb{R}^{d_X}\times\mathbb{R}^{d_W}\times\mathbb{R}$, is strictly stationary and strongly mixing with mixing coefficients $\alpha_m$ satisfying
\begin{equation}\label{eq:def_alpha_ew} 
\alpha_m = O(m^{-\xi}) 
\textrm{ with } 
 \xi > 10 \frac{s}{s-3}
\end{equation}
for some
 $s>6$ such that 
\begin{equation}\label{eq:ex_mom}
\sup_{\|c\|=1} \mathbb E [|U_i^\top c|^{s}]<\infty.
\end{equation}

\item \label{ass:bounds}  The marginal density of the index $\eta_{\gamma_0,f} (\cdot)$ of the index $W_i^\top\gamma_{0,2}$
is such that 
$
\sup_{t\in\mathbb R } \eta_{\gamma_0,f} (t) <\infty $
and
\begin{equation}\label{mom_q}
\sup_{\|c\|=1}
\sup_{t\in\mathbb R} \; \mathbb E  [ |U_i^\top c| \{ |t|+|U_i^\top c|^{s-1} \} \mid W_i^\top\gamma_{0,2}=t]  \eta_{\gamma_0,f} (t) <\infty. 
\end{equation}
Moreover, there is some $j^\star<\infty$ such that, for all $j \geq j^\star$,
$$\sup_{(t,t')\in \mathbb R^2 } \mathbb E [|U_0^\top U_j| \mid W_0^\top\gamma_{0,2}=t, W_j^\top\gamma_{0,2}=t']f_{W_0^\top\gamma_{0,2},W_j^\top\gamma_{0,2}}(t,t')<\infty,
$$
where $f_{W_0^\top\gamma_{0,2},W_j^\top\gamma_{0,2}}(\cdot)$ is the joint density of $ W_0^\top\gamma_{0,2}$ and $W_j^\top\gamma_{0,2}$.
\item \label{ass:funcdiff} The second partial derivatives of $\mathbb E[ \nabla_{\gamma_1} l(X_i;\gamma_1) \mid W_i^\top\gamma_{0,2}=\cdot]$,  $\mathbb E[W_i \mid W_i^\top\gamma_{0,2}=\cdot]\eta_{\gamma_0,f} (\cdot)$ and $\eta_{\gamma_0,f} (\cdot)$, as well as the third derivatives of $m_0(\cdot)$, are uniformly continuous and bounded. Moreover, the first derivative of $m_0(\cdot)$ is bounded, and the vector $\nabla_\beta \sigma^2(V_i,Z_{i}^{\{r\}};\beta_0) $ is also bounded.  

\end{enumerate}
\end{Ass}

 \begin{Ass}\label{variance_s}
The matrix $$\Sigma= \mathbb{E}\left[\Psi(Z_i,Z_{i}^{\{r\}};\theta_0,\eta_0)\Psi(Z_i,Z_{i}^{\{r\}};\theta_0,\eta_0)^\top\right]$$ is positive definite. 
\end{Ass}

\begin{Ass}\label{ass:ass2}
The Hessian matrix $H_\mu(\gamma)$, defined with the weight $\omega_1(V_i )= \eta^4_{\gamma,f}(W_i^\top\gamma_2)$, is  positive definite. 
Moreover, when the model is defined by \eqref{eq:modelbase}-\eqref{eq:modelbase_c}, both the Hessian matrices $H_\mu(\gamma)$ and $H_\sigma(\beta)$  with their corresponding weights  $\omega_1(V_i )= \eta^4_{\gamma,f}(W_i^\top\gamma_2)$ and $\omega_2(V_i )= \eta^2_{\gamma,f}(W_i^\top\gamma_2)$ are positive definite.
\end{Ass}

\begin{Ass}\label{ass:ass3}
The bandwidth $h$ used for the non-parametric part of the estimation is such that   $nh^3/\ln n \to\infty$  and $nh^8\to 0$. The univariate kernel $K$ is symmetric, bounded, integrable,
such that $\int_{\mathbb{R}}t^2 \{| K(t)|+|t K^\prime(t)|\} dt <\infty$ and $\int_{\mathbb{R}}t^2 K(t)dt \neq 0.$ The Fourier Transform of $K$, denoted by $\mathcal F [K]$, satisfies the condition $\sup_{t\in\mathbb R} |t|^{c_K}|\mathcal F [K](t)|<\infty$ for some $c_K>3$. 
 Moreover,  $t\mapsto |t|^{\color{black} s/2\color{black}} \{K(t) + K^\prime (t)\}$ is bounded on $\mathbb{R}$, where $s$ is defined by Assumption \ref{ass:ass1}\ref{ass:process}.  
\end{Ass}

\color{black}

Assumption~\ref{ass:ass1} guarantees suitable rates of uniform convergence for the kernel estimators of the infinite-dimensional parameters gathered in the vector $\eta_\gamma$. More precisely, they imply the conditions used in Theorem 4 of   \citet{hansen2008uniform}, 
with $q=d=1$. We also use the condition on $\xi$ to apply Davydov's inequality and show that the effect of the nonparametric estimation is negligible and does not alter the pivotalness of the empirical log-likelihood ratio statistic. Due to this purpose, some conditions in Assumption~\ref{ass:ass1} are more restrictive than in Theorem 4 of   \citet{hansen2008uniform}. Condition~\eqref{eq:def_alpha_ew}  reveals a link between the existence of some moments of order $s$ and the strength of the dependency given by the coefficient $\xi$. The more moments for $U_i$ exist, the stronger the time dependency can be. In particular, 
if $U_i$ has finite moments of any order, then $s=\infty$ and thus $\xi$ could be larger but arbitrarily close to $10$.
\color{black}  
There is a wide literature on the mixing properties for time series. The most popular technique for proving this property relies on rewriting the process as a Markov chain and showing the geometrically decay of the mixing coefficients $\alpha_m$. For example, ARMA processes were treated in  \citet{Mok88}, while some non-linear time series were investigated by \citet{Mok90},  \citet{tjostheim_1990}, \citet{MasryTjo95},  and more recently by \citet{LuJiang}, \citet{Liebscher05}, \citet{MeitzSai10}. See also the references therein. Another technique has been developed in \citet{FryzRao}. They show mixing properties for time-varying ARCH and ARCH$(\infty)$ processes by computing explicit bounds for the mixing coefficients using the density function of the processes. Their method could possibly be applied in our context to obtain the conditions of Assumption~\ref{ass:ass1}.
\color{black} Assumption~\ref{variance_s} guarantees a non-degenerate limit distribution in the CLT  for the sample mean of the 
$\Psi(Z_i,Z_{i}^{\{r\}};\gamma_0,\eta_0)$'s. 
Assumption~\ref{ass:ass2} is used to prove Lemma~\ref{lem:equivalence} and Lemma~\ref{lem:equivalence2}. 
Concerning the bandwidth conditions,  one could of course use different bandwidths for the different nonparametric estimators involved. For readability and practical simplicity, we propose a same bandwidth $h$. Moreover, Assumption \ref{ass:ass3} allows one to use, for instance the Gaussian kernel.

\subsection{Wilks' Theorem}\label{subsec:wilks}
When the infinite-dimensional parameters $\eta_{\gamma}$ are given and the observations are independent,  Theorem~2 of \citet{qin1994Annals} guarantees  that the empirical log-likelihood ratio (ELR) statistic $ 2\ell_n(\theta_0,\eta_0)$ converges in distribution to a $\mathcal{X}_{d_\theta-1}^2$ as $n\to\infty$ (where $d_\theta$ is the dimension of the model parameters). The following theorem states that, under suitable conditions, the chi-squared limit in law  is preserved for the ELR defined with our moment conditions for the (CH)PLSIM, with dependent data and estimated $\eta_{\gamma}$. Let us define the ELR statistic 
$$
W(\theta_0) =  2\ell_n(\theta_0,\widehat \eta_{\gamma_0}),
$$
where $\ell_n$ and $\widehat \eta_{\gamma_0}$ are respectively given by \eqref{eq:ELLR} and \eqref{eq:eta_chap}. Let $d_\theta=d_\gamma$ for the PLSIM and $d_\theta=d_\gamma + d_\beta$ for the CHPLSIM.  In the following $\xrightarrow{d} $ denotes the convergence in distribution.

\begin{theorem} \label{thm:chi2etaunknown}
Consider that Assumptions~\ref{ass:ass1}, \textcolor{black} {\ref{variance_s}}, \ref{ass:ass2} and \ref{ass:ass3} hold true. 
Moreover, condition \eqref{simeq1} is satisfied, as well as condition \eqref{model:full} in the case of PLSIM or condition \eqref{model:easy} in the case of CHPLSIM. Then, $W(\theta_0)  \xrightarrow{d} \mathcal{X}_{d_\theta-1}^2$ as $n$ tends to infinity.
\end{theorem}

For the proof, we use a central limit theorem for mixing processes implying that ${n}^{-1/2} \sum_{i=1}^n \Psi(Z_i,Z_{i}^{\{r\}};\theta_0, \eta_0)$ converges in distribution to a multivariate centered normal distribution, to deal with the dependency between observations. 
Moreover, the behavior of the Lagrange multipliers has to be carefully investigated. 
However, the major difficulty in the proof is to show $\ell_n(\theta_0,\widehat \eta_{\gamma_0})-\ell_n(\theta_0, \eta_{0}) = o_{\mathbb P} (1)$, that is to show that the nonparametric estimation of the nuisance infinite-dimensional parameters does not break the pivotalness of the ELR statistic. \color{black}  This negligibility requirement is a well-known issue, see Remark 2.3 in \citet{hjort2009}. See also 
\citet{chang2013, chang2016, chang2020} for a related discussion in the context of  high-dimension empirical likelihood inference. However, this type of negligibility, obtained under mild technical conditions,  seems to be a new result in the context of semiparametric regression models with weakly dependent data.  \color{black}  It is obtained using arguments  based on Inverse Fourier Transform and Davydov's inequality in Theorem A.6 of  \citep{Hall_Heyde}. \color{black}   It is also worthwhile to notice that, in order to preserve the chi-squared limit for $W(\theta_0)$, we do not need to follow the general two-step procedure proposed by  \citet{bravo2020} and replace $\Psi(\cdot,\cdot;\cdot,\cdot)$ by some   estimated influence function. The reason is given by the gradient $\widetilde\nabla_\gamma g_\mu(Z_i;\gamma,m_\gamma) $ which has the key property $\mathbb E [\widetilde\nabla_\gamma g_\mu(Z_i;\gamma_0,m_{\gamma_0}) \mid W_i^\top \gamma_{0,2}]=0$ a.s.  \color{black}

\section{Numerical experiments} \label{sec:num}
\subsection{Simulations}
We generated data from model \eqref{eq:modelbase}-\eqref{eq:modelbase_c} with $\varepsilon_i=\sigma(V_i,Z_i^{\{r\}};\beta)\zeta_i$ and $$\sigma^2(V_i,Z_i^{\{r\}};\beta)=\beta_1 + \beta_2 Y_{i-1}^2,$$
where the $\zeta_i$ are independently drawn from a distribution such that $\mathbb{E}(\zeta_i)=0$ and $\text{Var}(\zeta_i)=1$. That means, we allow for conditional heteroscedasticity in the mean regression error term.     The covariates $X_i=(Y_{i-1},Y_{i-2})^\top$ are two lagged values of the target variable $Y_i$ and the covariates $W_i=(W_{i1},W_{i2},W_{i3})^\top$ are generated from a multivariate Gaussian distribution with  mean  $W_{i-1}/4$  and covariance \textcolor{black} {matrix $S$ defined by} $\textrm{cov}(W_{ik},W_{i\ell})=0.5^{|k-\ell|}$. \textcolor{black} {Thus, the marginal distribution of the index $W_i^\top \gamma_2$ is a centered Gaussian distribution with variance $(16/15)\gamma_2^\top S \gamma_2$}. We set 
\begin{equation}\label{simu_design1}
\ell(X_i;\gamma_1)=\gamma_{11}Y_{i-1} + \gamma_{12}Y_{i-2} \quad \text{ and }  \quad m(u)= \frac{3}{4}\sin^2(u\pi) ,
\end{equation}
with $\gamma_1=(0.1,0)^\top$, $\gamma_2=(1,1,1)^\top$ and $\beta=(0.9,0.1)^\top$.

Hypothesis testing is based on  Wilks' Theorem   in Section~\ref{subsec:wilks} (results related to this method are named \emph{estim}), along with the unfeasible EL approach  that \textcolor{black} {uses the true density of the index} and that previously learns the nonparametric estimators on a sample of size $10^4$ (this case mimics the situation where $m$,  $m'$ \color{black}  and the conditional expectations involved in the definition of $\eta_\gamma$ \color{black} are known; results related to this method are named \emph{ref}). The nonparametric elements are estimated by the Nadaraya-Watson method with Gaussian kernel and bandwidth $h=C^{-1} n^{-1/5}$ where $C$ is the standard deviation of the index. In the experiments, we consider four sample sizes (100, 500, 2000 and 5000) and three distributions for $\zeta_i$:  a standard Gaussian distribution (\emph{Gaussian}), an uniform distribution on $[-\sqrt{3},\sqrt{3}]$ (\emph{uniform}) and a mixture of Gaussian distributions (\emph{mixture}) $pN(m_1,v_1^2)+(1-p)N(m_2,v_2^2)$, with $p=0.5$, $m_2 = - m_1= 1/\sqrt{6}$, $v_1^2 = 1/6$,  $v_2^2 = 3/2$. For each scenario, we generated 5000 data sets.

First, we want to  test   the order for the lagged values of $Y_i$ in the parametric function $\ell$.   For this purpose,  we use the PLSIM 
 and we consider the following tests: 
\begin{itemize}
\item \emph{Test Lag(1)} which corresponds to the true order equal to 1, 
and which is defined by $H_0: \gamma_1=(0.1,0)^\top$ and $\gamma_2=(1,1,1)^\top$;
\item \emph{Test Lag(0)} which neglects the lagged values of $Y_i$ in the linear part and which is defined by $H_0: \gamma_1=(0,0)^\top$ and $\gamma_2=(1,1,1)^\top$;
\item \emph{Test Lag(2)} which overestimates the order for the lagged values of $Y_i$ and which is defined by $H_0: \gamma_1=(0.1,0.1)^\top$ and $\gamma_2=(1,1,1)^\top$.
\end{itemize}
The empirical probabilities of rejection are presented in Table~\ref{resPLSIM} for a nominal level of $0.05$.   A first, not surprising, conclusion: EL inference in such flexible nonlinear models, with dependent data, requires sufficiently large sample sizes. The results with $n=100$ are quite poor even when $m(\cdot)$ is given, that is in a purely parametric setup.  Next, we notice that for the three distributions of the noise, our EL inference approach allows to identify the correct order for the lagged values when the sample size is sufficiently large. 
Indeed, only \emph{Test Lag(1)} has an asymptotic empirical probability of rejection converging to the nominal level 0.05 while the other tests have a probability of rejection converging to one. Moreover, the differences between the unfeasible EL approach (\emph{ref.} columns) and our approach (\emph{estim.} columns) become quickly negligible. This result was expected because the statistics of both methods converge to the same chi-squared distribution.

\begin{table}
\caption{\label{resPLSIM} Empirical probabilities of rejection obtained   from 5000 replications using the PLSIM  for testing the order for the lagged values of $Y_i$ in the parametric part $\ell(\cdot;\gamma_1)$ in \eqref{simu_design1}.} 
\centering
\color{black} 
\begin{tabular}{*{10}{c}}
\hline
Test & $\zeta_i$ & \multicolumn{2}{c}{$n=100$} & \multicolumn{2}{c}{$n=500$}  & \multicolumn{2}{c}{$n=1000$}  & \multicolumn{2}{c}{$n=2000$}   \\ 
 &  & ref. & estim.& ref. & estim.& ref. & estim.& ref. & estim.\\ 
\hline 
 Lag(1) & Gaussian & 0.167 & 0.214 & 0.066 & 0.075 & 0.054 & 0.054 & 0.055 & 0.055 \\ 
    & uniform & 0.125 & 0.185 & 0.058 & 0.074 & 0.058 & 0.056 & 0.053 & 0.050 \\ 
    & mixture & 0.196 & 0.229 & 0.080 & 0.094 & 0.063 & 0.060 & 0.053 & 0.051 \\ 
   Lag(0) & Gaussian & 0.208 & 0.243 & 0.254 & 0.231 & 0.705 & 0.665 & 0.983 & 0.980 \\ 
    & uniform & 0.160 & 0.204 & 0.215 & 0.207 & 0.742 & 0.718 & 0.991 & 0.990 \\ 
    & mixture & 0.236 & 0.263 & 0.241 & 0.228 & 0.647 & 0.619 & 0.969 & 0.965 \\ 
   Lag(2) & Gaussian & 0.216 & 0.270 & 0.266 & 0.268 & 0.783 & 0.760 & 0.996 & 0.995 \\ 
    & uniform & 0.164 & 0.227 & 0.241 & 0.243 & 0.775 & 0.769 & 0.996 & 0.997 \\ 
   & mixture & 0.263 & 0.301 & 0.308 & 0.299 & 0.773 & 0.725 & 0.993 & 0.990 \\ 
\hline 
\end{tabular} 
\color{black}
\end{table}

We now investigate the order for the lagged values of $Y_i$ in the conditional mean and variance of the noise. Thus, we use the CHPLSIM and we consider the following tests:
\begin{itemize}
\item \emph{Test Lag(1)-CH(1)} which corresponds to the true values of the conditional mean and variance   and which is defined by $H_0: \gamma_1=(0.1,0)^\top$, $\gamma_2=(1,1,1)^\top$ and $\beta=(0.9,0.1)^\top$;
\item \emph{Test Lag(0)-CH(1)} which neglects the lagged values of $Y_i$ in the conditional mean  and which is defined by $H_0: \gamma_1=(0,0)^\top$, $\gamma_2=(1,1,1)^\top$ and $\beta=(0.9,0.1)^\top$;
\item \emph{Test Lag(2)-CH(1)} which overestimates the order of the lagged values of $Y_i$ in the conditional mean and which is defined by $H_0: \gamma_1=(0.1,0.1)^\top$, $\gamma_2=(1,1,1)^\top$ and $\beta=(0.9,0.1)^\top$;
\item \emph{Test Lag(1)-CH(0)} which corresponds to the true value of the conditional mean but neglects the lagged value of $Y_i$ in the conditional variance and which is defined by $H_0: \gamma_1=(0.1,0)^\top$, $\gamma_2=(1,1,1)^\top$ and $\beta=(0.9,0)^\top$.
\end{itemize}

The empirical probabilities of rejection are presented in Table~\ref{resCHPLSIM} for a nominal level of $0.05$. Again, the true order of the lagged values is detected by the procedure and the differences between the unfeasible EL approach  and our approach  become quickly negligible. As expected given that the model is more complex, the rate of convergence to the nominal level is slower than for the tests on the PLSIM. 
  However, our procedure allows the conditional heteroscedasticity of the noise to be detected, and meanwhile it identifies the correct order for the lags of $Y_i$ in the mean equation.  

\begin{table}
\color{black} 
\caption{\label{resCHPLSIM} Empirical probabilities of rejection 
obtained  from 5000 replications using  the
CHPLSIM for testing the order of the lagged values of $Y_i$ in the conditional mean and variance.}
\centering
\begin{tabular}{*{10}{c}}
\hline
Test & $\zeta_i$ & \multicolumn{2}{c}{$n=100$} & \multicolumn{2}{c}{$n=500$}  & \multicolumn{2}{c}{$n=1000$}  & \multicolumn{2}{c}{$n=2000$}   \\ 
 &  & ref. & estim.& ref. & estim.& ref. & estim.& ref. & estim.\\ 
\hline 
 Lag(1) & Gaussian & 0.292 & 0.388 & 0.105 & 0.111 & 0.068 & 0.074 & 0.074 & 0.069 \\ 
  CH(1) & uniform & 0.167 & 0.277 & 0.070 & 0.077 & 0.067 & 0.072 & 0.084 & 0.072 \\ 
    & mixture & 0.392 & 0.461 & 0.151 & 0.170 & 0.090 & 0.098 & 0.079 & 0.078 \\ 
   Lag(0) & Gaussian & 0.331 & 0.406 & 0.260 & 0.249 & 0.669 & 0.641 & 0.978 & 0.972 \\ 
   CH(1) & uniform & 0.197 & 0.291 & 0.198 & 0.190 & 0.684 & 0.675 & 0.986 & 0.983 \\ 
    & mixture & 0.446 & 0.493 & 0.333 & 0.327 & 0.653 & 0.637 & 0.963 & 0.958 \\ 
   Lag(2) & Gaussian & 0.337 & 0.426 & 0.277 & 0.287 & 0.743 & 0.727 & 0.993 & 0.992 \\ 
   CH(1) & uniform & 0.205 & 0.304 & 0.219 & 0.227 & 0.724 & 0.728 & 0.993 & 0.993 \\ 
    & mixture & 0.438 & 0.511 & 0.359 & 0.352 & 0.738 & 0.704 & 0.990 & 0.985 \\ 
   Lag(1) & Gaussian & 0.289 & 0.332 & 0.533 & 0.523 & 0.985 & 0.986 & 1.000 & 1.000 \\ 
   CH(0) & uniform & 0.283 & 0.294 & 0.777 & 0.748 & 1.000 & 1.000 & 1.000 & 1.000 \\ 
    & mixture & 0.343 & 0.392 & 0.489 & 0.499 & 0.970 & 0.970 & 1.000 & 1.000 \\ 
  \hline
\end{tabular} 
\color{black}
\end{table}

\subsection{Real data analysis}
We analyze the data set containing weather (temperature, dew point temperature, relative humidity) and pollution data (PM10 and ozone) for the city of Chicago in the period 1987-2000 from the National Morbidity, Mortality and Air Pollution Study. The analyzed data is freely available in the R package \emph{dlnm} \citep{dlnm}. 
\citet{Lian_Carroll_2015} considered the same data set under the assumption of i.i.d. observations. 

We use the (CH)PLSIM with  a linear function in the parametric part  to predict  daily mean ozone level ($\widetilde{o3_i}$). For this purpose we use previous daily values of mean ozone level and four other predictors, that are the daily relative humidity ($\widetilde{\textit{rhum}_i}$), the daily mean temperature (in Celsius degrees) $\widetilde{\textit{temp}}_i$, the daily dew point temperature $\widetilde{\textit{dptp}}_i$  and the daily PM10-level $\widetilde{\textit{pm10}}_i$. 
 The first step of our analysis was to remove seasonality for each variable we considered in the models. To remove seasonality, we used the function \emph{seasadj} of the R package \emph{forecast} on the data of from  year 1994 to year 1997. Thus, we obtain the series \emph{o3$_i$}, \emph{rhum$_i$}, \emph{temp$_i$}, \emph{dptp$_i$} and \emph{pm10$_i$} by removing the seasonnality of the series $\widetilde{\textit{o3}}_i$, $\widetilde{\textit{rhum}}_i$, $\widetilde{\textit{temp}}_i$, $\widetilde{\textit{dptp}}_i$ and $\widetilde{\textit{pm10}}_i$. Note that the series  \emph{temp$_i$}, \emph{dptp$_i$} and \emph{pm10$_i$}  have been scaled to facilitate the interpretation $\gamma_{12}$. Figures~\ref{fig:ozone}-\ref{fig:pm10} provided in the Section~\ref{supp:appli} of the Supplementary Material present the original series and the series obtained by removing the seasonality. Thus, all the variables we refer hereafter in this section are deseasonalized. In this application the observations clearly have a time dependency. We split the sample into a learning sample (composed of the observations of years 1994 and 1995) and a testing sample (composed of the observations of years 1996 and 1997). After removing the seasonality, the autocorrelations of $o3$  for the learning and testing samples are 0.469 ($p-$value 0.000)  and 0.450 ($p-$value 0.000), respectively; Note that all the covariates have significant autocorrelations (all the $p-$values are 0.000, see Table~\ref{tab:autocorrelations} in the Section~\ref{supp:appli} of the Supplementary Material).

The covariates included in the linear part are the mean relative humidity (\emph{rhum$_i$}) and the mean ozone level computed on the three previous days ($o3_{i-1}$, $o3_{i-2}$, $o3_{i-3}$).  The covariates included in the nonparametric part of the conditional mean are \emph{temp$_i$},  \emph{dptp$_i$} and \emph{pm10$_i$}. The eigenvalues of the covariance matrix computed on the three variables used in the nonparametric part are 1.995, 0.901 and 0.168 for the data of learning sample, and 1.989, 0.758 and 0.139 for the data of testing sample.

Thus, the equation of the PLSIM is 
\begin{multline}\label{m1_app}
o3_i = \gamma_{11}rhum_i + \gamma_{12}o3_{i-1} + \gamma_{13}o3_{i-2} + \gamma_{14}o3_{i-3}\\ + m(\gamma_{21}  temp_i + \gamma_{22}  dptp_i + \gamma_{23}  pm10_i) + \varepsilon_i.
\end{multline}
We estimate the parameters of the models,  on the testing sample,  by minimizing the least squares using kernel smoothing (with Gaussian kernel and bandwidth $n^{-1/5}$). 
Hypothesis testing is conducted on the testing sample. We begin by investigating the order $H$ for the lagged values of the ozone measures to be included in the linear part of the conditional mean. Using PLSIM, we define different models, called $Lag(H)$ (with $H=0,1,2$ or 3), where  only $H$ lagged values of the mean ozone levels are included in the linear part (meaning  the coefficients  related to the other previous days is zero). The results  for  different orders $H$  presented in Table~\ref{resAR} show that the time dependency cannot be neglected for analyzing these data. It is relevant to include lagged values of the mean ozone level variable to build  its daily prediction. 

\begin{table}
\caption{\label{resAR} Estimators of the parameters obtained by the PLSIM, on the learning sample, with different orders of lagged values, and $p-$values obtained by testing these values on the testing sample for the ‘National morbidity and mortality air pollution study’ example.} 
\centering
\begin{tabular}{cccccc}
\hline
& & Lag(0)& Lag(1) & Lag(2) & Lag(3)\\
\hline
$\hat\gamma_1$ &\emph{rhum$_i$} & -0.122 & -0.157 & -0.154 & -0.154 \\
  & $o3_{(i-1)}$ & 0.000 & 0.412 & 0.459 & 0.461 \\  
  & $o3_{(i-2)}$ & 0.000 & 0.000 & -0.102 & -0.116 \\
  & $o3_{(i-3)}$ & 0.000 & 0.000 & 0.000 & 0.025 \\ 
$\hat\gamma_2$ &\emph{temp$_i$}& 0.976 & 0.937 & 0.941 & 0.939 \\ 
 & \emph{dptp$_i$} &  -0.215 & 0.343 & 0.332 & 0.336 \\ 
 & \emph{pm10$_i$} & 0.043 & 0.062 & 0.066 & 0.073 \\ 
\hline
&$p-$value &  0.000 & 0.001 & 0.107 & 0.044 \\ 
\hline 
\end{tabular} 
\end{table}

The autocorrelation of the residuals, obtained with the $Lag(2)$ setup, on the testing sample, has a value of $0.035$ ($p-$value 0.346). This suggests that $H=2$ is a reasonable choice. Figure~\ref{fig:density} and Figure~\ref{fig:m}, given in  Section~\ref{supp:appli} of the Supplementary Material, present the estimated density of the index and the estimated function $\hat m(\cdot)$, obtained with the $Lag(2)$ setup.

We also calculated the autocorrelation of the squared of the residuals, obtained with the $Lag(2)$ setup, and we obtain the value     0.095  ($p-$value 0.010).   This suggests  to also investigate the conditional heteroscedasticity of the noise using the CHPLSIM  with the $Lag(2)$ setup. For the conditional variance equation  we consider 
\begin{equation}\label{ch(1)}
\mathbb E (\varepsilon_i^2 \mid rhum_i, temp_i , dptp_i, pm10_i,\mathcal F_{i-1}) =\beta_1 + \beta_2 \ln\big(\max(o3_{i-1}^2,1)\big). 
\end{equation}

To estimate the parameters of the conditional variance, we use again the learning sample. 
The estimators for the CHPLSIM with conditional variance as in \eqref{ch(1)}  are $\hat\beta_1=1.553$ and $\hat\beta_2=3.786$. If we consider  constant conditional,  we obtain  $\tilde\beta_1=23.816$. The $p-$value obtained by testing the values  $\beta_1=\hat\beta_1$ and $\beta_2 = \hat\beta_2$ in \eqref{ch(1)} on the testing sample is 0.100. Meanwhile, the $p-$value obtained by testing the values $\beta_1=\tilde\beta_1$ and $\beta_2 = 0$ is 0.020. Thus, we conclude to a non constant conditional variance for the error term in \eqref{m1_app}. This effect should be considered to build forecast confidence intervals.

\section{Discussion and conclusion}\label{sec:concl}

We propose EL inference in a semiparametric mean regression model with strongly mixing data. Our model could include an additional condition on the second order conditional moment of the error term. The regression function has a partially linear single-index form,  while for the conditional variance we consider a parametric function. This function could depend on the past values of the observed variables, but it cannot depend directly on the regression error term. A parametric function of the past error terms would break the asymptotic pivotal distribution of the empirical log-likelihood ratio. See \citet{hjort2009} for a description of this common  phenomenon in semiparametric models.

We prove Wilks' Theorem under mild technical conditions, in particular without using any trimming and allowing for unbounded series. To obtain this result, first we rewrite the regression model under the form of a fixed number of suitable unconditional moment conditions. These moment conditions include infinite dimensional nuisance parameters estimated by kernel smoothing. Then, we show that estimating the nuisance parameters does not break the asymptotic pivotality of the empirical log-likelihood ratio which behaves asymptotically as if the nuisance parameters were given. Our theoretical result opens the door of the EL inference approach to new applications in nonlinear time series models. We illustrate our result by several simulation experiments and an application to air pollution where assuming time dependency seems reasonable, a fact confirmed by the data. 

The models proposed in this paper have several straightforward extensions. First, the variable $Y_i$ could be allowed to be measured with some error. For instance, $Y_i$ could be a function of the error term in a parametric model for some time series $(R_i)$, such as an $AR(1)$ model  $R_i = \rho R_{i-1} + u_i$. Taking $Y_i = u_i^2$, \eqref{eq:modelbase} could be used for inference on the conditional variance of $(u_i)$, while \eqref{eq:modelbase_c} could serve to test  the value of the kurtosis. This example that could be of interest for financial series is detailed in the Supplement.   

Another easy extension is to consider more general conditions than \eqref{eq:modelbase_c}. Our theoretical arguments apply with practically no change if  \eqref{eq:modelbase_c} is replaced by one or several conditions like 
$
  \mathbb E[T(\varepsilon_i) \mid V_i,\mathcal F_{i-1}]=\nu (V_i,Z_{i}^{\{r\}};\beta),
$
where the $T(\cdot)$'s are some given twice continuously differentiable functions such that  $\mathbb E[T^\prime (\varepsilon_i) \mid V_i,\mathcal F_{i-1}]=0$ a.s., and $\nu(\cdot,\cdot;\cdot)$ is given parametric function.
 For instance, taking $T(y)= y^4,$ we could include a  fourth order conditional moment equation in the model, provided $\mathbb E[\varepsilon_i^3 \mid V_i,\mathcal F_{i-1}]=0$ a.s. Such higher-order moment condition could replace or could be added to  \eqref{eq:modelbase_c}.

Finally, one might want to consider some partially linear function, with possibly different index, on the right-had side of \eqref{eq:modelbase_c}. \citet{Lian_Carroll_2015} followed a similar idea in the i.i.d. case.  While considering several series $(Y_i)$ and equations like \eqref{eq:modelbase} is a straightforward matter, a semiparametric model for the square of the error term requires some additional effort.
We argue that our methodology could be extended to such cases, however the investigation of this extension is left for future work.


\section{Appendix: proof of Theorem~\ref{thm:chi2etaunknown}}

The proof contains three parts. 
\begin{itemize}
\item In Section~\ref{app:ELeta} we show that, when the nonparametric elements $\eta_0$ are known,  twice the empirical likelihood ratio converges to a chi-squared distribution, that is $
2\ell_n(\theta_0, \eta_0) \xrightarrow{d} \mathcal{X}_{d_\theta - 1}^2,$ provided $\theta_0$ and $\eta_0$ are the true values of the parameters in the model. 
This part follows the lines of Chapter~11 of \cite{owen2001empirical}  but, because of dependencies between observations, the central limit theorem of ${n}^{-1/2} \sum_{i=1}^n \Psi(Z_i,Z_{i}^{\{r\}};\theta_0, \eta_{ 0}) $ and the Lagrange multipliers should be investigated (see Lemmas~\ref{lemma:TCLfi} and \ref{lemma:lambda} in section~\ref{app:ELeta}).

\item Section~\ref{app:etaestim} is the key part of the proof where we investigate the impact of the estimation of $\eta_0$. We show that
$$
\frac{1}{n} \sum_{i=1}^n \left[\Psi(Z_i,Z_{i}^{\{r\}};\theta_0,\widehat\eta_{\gamma_0}) - \Psi(Z_i,Z_{i}^{\{r\}};\theta_0,\eta_0)\right] = o_{\mathbb{P}}(n^{-1/2}),
$$
and
\begin{multline*}
\frac{1}{n} \sum_{i=1}^n \left[\Psi(Z_i,Z_{i}^{\{r\}};\theta_0,\widehat\eta_{\gamma_0})  \Psi(Z_i,Z_{i}^{\{r\}};\theta_0,\widehat\eta_{\gamma_0})^\top  \right.\\ -\left. \Psi(Z_i,Z_{i}^{\{r\}};\theta_0,\eta_0) \Psi(Z_i,Z_{i}^{\{r\}};\theta_0,\eta_0)^\top \right] = o_\mathbb{P}(1).
\end{multline*}
These differences will be decomposed in several terms. For some of them we simply take the norm, use the triangle inequality and the uniform converge rates for nonparametric estimators for dependent data as presented in \citet{hansen2008uniform}. Some other terms will require a more refined treatment. To show that they are negligible, we use more elaborated arguments based on Inverse Fourier Transform and Davydov's inequality.

\item In Section~\ref{sec:rep2}, we conclude the proof by showing that the asymptotic distribution of the ELR  is not impacted by the estimation of $\eta_0$ and thus, if $\theta_0$ is the true value of the finite-dimensional parameter, 
$
2\ell_n(\theta_0, \widehat \eta_{\gamma_0}) \xrightarrow{d} \mathcal{X}_{d_X + d_W - 1}^2.
$
\end{itemize}

\subsection{Empirical likelihood ratio with $\eta_0$ known} \label{app:ELeta}

Before showing the convergence in distribution of the empirical likelihood ratio we need  three technical lemmas. They are mainly used to show that, when $\eta_0$ are known, then the empirical likelihood ratio converges, under $H_0$, to a chi-squared distribution. Proofs are given in the Supplementary Material. 
\begin{lemma} \label{lemma:TCLfi}
Suppose that Assumptions \ref{ass:ass1}\ref{ass:process} and \ref{variance_s} hold true. Then we have the following central limit theorem
$$
\frac{1}{\sqrt{n}} \sum_{i=1}^n \Psi(Z_i,Z_{i}^{\{r\}};\theta_0, \eta_0) \xrightarrow{d} \mathcal{N}(0,  \Sigma).
$$ 
\end{lemma}

\begin{lemma} \label{lemma:max}
Under the assumptions of Lemma~\ref{lemma:TCLfi}, 
$$
\max_{1\leq i \leq n} \|\Psi(Z_i,Z_{i}^{\{r\}};\theta_0,\eta_0)\| = o_\mathbb{P}(n^{1/2}) \quad \text{ and } 
\sum_{i = 1}^n \|\Psi(Z_i,Z_{i}^{\{r\}};\theta_0,\eta_0)\|^3 = o_\mathbb{P}(n^{3/2}).
$$
\end{lemma}

\begin{lemma}\label{lemma:lambda}
Under the assumptions of Lemma~\ref{lemma:TCLfi}, 
$$
\lambda(\theta_0,\eta_0) = S(\theta_0,\eta_0)^{-1}\frac{1}{n} \sum_{i=1}^n  \Psi(Z_i,Z_{i}^{\{r\}};\theta_0,\eta_0) + o_\mathbb{P}(n^{-1/2}),
$$
and $$S(\theta_0,\eta_0)^{-1}\frac{1}{n} \sum_{i=1}^n  \Psi(Z_i,Z_{i}^{\{r\}};\theta_0,\eta_0)= O_{\mathbb{P}}(n^{1/2}),$$
where $S(\theta,\eta)={n}^{-1} \sum_{i=1}^n \Psi(Z_i,Z_{i}^{\{r\}};\theta,\eta_\gamma) \Psi(Z_i,Z_{i}^{\{r\}};\theta,\eta_\gamma)^\top$.
\end{lemma}

\quad

By a third order Taylor expansion, 
\begin{multline*}
2\ell_n(\theta_0,\eta_0) = 2 \lambda(\theta_0,\eta_0)^\top \sum_{i=1}^n \Psi(Z_i,Z_{i}^{\{r\}};\theta_0,\eta_0)  \\
- \lambda(\theta_0,\eta_0)^\top \left[\sum_{i=1}^n \Psi(Z_i,Z_{i}^{\{r\}};\theta_0,\eta_0)\Psi(Z_i,Z_{i}^{\{r\}};\theta_0,\eta_0)^\top\right]  \lambda(\theta_0,\eta_0) + R_n,
\end{multline*}
where  $R_n$ is the reminder. 
By Lemma~\ref{lemma:max} and Lemma~\ref{lemma:lambda}, we have 
\begin{multline*}
|R_n|\leq 
 O_{\mathbb{P}}\left(\left\| \lambda(\theta_0,\eta_0)\right\|^3\right) \max\left (1, \max_{1\leq i\leq n} |1+\lambda(\theta_0,\eta_0)^\top 
 \Psi(Z_i,Z_{i}^{\{r\}};\theta_0,\eta_0)|^{-1}\right)\\ \times
 \sum_{i=1}^n \left\|\Psi(Z_i,Z_{i}^{\{r\}};\theta_0,\eta_0)\right\|^3
 = O_{\mathbb{P}}(n^{-3/2})O_{\mathbb{P}}(1)o_{\mathbb{P}}(n^{3/2}) = o_{\mathbb{P}}(1).
\end{multline*}
Thus, 
replacing $\lambda(\theta_0,\eta_0)$ by its definition given in Lemma~\ref{lemma:lambda}, we obtain
\begin{multline*}
2\ell_n(\theta_0,\eta_0) =  \left(\frac{1}{\sqrt{n}} \sum_{i=1}^n  \Psi(Z_i,Z_{i}^{\{r\}};\theta_0,\eta_0) \right)S(\theta_0,\eta_0)^{-1} \\ \times \frac{1}{\sqrt{n}}\sum_{i=1}^n \Psi(Z_i,Z_{i}^{\{r\}};\theta_0,\eta_0)   + o_{\mathbb{P}}(1).
\end{multline*}
Note that, by the ergodicity of the process $(Z_i)_{i\in\mathbb Z}$, we have $S(\theta_0,\eta_0)\rightarrow \Sigma$ almost surely  when $n\rightarrow \infty$. From this and Lemma~\ref{lemma:TCLfi}, we have that $S(\theta_0,\eta_0)^{-1/2} {n}^{-1/2}\times\sum_{i=1}^n \Psi(Z_i,Z_{i}^{\{r\}};\theta_0,\eta_0) $ converges to a standard multivariate normal distribution. Hence $2\ell_n(\theta_0,\eta_0) $ converges in distribution to $\mathcal{X}_{d_\theta - 1}^2$.
 
\subsection{Controlling the effect of  the  nonparametric  estimation}\label{app:etaestim}

In the following, $C$, $C_1$, $C_2$, $c\ldots$ denote constants that may change from line to line. 

\subsubsection{The rate of $\|\widehat \eta_{\gamma_0} -  \eta_0 \|$} 

We aim to apply the uniform convergence result from Theorem~4 of \citet{hansen2008uniform}, for which our assumptions allow to verify the required conditions with $d=q=1$. In particular, we guarantee Hansen's condition  (10), with our $\xi$ playing the role of $\boldsymbol \beta$ defined in Hansen's equation (2), and we have $\boldsymbol \theta>1/3$ with $\boldsymbol \theta$ defined in Hansen's equation (11). Then, by Theorem~4 of \citet{hansen2008uniform} and standard second order Taylor expansion, we have
\begin{equation*}
\sup_{t\in \mathbb R } \left| \widehat\eta_{\gamma_0,m}(t) - \eta_{0,m}(t)\right|=O_\mathbb{P}\left( \left(\frac{\ln n}{nh} \right)^{1/2} \right) + O \left( h^2 \right) = o_\mathbb{P}(n^{-1/4}),
\end{equation*}
and similar uniform rates hold for $\widehat\eta_{\gamma_0,f}(t) $, $\widehat\eta_{\gamma_0,X}(t) $ and $\widehat\eta_{\gamma_0,W}(t) $.  
Thus, 
\begin{equation}
\sup_{x\in\mathbb R^{d_X}}\sup_{ w\in\mathbb R^{d_W}} \left\| \widehat \eta_{\gamma_0}^{(-2)}(w^\top\gamma_{0,2}) -  \eta_0^{(-2)}(w^\top\gamma_{0,2}) \right\| \omega(x,w) = o_\mathbb{P}(n^{-1/4}), \label{eq:boundsnp}
\end{equation}
with $\widehat \eta_{\gamma_0}^{(-2)}(\cdot)$ the sub-vector of $ \widehat \eta_{\gamma_0}  (\cdot) $ defined in \eqref{eq:eta_chap} obtained after removing the second component $\widehat\eta_{\gamma_0,m'}(\cdot) $, and $\eta_0^{(-2)}(\cdot)$ the limit of the sub-vector. 
Meanwhile,
\begin{equation}\label{eq:boundsnp_der}
\sup_{t\in \mathbb R  } \left| \widehat\eta_{\gamma_0,m'}(t) - \eta_{0,m'}(t)\right|=O_\mathbb{P}\left( \left(\frac{\ln n}{nh^3} \right)^{1/2} \right) + O \left( h^2 \right)
\end{equation}
Moreover, since  $\eta_{\gamma_0,f}(\cdot) $ is bounded, by the identity $a^k - b^k = (a-b)(a^{k-1}+a^{k-2}b+\cdots +b^{k-1})$, we also have
\begin{equation}\label{etaa_1}
\widehat\eta^k_{\gamma_0,f}(t) =  \eta^k_{0,f}(t) + O_\mathbb{P}\left( \left(\frac{\ln n}{nh} \right)^{1/2} \right)  +O\left( h^2 \right) =  \eta^k_{0,f}(t) + o_\mathbb{P}(n^{-1/4}), \quad t\in\mathbb R,
\end{equation}
with $o_\mathbb{P}(n^{-1/4})$ rate holding uniformly with respect to $t$. We will use this result with $k=2$ and $k=3$.

\subsubsection{PLSIM: the rate of  ${n} ^{-1}\sum_{i=1}^n [\Psi(Z_i;\gamma_0,\widehat\eta_{\gamma_0}) - \Psi(Z_i;\gamma_0,\eta_0) ]$  }
Let $\Phi(\cdot;\dot,\cdot)$ be defined as in \eqref{def_Phi1}. We want to show that
\begin{equation} \label{eq:toshow1}
\left\|\frac{1}{n} \sum_{i=1}^n \Psi(Z_i;\gamma_0,\widehat\eta_{\gamma_0}) - \Psi(Z_i;\gamma_0,\eta_0)\right\| = o_{\mathbb{P}}(n^{-1/2}).
\end{equation}
Let $\Upsilon_i = W_i^\top \gamma_{0,2}$. 
We have 
\begin{equation}\label{eq:decomp_psi}
\frac{1}{n} \sum_{i=1}^n \Psi(Z_i;\gamma_0,\widehat\eta_{\gamma_0})  = \begin{bmatrix}
A_X  \\
A_W  
\end{bmatrix},
\end{equation}
with
\begin{multline*}A_X=\frac{1}{n} \sum_{i=1}^n \left[\left\{\varepsilon_i + m_0(\Upsilon_i ) \right\} \widehat\eta_{\gamma_0,f}(\Upsilon_i ) - \widehat \eta_{\gamma_0,m}(\Upsilon_i ) \right] \\\times  \left[\widehat\eta_{\gamma_0,f}(\Upsilon_i ) \nabla_{\gamma_1} l(X_i;\gamma_{0,1}) -  \widehat \eta_{\gamma_0,X}(\Upsilon_i ) \right]\widehat\eta^2_{\gamma_0,f}(\Upsilon_i ),\end{multline*}
and 
\begin{multline*}
A_W=\frac{1}{n} \sum_{i=1}^n \left[\left\{\varepsilon_i + m_0(\Upsilon_i ) \right\} \widehat\eta_{\gamma_0,f}(\Upsilon_i ) - \widehat \eta_{\gamma_0,m}(\Upsilon_i ) \right] \\\times  \widehat\eta_{\gamma_0,m'}(\Upsilon_i )   \left[ \widehat\eta_{\gamma_0,f}(\Upsilon_i )W_i - \widehat \eta_{\gamma_0,W}(\Upsilon_i )\right].
\end{multline*}

\begin{lemma} \label{lemma:A_X}
 \begin{equation}\label{eq:AXa}
 A_{X}=\frac{1}{n} \sum_{i=1}^n \varepsilon_i \left[  \nabla_{\gamma_1} l(X_i;\gamma_{0,1}) -  \frac{ \eta_{\gamma_0,X}(\Upsilon_i)}{\eta_{\gamma_0,f}(\Upsilon_i)} \right]\eta^4_{\gamma_0,f}(\Upsilon_i) +o_{\mathbb P} (n^{-1/2}) .
 \end{equation}
\end{lemma}

\quad

The proof of Lemma \ref{lemma:A_X} is provided in the Supplement. 

\begin{lemma} \label{lemma:A_W}
 \begin{equation}\label{eq:AWa}
A_{W} = \frac{1}{n} \sum_{i=1}^n  \varepsilon_i  m^\prime _0(\Upsilon_i)  \left[ W_i - \frac{\eta_{\gamma_0,W}(\Upsilon_i)}{\eta_{\gamma_0,f}(\Upsilon_i)} \right] \eta^3_{\gamma_0,f}(\Upsilon_i) + o_\mathbb{P}(n^{-1/2}). 
 \end{equation}
\end{lemma}

\begin{proof}[\color{black}  Proof of \color{black} Lemma \ref{lemma:A_W}]
We rewrite
\begin{multline*}
A_W= \frac{1}{n} \sum_{i=1}^n \left[\left\{\varepsilon_i + m_0(\Upsilon_i ) \right\} \widehat\eta_{\gamma_0,f}(\Upsilon_i ) - \widehat \eta_{\gamma_0,m}(\Upsilon_i ) \right] \\ \times  \widehat\eta_{\gamma_0,m'}(\Upsilon_i )   \left[ \widehat\eta_{\gamma_0,f}(\Upsilon_i )W_i - \widehat \eta_{\gamma_0,W}(\Upsilon_i )\right]\\
= A_{W,0}+ \sum_{l\in\{a,b,\ldots,g\}} A_{W,l}+\sum_{k\in\{1,2,\ldots,6\}} A_{W,k},
\end{multline*}
where
$$
A_{W,0} = \frac{1}{n} \sum_{i=1}^n  \varepsilon_i  m^\prime _0(\Upsilon_i)  \left[ W_i - \frac{\eta_{\gamma_0,W}(\Upsilon_i)}{\eta_{\gamma_0,f}(\Upsilon_i)} \right] \eta^3_{\gamma_0,f}(\Upsilon_i) ,
$$
is the dominating term. The negligible terms could be separated in two groups
$$
A_{W,a} = \frac{1}{n} \sum_{i=1}^n  \varepsilon_i  \eta_{\gamma_0,m'}(\Upsilon_i)  \left[ \widehat\eta_{\gamma_0,f}(\Upsilon_i) W_i -  \eta_{\gamma_0,W}(\Upsilon_i) \right] \widehat\eta_{\gamma_0,f}(\Upsilon_i),
$$
$$
A_{W,b} =  \frac{1}{n}\sum_{i=1}^n   \varepsilon_i  \eta_{\gamma_0,m'}(\Upsilon_i) \left[ \eta_{\gamma_0,W}(\Upsilon_i) -\widehat \eta_{\gamma_0,W}(\Upsilon_i)\right]   \widehat\eta_{\gamma_0,f}(\Upsilon_i),
$$
$$
A_{W,c}= \frac{1}{n} \sum_{i=1}^n  \left[ m_0(\Upsilon_i) \eta_{\gamma_0,f}(\Upsilon_i)- \widehat \eta_{\gamma_0,m}(\Upsilon_i)\right]  \eta_{\gamma_0,m'}(\Upsilon_i) 
\left[\eta_{\gamma_0,f}(\Upsilon_i)  W_i - \widehat  \eta_{\gamma_0,W}(\Upsilon_i)  \right] ,
$$
$$
A_{W,d}= \frac{1}{n} \sum_{i=1}^n m_0(\Upsilon_i)  
\left[  \widehat\eta_{\gamma_0,f}(\Upsilon_i)-  \eta_{\gamma_0,f}(\Upsilon_i)\right] 
\eta_{\gamma_0,m'}(\Upsilon_i)   \left[\eta_{\gamma_0,f}(\Upsilon_i)  W_i  - \widehat  \eta_{\gamma_0,W}(\Upsilon_i)  \right] ,
$$
$$
A_{W,e} = \frac{1}{n}\sum_{i=1}^n  \varepsilon_i \left[\widehat \eta_{\gamma_0,m'}(\Upsilon_i) - \eta_{\gamma_0,m'}(\Upsilon_i)\right] \left[\eta_{\gamma_0,W}(\Upsilon_i) - \widehat \eta_{\gamma_0,W}(\Upsilon_i)\right] 
 \widehat\eta_{\gamma_0,f}(\Upsilon_i),
$$
$$
A_{W,f} = \frac{1}{n}\sum_{i=1}^n  \varepsilon_i \left[ \widehat \eta_{\gamma_0,m'}(\Upsilon_i) -  \eta_{\gamma_0,m'}(\Upsilon_i)  \right] 
\left[\widehat\eta_{\gamma_0,f}(\Upsilon_i)  -   \eta_{\gamma_0,f}(\Upsilon_i) \right]  W_i 
\widehat\eta_{\gamma_0,f}(\Upsilon_i),
$$ 
$$
A_{W,g} = \frac{1}{n}\sum_{i=1}^n  \varepsilon_i \left[ \widehat \eta_{\gamma_0,m'}(\Upsilon_i) -  \eta_{\gamma_0,m'}(\Upsilon_i)  \right] 
\left[\eta_{\gamma_0,f}(\Upsilon_i)  W_i -  \eta_{\gamma_0,W}(\Upsilon_i)  \right] 
 \widehat\eta_{\gamma_0,f}(\Upsilon_i),
$$
and
\begin{multline*}
A_{W,1}= \frac{1}{n} \sum_{i=1}^n  \left[ m_0(\Upsilon_i) \eta_{\gamma_0,f}(\Upsilon_i)- \widehat \eta_{\gamma_0,m}(\Upsilon_i)\right]\\ \times 
\eta_{\gamma_0,m'}(\Upsilon_i) 
\left[\widehat\eta_{\gamma_0,f}(\Upsilon_i)  -   \eta_{\gamma_0,f}(\Upsilon_i) \right]  W_i \widehat\eta_{\gamma_0,f}(\Upsilon_i),
\end{multline*}
\begin{multline*}
A_{W,2}=\frac{1}{n} \sum_{i=1}^n   m_0(\Upsilon_i) \left[ \widehat\eta_{\gamma_0,f}(\Upsilon_i)- \eta_{\gamma_0,f}(\Upsilon_i)\right] \\ \times 
\eta_{\gamma_0,m'}(\Upsilon_i)
\left[\widehat\eta_{\gamma_0,f}(\Upsilon_i)  - \eta_{\gamma_0,f}(\Upsilon_i)  \right] W_i \widehat\eta_{\gamma_0,f}(\Upsilon_i),
\end{multline*}
\begin{multline*}
A_{W,3} = \frac{1}{n}\sum_{i=1}^n 
\left[m_0(\Upsilon_i) \widehat \eta_{\gamma_0,f}(\Upsilon_i)- \widehat \eta_{\gamma_0,m}(\Upsilon_i)\right] \\ \times 
\left[ \widehat \eta_{\gamma_0,m'}(\Upsilon_i) -  \eta_{\gamma_0,m'}(\Upsilon_i)  \right] \left[ \widehat\eta_{\gamma_0,f}(\Upsilon_i)W_i - \eta_{\gamma_0,W}(\Upsilon_i)  \right] ,
\end{multline*}
\begin{multline*}
A_{W,4} = \frac{1}{n}\sum_{i=1}^n  
\left[m_0(\Upsilon_i) \widehat \eta_{\gamma_0,f}(\Upsilon_i)- \widehat \eta_{\gamma_0,m}(\Upsilon_i)\right] \\ \times 
\left[\widehat \eta_{\gamma_0,m'}(\Upsilon_i) - \eta_{\gamma_0,m'}(\Upsilon_i) \right] \left[\eta_{\gamma_0,W}(\Upsilon_i) - \widehat \eta_{\gamma_0,W}(\Upsilon_i)\right] ,
\end{multline*}
\begin{multline*}
A_{W,5} = \frac{1}{n}\sum_{i=1}^n 
\left[  \widehat\eta_{\gamma_0,f}(\Upsilon_i)-  \eta_{\gamma_0,f}(\Upsilon_i)\right]  \\ \times 
\left[ \widehat \eta_{\gamma_0,m'}(\Upsilon_i) -  \eta_{\gamma_0,m'}(\Upsilon_i)  \right] \left[ \widehat\eta_{\gamma_0,f}(\Upsilon_i)W_i - \eta_{\gamma_0,W}(\Upsilon_i)  \right] ,
\end{multline*}
\begin{multline*}
A_{W,6} = \frac{1}{n}\sum_{i=1}^n  
\left[  \widehat\eta_{\gamma_0,f}(\Upsilon_i)-  \eta_{\gamma_0,f}(\Upsilon_i)\right]  \\ \times 
\left[\widehat \eta_{\gamma_0,m'}(\Upsilon_i) - \eta_{\gamma_0,m'}(\Upsilon_i) \right] \left[\eta_{\gamma_0,W}(\Upsilon_i) - \widehat \eta_{\gamma_0,W}(\Upsilon_i)\right] .
\end{multline*}
Taking the norm, by triangle inequality and \eqref{eq:boundsnp}, $ \| A_{W,1}\| +\cdots+ \| A_{W,6}\| = o_\mathbb{P}(n^{-1/2}).$

 Now we investigate $A_{W,b}$, the same arguments will apply to $A_{W,a}$. We have
\begin{multline*}
A_{W,b} = \frac{1}{n} \sum_{i=1}^n \varepsilon_i  \eta_{\gamma_0,m'}(\Upsilon_i) \left[ \widehat\eta_{\gamma_0,W}(\Upsilon_i) - \eta_{\gamma_0,W}(\Upsilon_i)\right]   \widehat\eta_{\gamma_0,f}(\Upsilon_i)\\
=  \frac{1}{n} \sum_{i=1}^n \varepsilon_i  \eta_{\gamma_0,m'}(\Upsilon_i) \left[\widehat\eta^{(-i)}_{\gamma_0,W}(\Upsilon_i) -  \eta_{\gamma_0,W}(\Upsilon_i)\right]   \eta_{\gamma_0,f}(\Upsilon_i)+ r_{W,b}=: A_{W,b}^0 + r_{W,b},
\end{multline*}
with 
$$
\widehat\eta^{(-i)}_{\gamma_0,W}(\Upsilon_i) = \frac{1}{n} \sum_{1\leq j \neq i \leq n}  
    W_j   \frac{1}{h} K\left(\frac{\Upsilon_j - \Upsilon_i}{h}\right) ,
$$ 
and $\|r_{W,b}\|=o_\mathbb{P}(n^{-1/2})$, which is obtained after taking the norm of the sums and using \eqref{eq:boundsnp}. Thus it suffices to show
$
 \| A_{W,b}^0\|   = o_\mathbb{P}(n^{-1/2}).
$
The rate of the norm could be deduced from the same rates of the components. After replacing the expression of $\widehat \eta_{\gamma_0,W}(\Upsilon_i)$  a component of $A_{W,b}^0$ could be written under the form 
$  R_{1,n } - R_{2,n }$ where 
$$
R_{1,n } = \frac{1}{n^2} \sum_{1\leq i \neq j \leq n  }  \lambda(Z_i) \tau (Z_j)\frac{1}{h} K\left(\frac{\Upsilon_j - \Upsilon_i}{h}\right),
$$ and $$
 R_{2,n } =   
 \color{black}  \frac{1}{n} \color{black} \sum_{1\leq i   \leq n  }  \lambda(Z_i) \mathbb E[\tau (Z_i) \mid \Upsilon_i] \eta_{\gamma_0,f}(\Upsilon_i),
$$
with $\lambda(\cdot)$ and $\tau(\cdot)$ some real-valued functions with $\mathbb E [|\lambda (Z_i)|^s + |\tau (Z_j)|^s]<\infty$, and 
\begin{equation}\label{0_cond_mean}
\mathbb E [\lambda(Z_i) \mid \Upsilon_i]=0. 
\end{equation}
Thus our purpose will be to show that
\begin{equation}\label{second_mom}
R_{1,n } - R_{2,n } = o_\mathbb{P}(n^{-1/2}).
\end{equation}
To this end, we will show that
$
\mathbb E \left[( R_{1,n } - R_{2,n })^2 \right] = o_\mathbb{P}(n^{-1}).
$

First, we want to control  $\mathbb E [R^2_{1,n} ] $. Let us note that, applying the Inverse Fourier Transform, 
$$
\frac{1}{h} K\left(\frac{\Upsilon_j - \Upsilon_i}{h}\right)  = \int_{\mathbb R} e^{2\pi \iota t (\Upsilon_j - \Upsilon_i)}\mathcal{F}[K](th) dt,
$$
where $\iota = \sqrt{-1}$ and $\mathcal{F}[K](\cdot) $
is the Fourier Transform of the kernel $K(\cdot)$.  Note that, since $\mathcal{F}[K](\cdot)$ is supposed to be integrable, there exists some constant $C$ such that
$0<\int_{\mathbb R}| \mathcal{F}[K](th) | dt \leq C h^{-1}.$ 
Next, we can write
\begin{multline*}
n^4R^2_{1,n } =   \int_{\mathbb R}   \left[ \sum_{1\leq i \neq j \leq n} \left\{ \lambda(Z_i) e^{-2\pi \iota t \Upsilon_{i}}   \tau (Z_j)  e^{2\pi \iota t \Upsilon_j  }\right\} \right] \mathcal{F}[K](th) dt  \\
\times  \int_{\mathbb R}   \left[ \sum_{1\leq i^\prime \neq j^\prime  \leq n} \left\{ \lambda(Z_{i^\prime}) e^{-2\pi \iota t^\prime  \Upsilon_{i^\prime }}   \tau (Z_{j^\prime}) e^{2\pi \iota t^\prime \Upsilon_{j^\prime} }\right\} \right] \mathcal{F}[K](t^\prime h) dt^\prime . 
\end{multline*}
Thus we have
$$
n^4\mathbb E [R^2_{1,n }] = \int_{\mathbb R} \int_{\mathbb R} \sum_{1\leq i \neq j \leq n} \sum_{1\leq i^\prime \neq j^\prime  \leq n}  \mathbb E \left\{ \Lambda(i,j,i^\prime,j^\prime;t,t^\prime) \right\}  \mathcal{F}[K](th) dt  \mathcal{F}[K](t^\prime h) dt^\prime ,
$$
with 
$$
\Lambda(i,j,i^\prime,j^\prime;t,t^\prime) = \lambda(Z_i) e^{-2\pi \iota t \Upsilon_{i}}   \tau (Z_j)  e^{2\pi \iota t \Upsilon_j  }\lambda(Z_{i^\prime}) e^{-2\pi \iota t^\prime  \Upsilon_{i^\prime }}   \tau (Z_{j^\prime}) e^{2\pi \iota t^\prime \Upsilon_{j^\prime} }.
$$

For $1\leq m\leq n$, let $\mathcal I(i;m) = \{k:1\leq k\leq n, |k-i|< m\}$, the set of indices from 1 to $n$ in the $m-$neighborhood of $i$, and $\mathcal I^c(i;m) = \{k:1\leq k\leq n, k\not\in \mathcal I(i;m)\}$. Let $0 < \delta<1$ be a real number that will be specified below. Note that the set 
$$
\mathcal I = \bigcup_{1\leq i\neq i^\prime \leq n} \left\{ (i,j,i^\prime,j^\prime): \{j,j^\prime\}\subset \mathcal I(i;n^\delta/2) \cup \mathcal I(i^\prime;n^\delta/2) \right\} ,
$$
is a set of cardinality of order   $n^{2\delta}n^2$. First, consider the case 
$$
|i-i^\prime|\geq n^{\delta}. 
$$ 
In this case, by \color{black}  \eqref{0_cond_mean}, \color{black} \eqref{eq:def_alpha_ew} and Davydov's inequality from Theorem A.6 of  \citep{Hall_Heyde} with $p, q>1$ that will be determined in the sequel, we have, for any $t,t^\prime\in\mathbb R$ and for any $|i-i^\prime|\geq n^\delta$ and any $\{j,j^\prime\}\not\subset \mathcal I(i;n^\delta/2) \cup \mathcal I(i^\prime;n^\delta/2)$
\begin{equation}\label{LL3}
 \left| \mathbb E \left\{ \Lambda(i,j,i^\prime,j^\prime;t,t^\prime) \right\} \right| \leq C n^{-\xi \delta/p}, 
\end{equation}
for some constant $C>0$ independent of $i,i^\prime, j,j^\prime, t$ and $t^\prime$. Indeed, if at least one of the indices $j$ and $j^\prime$ is not in the $n^\delta/2-$neighborhood of $i$ or $i^\prime$, then we have $\max\{\min(|i-j|, |i-j^\prime|), \min(|i^\prime -j|, |i^\prime-j^\prime|)\}\geq n^\delta/2.$ This means we could isolate one of the  indices $i$ and $i^\prime$ by a  $n^\delta /2-$neighborhood, and, possibly after repeated applications,  we could apply Davydov's inequality with, say, $Y=\lambda(Z_i) e^{-2\pi \iota t \Upsilon_{i}} $. Then, \color{black} by  \eqref{0_cond_mean} we have \color{black} $\mathbb E (Y)=0$ and we thus obtain \eqref{LL3}. 
For the multi-indices satisfying $|i-i^\prime|\geq n^{\delta}$ but belonging to $\mathcal I$, we could simply bound $\left| \mathbb E \left\{ \Lambda(i,j,i^\prime,j^\prime;t,t^\prime) \right\} \right|$ using Cauchy-Schwarz inequality and recall the negligible cardinality  of order   $n^{2\delta}n^2$ of $\mathcal I$. From the investigation of all types of situations we note that we could take 
  $p=s/(s-3)$ and $q=s/3$. 

Now, we  distinguish two sub-cases to deal with the case where
$
0<|i-i^\prime|< n^{\delta}. 
$ 
The absolute value of another pair of indices is smaller than $n^{\delta}$ or the absolute values of all the other five pairs we could make with $i,j,i^\prime$ and $j^\prime$ are larger than  $n^{\delta}$. In the former case, the cardinality of the set of multi-indices is of order at most $n^{2\delta}n^2$. In the later case, we could apply Davydov's inequality with a split of $\Lambda(i,j,i^\prime,j^\prime;t,t^\prime) $ in $X$ and $Y$ such that, say, \color{black}  $Y=\lambda(Z_i) e^{-2\pi \iota t \Upsilon_{i}} $.
In such a case, by  \eqref{0_cond_mean}, $\mathbb E (Y)=0$ \color{black} and we have a bound as in \eqref{LL3}.

 The case $i=i^\prime$ requires a special attention, that is we have to study 
\begin{multline*}
\Lambda^0_n = \int_{\mathbb R} \int_{\mathbb R} \sum_{1\leq i \neq j\neq j^\prime  \leq n}  \mathbb E \left\{ \Lambda(i,i,i^\prime,j^\prime;t,t^\prime) \right\}  \mathcal{F}[K](th) dt  \mathcal{F}[K](t^\prime h) dt^\prime\\
=  \int_{\mathbb R} \int_{\mathbb R} \sum_{1\leq i \neq j\neq j^\prime  \leq n}  \mathbb E \left\{  \lambda^2(Z_i) e^{-2\pi \iota (t+t^\prime) \Upsilon_{i}}   \tau (Z_j)   e^{2\pi \iota t  \Upsilon_{j }}   \tau (Z_{j^\prime}) e^{2\pi \iota t^\prime \Upsilon_{j^\prime} } \right\} \\ \times  \mathcal{F}[K](th) dt  \mathcal{F}[K](t^\prime h) dt^\prime.
\end{multline*}
In the case where in addition $\min (|i-j|,|i-j^\prime|, |j-j^\prime|) < n^\nu$, for some $0<\nu<1$ that will be specified below, the cardinality of the set of multi-indices is of order at most $n^{2+\nu}$. Then a bound for the sum over the set of these multi-indices is obtained easily using the small cardinality of the set  (small compared to $n^4$, the order of the cardinality of the full set of multi-indices $(i,j,i^\prime,j^\prime)$ with $i\neq j$ and $i^\prime\neq j^\prime$) 
and Cauchy-Schwarz inequality. For multi-indices such that $i=i^\prime $ and $\min (|i-j|,|i-j^\prime|, |j-j^\prime|) \geq n^\nu$, applying twice Davydov's inequality with $p=s/(s-2)$ and $q=s/2$, 
\begin{multline}\label{sep_a}
\mathbb E \left\{  \lambda^2(Z_i) e^{-2\pi \iota (t+t^\prime) \Upsilon_{i}}   \tau (Z_j)   e^{2\pi \iota t  \Upsilon_{j }}   \tau (Z_{j^\prime}) e^{2\pi \iota t^\prime \Upsilon_{j^\prime} } \right\} \\ = \mathbb E \left\{  \lambda^2(Z_i) e^{-2\pi \iota (t+t^\prime) \Upsilon_{i}}  \right\}\mathbb E \left\{   \tau (Z_j)   e^{2\pi \iota t  \Upsilon_{j }}   \right\} \mathbb E \left\{  \tau (Z_{j^\prime}) e^{2\pi \iota t^\prime \Upsilon_{j^\prime} } \right\} +  O(n^{-\xi \nu(s-2)/s)})\\= \mathbb E \left\{  \overline{\lambda^2}(\Upsilon_i) e^{-2\pi \iota (t+t^\prime) \Upsilon_{i}}  \right\} \mathbb E \left\{   \overline{\tau} (\Upsilon_j)   e^{2\pi \iota t  \Upsilon_{j }}   \right\}\mathbb E \left\{  \overline{\tau} (\Upsilon_{j^\prime}) e^{2\pi \iota t^\prime \Upsilon_{j^\prime} } \right\} +  O(n^{-\xi \nu(s-2)/s)}),
\end{multline}
where $\overline{\lambda^2}(\Upsilon_i)  = \mathbb E [\lambda^2(Z_i) \mid \Upsilon_i]$, 
 $\overline{\tau}(\Upsilon_j)  = \mathbb E [\tau(Z_j) \mid \Upsilon_j]$,  $\overline{\gamma}(\Upsilon_{j^\prime})  = \mathbb E [\tau(Z_{j^\prime}) \mid \Upsilon_{j^\prime}]$. Moreover, the rate $O(n^{-\xi \nu(s-2)/s)})$ of the reminder is  uniform with respect to $t$ and $t^\prime$. We deduce 
\begin{multline*}
\int_{\mathbb R} \int_{\mathbb R} \mathbb E \left\{ \Lambda(i,i,i^\prime,j^\prime;t,t^\prime) \right\}  \mathcal{F}[K](th) dt  \mathcal{F}[K](t^\prime h) dt^\prime\\ = \int_{\mathbb R} \int_{\mathbb R} \mathcal{F}[\overline{\lambda^2}\eta_{\gamma_0,f}](t+t^\prime) \mathcal{F}[\overline \tau\eta_{\gamma_0,f}](-t) \mathcal{F}[\overline\tau\eta_{\gamma_0,f}](-t^\prime)  \mathcal{F}[K](th) dt  \mathcal{F}[K](t^\prime h) dt^\prime\\
 + O(n^{-\xi \nu(s-2)/s)})\int_{\mathbb R} \int_{\mathbb R}  \mathcal{F}[K](th) dt  \mathcal{F}[K](t^\prime h) dt^\prime.
\end{multline*}
Let $0<c<1$ to be specified below. 
Note that, 
since $\mathcal{F}[K](0)=1$,
\begin{multline*}
\int_{\mathbb R} \left\{ \mathcal{F}[\overline{\lambda^2}\eta_{\gamma_0,f}](t+t^\prime) \mathcal{F}[\overline \tau\eta_{\gamma_0,f}](-t)  \right\}\mathcal{F}[K](th)   dt = \int_{|th|\leq h^c}\{\cdots\}  dt \\+  \int_{|th|\leq h^c}\{\cdots\}\left\{ \mathcal{F}[K](th)  -\mathcal{F}[K](0) \right\}dt 
+\int_{|th| > h^c} \{\cdots\}\mathcal{F}[K](th)   dt\\
=: I_1(t^\prime;h) + I_2(t^\prime;h) + I_3(t^\prime;h). 
\end{multline*}
Since $\mathcal{F}[\lambda^2 \eta_{\gamma_0,f}](\cdot)$ and $\mathcal{F}[\overline \tau\eta_{\gamma_0,f}](\cdot)$ are squared integrable,
$$
|I_1(t^\prime;h)| \leq \int_{\mathbb R} \left|  \mathcal{F}[\overline{\lambda^2}\eta_{\gamma_0,f}](t+t^\prime) \mathcal{F}[\overline \tau\eta_{\gamma_0,f}](-t)    \right|dt<\infty .
$$
Moreover, since $\mathcal{F}[K](\cdot) $ is Lipschitz continuous and  $\mathcal{F}[\lambda^2 \eta_{\gamma_0,f}](\cdot)$ and $\mathcal{F}[\overline \tau\eta_{\gamma_0,f}](\cdot)$ are bounded,  for some constant $C$,
$$
\left| I_2(t^\prime;h) \right| \leq Ch \int_{|t|\leq h^{c-1}}|t|dt = h^{2c-1}\rightarrow 0,
$$
the convergence to zero being guaranteed  as soon as $c> 1/2$. Finally, by Assumption \ref{ass:ass3}, 
\begin{multline*}
\left| I_3(t^\prime;h) \right| \leq C_1 h^{-1} \int_{|u|> h^{c-1}}|\mathcal{F}[K](u) |du \leq C_2 h^{-1} \int_{|u|> h^{c-1}}u^{-c_K}du \\= C_2 h^{-1}  h^{(c_K-1)(1-c)} \rightarrow 0,
\end{multline*}
with some constants $C_1, C_2$. The convergence to zero holds as soon as $c$ is smaller than $(c_K-2)/(c_K-1)$ which is larger than 1/2 provided $c_K>3$. Next, we integrate with respect to $t^\prime$ and we decompose the integral in a similar way, that is we write 
\begin{multline*}
\int_{\mathbb R}  \{I_1(t^\prime;h) + I_2(t^\prime;h) + I_3(t^\prime;h) \}\mathcal{F}[K](t^\prime h) dt^\prime = \int _{|t^\prime h|\leq h^c}  \{\cdots\} dt^\prime\\
+  \int_{|t^\prime h|\leq h^c}\{\cdots\}\left\{ \mathcal{F}[K](t^\prime h)  -\mathcal{F}[K](0) \right\}dt^\prime 
+\int_{|t^\prime h| > h^c} \{\cdots\}\mathcal{F}[K](t^\prime h)   dt^\prime \\
=: J_{11}(h) + J_{12}(h) +\cdots +J_{32}(h) + J_{33}(h).
\end{multline*}
By Dominated Convergence Theorem and Convolution Theorem for Fourier Transform, 
\begin{multline*}
J_{11}(h) \rightarrow \int_{\mathbb R} \int_{\mathbb R} \mathcal{F}[\overline{\lambda^2}\eta_{\gamma_0,f}](t+t^\prime) \mathcal{F}[\overline \tau \eta_{\gamma_0,f}](-t) \mathcal{F}[\overline\tau\eta_{\gamma_0,f}](-t^\prime)    dt   dt^\prime\\=\mathcal{F}[\overline{\lambda^2}\eta_{\gamma_0,f}\overline \tau\eta_{\gamma_0,f}\overline \tau\eta_{\gamma_0,f}](0)  =  \mathbb{E}[ \overline{\lambda^2}(\Upsilon_i)  \{\overline \tau \eta_{\gamma_0,f} \}^2 (\Upsilon_i) ] .
 \end{multline*}
Meanwhile, by the same arguments as above, the other eight terms $ J_{kl}(h) $, with $1\leq k,l\leq 3$ and $(k,l)\neq (1,1)$, tend to zero.

Now we can deduce 
\begin{multline*}
\mathbb E [R^2_{1,n }] -n^{-1}\mathbb{E}\left\{  \lambda^2(Z_i)    \mathbb E ^2 [\tau (Z_i) \mid \Upsilon_i] \eta_{f,\gamma_0}^2 (\Upsilon_i) \right\}\\ =n^{-1}\times \left[ h^{-2} O\left(n^{-\{\xi \delta s^{-1}(s-3)-1\}} + n^{-(1-2\delta)}\right) +  h^{-2} O\left(n^{-\xi \nu s^{-1} (s-2)}+n^{-(1-\nu)}\right)\right].
 \end{multline*}
Taking 
$$
\delta = \frac{2s}{\xi(s-3)+2s}\quad \textrm{ and } \quad \nu = \frac{s}{\xi(s-2)+s},
$$
to guarantee $ \mathbb E [R^2_{1,n }] - \mathbb{E}\left\{  \lambda^2(Z_i)    \mathbb E ^2 [\tau (Z_i) \mid \Upsilon_i] \eta_{f,\gamma_0}^2 (\Upsilon_i) \right\} = o(n^{-1}),$
we need the conditions $0<\delta,\nu<1$ and 
$$
nh^{\rho} \rightarrow \infty \; \textrm{ with } \; \rho =  \frac{2[\xi(s-3)+2s]}{\xi (s-3) - 2s}\quad\textrm{and} \quad  \rho =  \frac{2[\xi(s-2)+s]}{\xi (s-2) }.
$$ 
Since the first expression of $\rho$ is the last display is always larger than the second one, provided $s>3$, we only have to ensure 
$nh^{\rho} \rightarrow \infty$ for the first expression of $\rho$. 
Note that in both cases $\rho < 3$ as soon as $\xi > 10s/(s-3) $. 

Next, note that by \eqref{0_cond_mean}
$$
\mathbb E [R^2_{2,n} ] =  n^{-1}\mathbb E \left\{ \lambda^2(Z_i) \mathbb E^2[\tau (Z_i) \mid \Upsilon_i] \eta^2_{\gamma_0,f}(\Upsilon_i) \right\} \{1+o(n^{-1})\}.
$$
It remains to study 
$$
\color{black}  n^3 \color{black} \mathbb E [R_{1,n}R_{2,n} ]= 
\int_{\mathbb R}  \sum_{1\leq i \neq j \leq n} \sum_{1\leq i^\prime    \leq n}  \mathbb E \left\{ \Gamma(i,j,i^\prime;t) \right\}  \mathcal{F}[K](th) dt   ,
$$
with 
$$
\Gamma(i,j,i^\prime;t)  = \lambda(Z_i) e^{-2\pi \iota t \Upsilon_{i}}   \tau (Z_j)  e^{2\pi \iota t \Upsilon_j  }\lambda(Z_{i^\prime})  \mathbb E[\tau (Z_{i^\prime}) \mid \Upsilon_{i^\prime}] \eta_{\gamma_0,f}(\Upsilon_{i^\prime})   .
$$
Repeating the same arguments as above,  
the leading term of $n^4 (n-1)^{-1}\mathbb E [R_{1,n}R_{2,n} ]$ is obtained summing the terms $\int_{\mathbb R}
\mathbb E [\Gamma(i,j,i; t) ]\mathcal{F}[K](th) dt $ over all the pairs $(i,j)$. Moreover, only the pairs for which $|i-j|$ is sufficiently large will matter. As a consequence, after applying Davydov's inequality, the leading terms will be 
\begin{multline*}
\int_{\mathbb R}   \mathbb E\left\{\lambda^2(Z_i)   \mathbb E[\tau (Z_{i}) \mid \Upsilon_{i}] \eta_{\gamma_0,f}(\Upsilon_{i })  e^{-2\pi \iota t \Upsilon_{i}}\right\} \mathbb E\left\{ \tau (Z_j)  e^{2\pi \iota t \Upsilon_j  }   \right\}  \mathcal{F}[K](th) dt \\
= \int_{\mathbb R}   \mathbb E\left\{\lambda^2(Z_i)   \mathbb E[\tau (Z_{i}) \mid \Upsilon_{i}] \eta_{\gamma_0,f}(\Upsilon_{i })  e^{-2\pi \iota t \Upsilon_{i}}\right\} \mathbb E\left\{ \mathbb E[\tau (Z_{j}) \mid \Upsilon_{j}] e^{2\pi \iota t \Upsilon_j  }   \right\}  \mathcal{F}[K](th) dt \\
= \mathcal{F}[\overline{\lambda^2}\eta_{\gamma_0,f}\overline \tau\eta_{\gamma_0,f}\overline \tau\eta_{\gamma_0,f}](0)  \{1+o(1)\}  ,
\end{multline*}
where for the last equality we used the same arguments as above. Deduce that 
$$
\mathbb E [R_{1,n}R_{2,n} ]  = n^{-1}\mathbb E \left\{ \lambda^2(Z_i) \mathbb E^2[\tau (Z_i) \mid \Upsilon_i] \eta^2_{\gamma_0,f}(\Upsilon_i) \right\} \{1+o(n^{-1})\},
$$
and thus \eqref{second_mom} holds true, and 
$
\mathbb E \left( \| A_{W,b}\| \right) = o(n^{-1/2}).
$

Next, we have to investigate $A_{W,e}$, $A_{W,f}$ and $A_{W,g}$. This could not be bounded by simply taking the norm of the sum. Indeed, since the nonparametric estimator of the derivative has a slower rate of convergence given in \eqref{eq:boundsnp_der}, this would not yield  a sufficiently fast rate for these terms. To improve the rate  we have to exploit \eqref{eq:modelbase_b}. For this, we use again the steps we followed for $A_{W,b}$: replace $\widehat\eta_{\gamma_0,f}(\Upsilon_i)$ by $\eta_{\gamma_0,f}(\Upsilon_i)$, replace the expressions of the nonparametric estimator, and compute the second order moment of the resulting average over  three indices. Next, we partition the set of six components multi-indices, obtained when considering the second order moment, in three subsets that could be handled either  using  Cauchy-Schwarz inequality and the negligible cardinality of the subset, or using Davydov's inequality and a condition like \eqref{0_cond_mean}, or using the Inverse Fourier Transform for $K$ and $K^\prime$. The later category of multi-indices corresponds to the expectation of the terms containing the factor $\varepsilon_i^2$. The adaptation of the previous arguments for $A_{W,e}$, $A_{W,f}$ and $A_{W,g}$ is quite straightforward and thus we omit the details.
Deduce $
\mathbb E \left( \| A_{W,e}\| + \| A_{W,f}\| + \| A_{W,g}\|\right) = o(n^{-1/2}).
$

Finally, we have to investigate  $A_{W,c}$ and  $A_{W,d}$. \color{black}  Up to a sign, \color{black} we can write each of these terms under the form
\begin{multline*}
 \frac{1+o_{\mathbb P}(1)}{n}\sum_{i=1}^n \delta(\Upsilon_i) \left\{ \widehat \gamma_n(\Upsilon_i) - \gamma(\Upsilon_i)\right\}  \lambda_1 (W_i)  \\ -  \frac{1+o_{\mathbb P}(1)}{n}\sum_{\textcolor{black} {i=1}}^n \delta(\Upsilon_i)  \left\{ \widehat \gamma_n(\Upsilon_i) - \gamma(\Upsilon_i)\right\} 
\left\{ \frac{1}{n}\sum_{j=1}^n \lambda_2 (\Upsilon_i, \Upsilon_j;h) \right\}\\=: \{A_{1}-A_2\}\{1+o_{\mathbb P}(1)\}, 
\end{multline*}
where $A_{1}$ and $A_2$ are the sums corresponding to  
$$
\lambda_1 (W_i) = [W_i-\mathbb E (W_i \mid \Upsilon_i)]\eta_{\gamma_0,f}(\Upsilon_i),
$$
$$ \lambda_2 (t, \Upsilon_j;h)= W_j\frac{1}{h} K\left(\frac{\Upsilon_j -t}{h}\right)  - \mathbb E (W_i \mid \Upsilon_i=t)\eta_{\gamma_0,f}(t),
$$
respectively. Here, $\widehat \gamma_n(\cdot)$ is either  the kernel estimator of  $\gamma(\cdot)= \eta_{\gamma_0,m}(\cdot)$ (and then $\delta(\cdot) = \eta_{\gamma_0,m^\prime}(\cdot)$), or of  $\gamma(\cdot)= \eta_{\gamma_0,f}(\cdot)$ (and then $\delta(\cdot) =m_0(\cdot) \eta_{\gamma_0,m^\prime}(\cdot)$). By \eqref{eq:boundsnp},  $\widehat \gamma_n(\cdot)$  is uniformly convergent with rate $o_{\mathbb P}(n^{-1/4})$. First we investigate the variance of $A_2$ for which we could apply again the uniform rate \eqref{eq:boundsnp} and deduce $
\mathbb E \left( \| A_{2}\| \right) = o(n^{-1/2}).
$
\color{black}  Next we handle $A_1$. \color{black} After replacing the expression of the kernel estimator, \color{black}  up to a remainder of order $o_{\mathbb P }(n^{-1/2})$, we could rewrite 
\begin{multline*}
A_{1} = \frac{1}{n^2} \sum_{1\leq i \neq j \leq n  }  \color{black}  \delta(\Upsilon_i) \color{black} \lambda_1(W_i) \tau_1 (W_j)\frac{1}{h} K\left(\frac{\Upsilon_j - \Upsilon_i}{h}\right) \color{black}  \\
- \frac{1}{n} \sum_{1\leq i \leq n  }\!\! \delta(\Upsilon_i) \lambda_1(W_i) \mathbb E [\tau_1 (W_i)\mid \Upsilon_i] 
=:  R_{A,1,n }-R_{A,2,n }  , \color{black}
\end{multline*}
\color{black}  \emph{i.e.,} an expression similar  to $R_{1,n }-R_{2,n } $ from the decomposition of  $A_{W,b}$. For instance, $\tau_1 (W_j) = Y_j - l(X_j;\gamma_{0,1})$ for $A_{W,c}$. \color{black}
Here, instead of \eqref{0_cond_mean} we have
\begin{equation}\label{0_cond_mean2}
\mathbb E [ \color{black}  \delta(\Upsilon_i) \color{black} \lambda_1 (W_i)\mid \Upsilon_i] 
 \color{black}  =\delta(\Upsilon_i)\mathbb E [ \lambda_1 (W_i)\mid \Upsilon_i]  
= 0,
\end{equation}
and all the arguments used for $R_{1,n }$ and  $R_{2,n }$ remain valid for $R_{A,1,n }$ and  $R_{A,2,n }$ with 
$\lambda(Z_i)$ replaced by $\delta(\Upsilon_i)\lambda_1 (W_i)$.  
Gathering facts, deduce $\mathbb E \left( \| A_{W,c}\| + \| A_{W,d}\| \right) = o(n^{-1/2})$. Now the proof of Lemma \ref{lemma:A_W} is complete.
\end{proof}

\subsubsection{CHPLSIM: the rate of  ${n} ^{-1}\sum_{i=1}^n [\Psi(Z_i,Z_{i}^{\{r\}};\theta_0,\widehat\eta_{\gamma_0}) - \Psi(Z_i,Z_{i}^{\{r\}};\theta_0,\eta_0)] $}
Let $\Phi(\cdot,\cdot;\cdot,\cdot)$ be defined as in \eqref{def_Phi2}. We want to show a rate like \eqref{eq:toshow1}. The only difference compared to the PLSIM comes from the second set of equations. Let 
$$
\Psi_\sigma (Z_i,Z_{i}^{\{r\}};\theta,\eta_{\gamma}) = g_\sigma(Z_i,Z_{i}^{\{r\}};\theta,m_\gamma ) \nabla_\beta \sigma^2(V_i,Z_{i}^{\{r\}};\beta)\eta_{\gamma,f}^2(W_i^\top\gamma_2)\in\mathbb R^{d_\beta },
$$
where we recall that
$$
g_\sigma(Z_i,Z_{i}^{\{r\}};\theta,m_\gamma) =  \{ Y_i-  l(X_i; \gamma_1) - m_\gamma (\Upsilon_i)\}^2- \sigma^2(V_i,Z_{i}^{\{r\}};\beta).
$$
With the notations $\sigma^2_i=\sigma^2(V_i,Z_{i}^{\{r\}};\beta)$ and $\nabla_\beta \sigma^2_i=\nabla_\beta \sigma^2(V_i,Z_{i}^{\{r\}};\beta)$, we could decompose 
\begin{multline*}
\frac{1}{n}\sum_{i=1}^n \left[\Psi(Z_i,Z_{i}^{\{r\}};\theta_0,\widehat\eta_{\gamma_0}) -\Psi(Z_i,Z_{i}^{\{r\}};\theta_0,\eta_0) \right] \\
= \frac{1}{n}\sum_{i=1}^n \left\{\varepsilon_i^2 -  \sigma_i^2\right\} \nabla_\beta \sigma^2_i \;\{\widehat \eta_{\gamma,f}(\Upsilon_i) -\eta_{\gamma,f}(\Upsilon_i) \} \;\{2\eta_{\gamma,f}(\Upsilon_i) +o_{\mathbb P} (n^{-1/4})\}\\
+\frac{1}{n}\sum_{i=1}^n \left\{\widehat \eta_{\gamma_0,m}  (\Upsilon_i)- m_0 (\Upsilon_i) \widehat \eta_{\gamma_0,f}  (\Upsilon_i)\right\}^2  \nabla_\beta \sigma^2_i\\
+ \frac{2}{n}\sum_{i=1}^n  \varepsilon_i m_0 (\Upsilon_i)  [ \widehat \eta_{\gamma_0,f}  (\Upsilon_i)- \widehat \eta_{\gamma_0,f}  (\Upsilon_i)]  \nabla_\beta \sigma^2_i \;\{\eta_{\gamma,f}(\Upsilon_i) +o_{\mathbb P} (n^{-1/4})\} \\
- \frac{2}{n}\sum_{i=1}^n  \varepsilon_i [ \widehat \eta_{\gamma_0,m}  (\Upsilon_i)- m_0 (\Upsilon_i)  \eta_{\gamma_0,f}  (\Upsilon_i)]  \nabla_\beta \sigma^2_i \;\{\eta_{\gamma,f}(\Upsilon_i) +o_{\mathbb P} (n^{-1/4})\} 
\\= B_1+B_2+2B_3-2B_4 +o_{\mathbb P} (n^{-1/2}),
\end{multline*}
where the reminder $o_{\mathbb P} (n^{-1/2})$ is obtained by taking the norms of the sums where we obtain a product of two quantities with uniform rates $o_{\mathbb P} (n^{-1/4})$. 
Taking the norm of the sum, using triangle inequality and \eqref{eq:boundsnp}, $\|B_2\|=o_{\mathbb P}(n^{-1/2})$. Next, by the definition of the model, $\mathbb E\left[\varepsilon_i^2 -  \sigma_i^2 \mid V_i,\mathcal F_{i-1}\right]=0$ a.s. A careful inspection of the arguments for deducing the rate of $A_{W,a}$ and $A_{W,b}$ in Lemma \ref{lemma:A_W} reveals that the arguments remain valid if the function $\lambda (\cdot)$ appearing in the definition of $R_{1,n}$ and $R_{n,2}$, depends also on $Z_{i}^{\{r\}}$, that is $\lambda(Z_i)$ becomes $\lambda(Z_i, Z_{i}^{\{r\}})$, \color{black}  and 
\eqref{0_cond_mean} is replaced by the condition  $\mathbb E [\lambda(Z_i, Z_{i}^{\{r\}}) \mid \Upsilon_i,\mathcal F_{i-1}]=0$. \color{black} Here we consider
$$
\lambda(Z_i, Z_{i}^{\{r\}}) =  \left\{\varepsilon_i^2 -  \sigma_i^2\right\} \nabla_\beta \sigma^2_i \;\eta_{\gamma,f}(\Upsilon_i) ,$$
$$
\lambda(Z_i, Z_{i}^{\{r\}}) =  
\varepsilon_i m_0 (\Upsilon_i)   \nabla_\beta \sigma^2_i \; \eta_{\gamma,f}(\Upsilon_i) 
$$
and $$\lambda(Z_i, Z_{i}^{\{r\}}) =  
\varepsilon_i   \nabla_\beta \sigma^2_i \;\eta_{\gamma,f}(\Upsilon_i) $$ to handle $B_1$, $B_3$ and $B_4$, respectively. We have
$\mathbb E [\lambda(Z_i, Z_{i}^{\{r\}}) \mid V_i,\mathcal F _{i-1}]=0$ a.s., for these three definitions, and thus condition \eqref{0_cond_mean} holds true.  Deduce that  $\mathbb E (\|B_1\| + \|B_3\| + \|B_4\|) =o_{\mathbb P}(n^{-1/2})$.

\subsubsection{Controlling the variance estimation error}\label{sec:rep1}
By the previous arguments it is now easy to deduce that
\begin{multline} 
 \left\|\frac{1}{n} \sum_{i=1}^n\left[\Psi(Z_i ,Z_{i}^{\{r\}} ; \theta_0,\widehat \eta_{\gamma_0})\Psi(Z_i, Z_{i}^{\{r\}};\theta_0,\widehat \eta_{\gamma_0})^\top\right.\right. \\\! -\!  \left.\left. \Psi(Z_i,Z_{i}^{\{r\}}\!;\theta_0,\eta_0)\Psi(Z_i, Z_{i}^{\{r\}};\theta_0,\eta_0)^\top\right]\right\| = o_{\mathbb{P}}(1).
\end{multline}

\subsection{Empirical likelihood ratio with $\widehat \eta_{\gamma_0}$} \label{sec:rep2}
We have
\begin{align*}
 \lambda(\theta_0,\widehat \eta_{\gamma_0}) &= S(\theta_0,\widehat \eta_{\gamma_0})^{-1} \frac{1}{n} \sum_{i=1}^n \Psi(Z_i,Z_{i}^{\{r\}};\theta_0,\eta_0) + o_\mathbb{P}(n^{-1/2}) \\
&= S(\theta_0,\eta_0)^{-1} \frac{1}{n} \sum_{i=1}^n \Psi(Z_i,Z_{i}^{\{r\}};\theta_0,\eta_0)  \\& \;\;\; +o_\mathbb{P}(1)\frac{1}{n} \sum_{i=1}^n \Psi(Z_i,Z_{i}^{\{r\}};\theta_0,\eta_0) + o_\mathbb{P}(n^{-1/2}) .
\end{align*}
Thus, the CLT for  $n^{-1} \sum_{i=1}^n \Psi(Z_i;\theta_0,\eta_0)$ implies that $\| \lambda(\theta_0,\widehat \eta_{\gamma_0}) -  \lambda(\theta_0,\eta_0)\| = o_\mathbb{P}(n^{-1/2})$. Moreover, since $\| \lambda(\theta_0,\eta_0)\|=O_\mathbb{P}(n^{-1/2})$,
\begin{eqnarray*}
2\ell_n(\theta_0,\widehat \eta_{\gamma_0}) \!\!&=& \!\!\! 2  \lambda(\theta_0,\widehat \eta_{\gamma_0})^\top \sum_{i=1}^n \Psi(Z_i,Z_{i}^{\{r\}};\theta_0,\widehat \eta_{\gamma_0}) \\ \!\!\!&& \!-  \lambda(\theta_0,\widehat \eta_{\gamma_0})^\top \left[\sum_{i=1}^n \Psi(Z_i,Z_{i}^{\{r\}};\theta_0,\widehat \eta_{\gamma_0})\Psi(Z_i,Z_{i}^{\{r\}};\theta_0,\widehat \eta_{\gamma_0})^\top\right]  \lambda(\theta_0,\widehat \eta_{\gamma_0})\\ &&\!+ o_{\mathbb{P}}(1)\\
\!\!\!&=& \!\! \lambda(\theta_0,\eta_0)^\top\sum_{i=1}^n \Psi(Z_i,Z_{i}^{\{r\}};\theta_0,\eta_0) \\ \!\!\!&& \!-  \lambda(\theta_0,\eta_0)^\top \sum_{i=1}^n \Psi(Z_i,Z_{i}^{\{r\}};\theta_0,\eta_0)\Psi(Z_i,Z_{i}^{\{r\}};\theta_0,\eta_0)^\top  \lambda(\theta_0,\eta_0)
+ o_{\mathbb{P}}(1)\\ \!\!\!&=& \!\! 2\ell_n(\theta_0,\eta_0)   + o_{\mathbb{P}}(1).
\end{eqnarray*}
Thus $2\ell_n(\theta_0,\widehat \eta_{\gamma_0})$ and $2\ell_n(\theta_0,\eta_0)$ have the same $\mathcal{X}_{d_\theta - 1}^2$ asymptotic distribution.

\section*{Acknowledgments}
Valentin Patilea gratefully acknowledges support from the Joint Research Initiative ‘Models and mathematical processing of very large data’ under the aegis of Risk Foundation, in partnership with MEDIAMETRIE and GENES, France, and from the grant of the Romanian Ministry of Education and Research, CNCS--UEFISCDI, project number PN-III-P4-ID-PCE-2020-1112, within PNCDI III.

\section*{Supplement}

The online Supplementary Material is organized as follows.
Supplement \ref{supp:sec_A} contains additional proofs.
Supplement  \ref{supp:sec_B} collects additional details on the simulation design and real data analysis results.

\bibliographystyle{imsart-nameyear}
\bibliography{biblio}

\newpage

\medskip

\newpage

\thispagestyle{plain}
\centerline{\LARGE \bf Wilks' theorem for semiparametric regressions}
\smallskip
\centerline{\LARGE\bf with weakly dependent data}
\bigskip
\centerline{\Large Supplementary Material}
\bigskip
\centerline{\large Marie Du Roy de Chaumaray$^*$, Matthieu Marbac$^*$ and Valentin Patilea$^*$}
\medskip
\centerline{\it $^*$Univ. Rennes, Ensai, CNRS, CREST - UMR 9194, F-35000 Rennes, France}
\bigskip


This supplement is organized as follows.
Supplement \ref{supp:sec_A} contains additional proofs.
Supplement  \ref{supp:sec_B} collects additional details on the simulation design and real data analysis results, that were omitted from the main paper due to page limit.

\setcounter{section}{0}
\renewcommand{\thesection}{\Alph{section}}%

\section{Additional proofs}\label{supp:sec_A}


\begin{proof}[Proof of Lemma~\ref{lem:equivalence}: equivalence of the moment conditions\\]
First, without imposing any identification condition, note that for any 
positive function $\omega(V_i)$, we have
$$
\mathbb E[g_\mu(Z_i;\gamma,m)\mid V_i,\mathcal F_{i-1}]=0
\Leftrightarrow
\mathbb E[g_\mu(Z_i;\gamma,m)\mathbb E[g_\mu(Z_i;\gamma,m)\mid V_i,\mathcal F_{i-1}]\omega(V_i)]=0.
$$
Indeed, $\mathbb E[g_\mu(Z_i;\gamma,m)\mid V_i,\mathcal F_{i-1}]=0$ directly implies that $$\mathbb E[g_\mu(Z_i;\gamma,m)\mathbb E[g_\mu(Z_i;\gamma,m)\mid V_i,\mathcal F_{i-1}]\omega(V_i)]=0.$$
 Conversely, by elementary properties of the conditional expectation, 
$$
\mathbb E[g_\mu(Z_i;\gamma,m)\mathbb E[g_\mu(Z_i;\gamma,m)\mid V_i,\mathcal F_{i-1}]\omega(V_i)]=
\mathbb E[\mathbb E^2[g_\mu(Z_i;\gamma,m)\mid V_i,\mathcal F_{i-1}]\omega(V_i)],
$$
and thus $\mathbb E[g_\mu(Z_i;\gamma,m)\mathbb E[g_\mu(Z_i;\gamma,m)\mid V_i,\mathcal F_{i-1}]\omega(V_i)]=0$ implies that $$\mathbb E[g_\mu(Z_i;\gamma,m)\mid V_i,\mathcal F_{i-1}]=0.$$

For an identifiable model, $(\gamma_0,m_0)$ is the unique solution for  $\mathbb E[g_\mu(Z_i;\gamma,m)\mid V_i,\mathcal F_{i-1}]=0$. Therefore, by \eqref{model:easy}, for any $(\gamma,m)\neq(\gamma_0,m_0)$, we have
$$
\mathbb E[\mathbb E^2[g_\mu(Z_i;\gamma,m)\mid V_i,\mathcal F_{i-1}]\omega(V_i)]>0.
$$
By \eqref{simeq1}, we have that $\gamma_0$ is the minimum of the map $\gamma\mapsto \mathbb E[\mathbb E[g_\mu(Z_i;\gamma,m_\gamma)\mid V_i,\mathcal F_{i-1}]^2\omega(V_i)]$. Thus, we have
$$
\nabla_\gamma\mathbb E[\mathbb E^2[g_\mu(Z_i;\gamma_0,m_0)\mid V_i,\mathcal F_{i-1}] \omega(V_i)]=0.
$$
By construction, $\nabla_\gamma g_\mu(Z_i;\gamma,m_\gamma)$ only depends on $V_i$, and thus interchanging derivative and expectation operators  we have 
\begin{equation*}
\nabla_\gamma\mathbb E[\mathbb E^2[g_\mu(Z_i;\gamma,m_\gamma)\mid V_i,\mathcal F_{i-1}]\omega(V_i)]=2
\mathbb E[g_\mu(Z_i;\gamma,m_\gamma)\nabla_\gamma g_\mu(Z_i;\gamma,m_\gamma) \omega(V_i)],
\end{equation*}
which leads to 
$$
\mathbb E[g_\mu(Z_i;\gamma_0,m_{\gamma_0})\nabla_\gamma g_\mu(Z_i;\gamma_0,m_{\gamma_0}) \omega(V_i)]=0.
$$
The proof is completed after left-multiplying both sides in the last display by the non-random matrix $\mathbf{J}(\gamma_0)$  and noting that the assumption made on $H_\mu(\gamma)$ ensures that $\gamma_0$  is the only one critical point for the map $\gamma\mapsto \mathbb E[\mathbb E^2[g_\mu(Z_i;\gamma,m_\gamma)\mid V_i,\mathcal F_{i-1}]\omega(V_i)]$.
\end{proof}

\quad


\begin{proof}[Proof of Lemma~\ref{lemma:TCLfi}]
By the Cramer-Wold device, it suffices to show that for any $c \in \mathbb{R}^d$, 
\begin{equation}\label{eq:TCLdim1}
\frac{1}{\sqrt{n}} \sum_{i=1}^n c^{\top}\Psi(Z_i,Z_{i}^{\{r\}};\theta_0,\eta_0) \xrightarrow{d} \mathcal{N}(0, c^{\top} \Sigma c).
\end{equation}
As $\left(c^\top \Psi(Z_i,Z_{i}^{\{r\}};\theta_0,\eta_0)\right)$ is a strictly stationary, $\alpha$-mixing, centered process, we notice that the Central Limit Theorem follows by a direct application of Corollary 1.1 in \citep{Merlevede00}, under Assumption \ref{ass:process}. 
Indeed, let $\delta=s-2 >0$, where $s$ is given in Assumption \ref{ass:process}. By Cauchy-Schwarz inequality and \eqref{eq:ex_mom}, as $2(2+\delta)=2s$, we obtain that $\mathbb{E}\left[\|\Psi(Z_i,Z_{i}^{\{r\}};\theta_0,\eta_0)\|^{2+\delta}\right] <  \infty$. 
If  $\widetilde{\alpha}_m$ denotes the mixing coefficients of the process $\left(c^\top \Psi(Z_i,Z_{i}^{\{r\}};\theta_0,\eta_0)\right)$, we have $\widetilde{\alpha}_m \leq \alpha_m$ where, by \eqref{eq:def_alpha_ew},  $m \, \alpha_m^{ \delta/(2+\delta)} \to 0$  as $ \delta/(2+\delta)= s/(s-2) < \xi$.
To obtain \eqref{eq:TCLdim1}, it remains to check 
\begin{equation}\label{eq:cvcov}
\frac{1}{n}  \mathbb{E}\left[\left(\sum_{i=1}^n c^{\top}\Psi(Z_i,Z_{i}^{\{r\}};\theta_0,\eta_0)\right)^2\right] \to c^{\top} \Sigma c.
\end{equation}
Since $(Z_i)$ is stationary and, by construction, $\mathbb{E}\left[\Psi(Z_j,Z_{j}^{\{r\}};\theta_0,\eta_0) | V_j, \mathcal{F}_{j-1} \right]= 0$ a.s., we have 
\begin{multline*}
\mathbb{E}
\left[\left(\sum_{i=1}^n \textcolor{black} {c^{\top}} \Psi(Z_i,Z_{i}^{\{r\}};\theta_0,\eta_0)\right)^2 \right]  =  \sum_{i=1}^n \color{black}  c^{\top} \color{black} \mathbb{E}\left[\Psi(Z_i,Z_{i}^{\{r\}};\theta_0,\eta_0)\Psi(Z_i,Z_{i}^{\{r\}};\theta_0,\eta_0)^{\top}\right]   \color{black}  c \color{black}  \\   + 2  \sum_{1\leq i<j \leq n}  \color{black}  c^{\top} \color{black}  \mathbb{E}\left[\Psi(Z_i,Z_{i}^{\{r\}};\theta_0,\eta_0) \Psi(Z_j,Z_{j}^{\{r\}};\theta_0,\eta_0)^\top\right]  \color{black}  c \color{black}  \\
 = n \,  \color{black}  c^{\top} \color{black}  \mathbb{E}\left[\Psi(Z_1,Z_{1}^{\{r\}};\theta_0,\eta_0)\Psi(Z_1,Z_{1}^{\{r\}};\theta_0,\eta_0)^{\top}\right]   \color{black}  c \color{black} 
= n \,  \color{black}  c^{\top} \color{black}  \Sigma  \color{black}  c \color{black} .
\end{multline*}
\end{proof}

\quad

\begin{proof}[Proof of Lemma~\ref{lemma:max}]
The property 
$\mathbb{E}\left[ \left\| \Psi(Z_i,Z_{i}^{\{r\}};\theta_0,\eta_0)\right\|^3\right] < \infty$ follows by applying Cauchy-Schwarz inequality component-wise and using our moment conditions. 
Let $M_n=\max_{1\leq i\leq n} \| \Psi(Z_i,Z_{i}^{\{r\}};\theta_0,\eta_0) \|$ and $C>0$.
By Boole's inequality, the stationarity of the process $(Z_i)$ and Markov inequality, we have
\begin{multline*}
n^{1/2}\mathbb P (M_n > Cn^{1/2}) \leq n^{3/2} \mathbb E[\| \Psi(Z_i,Z_{i}^{\{r\}};\theta_0,\eta_0)\|^3]/(Cn^{1/2})^3 \\= C^{-3}E[\| \Psi(Z_i,Z_{i}^{\{r\}};\theta_0,\eta_0)\|^3]<\infty.
\end{multline*}
Therefore, we have
$
M_n= o_\mathbb{P} (n^{1/2}).
$
Moreover, we have
$$
\frac{1}{n} \sum_{i=1}^n  \left\| \Psi(Z_i,Z_{i}^{\{r\}};\theta_0,\eta_0) \right\|^3 \leq M_n \frac{1}{n} \sum_{i=1}^n  \left\| \Psi(Z_i,Z_{i}^{\{r\}};\theta_0,\eta_0) \right\|^2
$$
Using the fact that $ M_n = o_\mathbb{P} (n^{1/2})$,  $\mathbb{E}\left[ \left\| \Psi(Z_i,Z_{i}^{\{r\}};\theta_0,\eta_0)\right\|^2 \right] < \infty$ and by Lemma~\ref{lemma:TCLfi}, we have
$$
\frac{1}{n} \sum_{i=1}^n  \left\| \Psi(Z_i,Z_{i}^{\{r\}};\theta_0,\eta_0) \right\|^3 = o_\mathbb{P}(n^{1/2}).
$$
 \end{proof}

\quad

\begin{proof}[Proof of Lemma~\ref{lemma:lambda}]
For any $\theta$, we have that $\lambda(\gamma,\eta)$ satisfies
\begin{equation} \label{eq:deflambda}
\frac{1}{n} \sum_{i=1}^n \frac{1}{1 + \lambda(\theta,\eta)^\top \Psi(Z_i,Z_{i}^{\{r\}};\theta,\eta)} \Psi(Z_i,Z_{i}^{\{r\}};\theta,\eta) = 0,
\end{equation}
Let $ \lambda(\theta_0,\eta_0)=\| \lambda(\theta_0,\eta_0)\| u$, we want to show that $\| \lambda(\theta_0,\eta_0)\|=O_\mathbb{P}(n^{-1/2})$. Noting that \begin{multline*}
\{1 + \lambda(\theta,\eta)^\top \Psi(Z_i,Z_{i}^{\{r\}};\theta,\eta)\}^{-1}=\\1 - \lambda(\theta,\eta)^\top \Psi(Z_i,Z_{i}^{\{r\}};\theta,\eta)/\{1 + \lambda(\theta,\eta)^\top \Psi(Z_i,Z_{i}^{\{r\}};\theta,\eta)\},
\end{multline*}
we have from \eqref{eq:deflambda}
\begin{equation}\label{eq:lambdanorm}
\| \lambda(\theta_0,\eta_0)\| u^\top \widetilde S(\theta_0,\eta_0) u = \frac{1}{n}  u^\top \sum_{i=1}^n \Psi(Z_i,Z_{i}^{\{r\}};\theta_0,\eta_0) ,
\end{equation}
where 
\begin{multline*}
\widetilde S(\theta_0,\eta_0) = \frac{1}{n} \sum_{i=1}^n \{1 +  \lambda(\theta_0,\eta_0)^\top \Psi(Z_i,Z_{i}^{\{r\}};\theta_0,\eta_0)\}^{-1} \times \\ \Psi(Z_i,Z_{i}^{\{r\}};\theta_0,\eta_0)\Psi(Z_i,Z_{i}^{\{r\}};\theta_0,\eta_0)^\top.
\end{multline*}
By construction, for any $1\leq i\leq n$,   $  \lambda(\theta_0,\eta_0)^\top \Psi(Z_i,Z_{i}^{\{r\}};\theta_0,\eta_0)+1>0$. Thus, we obtain that
$$
\| \lambda(\theta_0,\eta_0)\| u^\top   S(\theta_0,\eta_0) u \leq \| \lambda(\theta_0,\eta_0)\| u^\top \widetilde S(\theta_0,\eta_0) u (1 + \| \lambda(\theta_0,\eta_0)\| M_n),
$$
where $S(\theta_0,\eta_0)=n^{-1}\sum_{i=1}^n \Psi(Z_i,Z_{i}^{\{r\}};\theta_0,\eta_0)\Psi(Z_i,Z_{i}^{\{r\}};\theta_0,\eta_0)^\top$ and where $M_n $ is the largest value among the $ \| \Psi(Z_i,Z_{i}^{\{r\}};\theta_0,\eta_0) \|$'s. Using \eqref{eq:lambdanorm} we deduce
\begin{multline*}
\| \lambda(\theta_0,\eta_0)\| \left[u^\top   S(\theta_0,\eta_0) u - M_n  u^\top  \frac{1}{n}\sum_{i=1}^n \Psi(Z_i,Z_{i}^{\{r\}};\theta_0,\eta_0) \right]   \leq \\ u^\top  \frac{1}{n}\sum_{i=1}^n \Psi(Z_i,Z_{i}^{\{r\}};\theta_0,\eta_0).
\end{multline*}
Lemma~\ref{lemma:max} implies that $M_n = o_\mathbb{P}(n^{1/2})$ and Lemma~\ref{lemma:TCLfi} allows to upper-bound the right side of the previous inequality by $O_\mathbb{P}(n^{-1/2})$. Moreover, we have 
$
\nu + o_\mathbb{P}(1) \leq u^\top S(\theta_0,\eta_0) u,
$
where $\nu>0$ is the smallest eigenvalue of $\Sigma$ defined in Assumption \ref{variance_s}. Therefore, we have
$$
\| \lambda(\theta_0,\eta_0)\| \left(\nu + o_\mathbb{P}(1) -  o_\mathbb{P}(n^{1/2})   O_\mathbb{P}(n^{-1/2}) \right)   \leq    O_\mathbb{P}(n^{-1/2}),
$$
which implies 
$
\| \lambda(\theta_0,\eta_0)\| =    O_\mathbb{P}(n^{-1/2}).
$
Noting that 
$$
n\pi_i(\theta_0,\eta_0) = 1 -  \lambda(\theta_0,\eta_0)^\top \Psi(Z_i,Z_{i}^{\{r\}};\theta_0,\eta_0)+
\frac{\left\{ \lambda(\theta_0,\eta_0)^\top \Psi(Z_i,Z_{i}^{\{r\}};\theta_0,\eta_0)\right\}^2}{1 +  \lambda(\theta_0,\eta_0)^\top \Psi(Z_i,Z_{i}^{\{r\}};\theta_0,\eta_0)},
$$
 we have from \eqref{eq:deflambda} that
\begin{multline*}
\frac{1}{n} \sum_{i=1}^n \Psi(Z_i,Z_{i}^{\{r\}};\theta_0,\eta_0) - S(\theta_0,\eta_0) \lambda(\theta_0,\eta_0) \\ + \frac{1}{n} \sum_{i=1}^n \frac{\left\{ \lambda(\theta_0,\eta_0)^\top \Psi(Z_i,Z_{i}^{\{r\}};\theta_0,\eta_0)\right\}^2}{1 +  \lambda(\theta_0,\eta_0)^\top \Psi(Z_i,Z_{i}^{\{r\}};\theta_0,\eta_0)} \Psi(Z_i,Z_{i}^{\{r\}};\theta_0,\eta_0) = 0.
\end{multline*}
Using Lemma~\ref{lemma:max}, we deduce $\max_{1\leq i\leq n} \{1 + \lambda(\theta_0,\eta_0)^\top \Psi(Z_i,Z_{i}^{\{r\}};\theta_0,\eta_0)\}^{-1} = 1 + o_\mathbb{P}(1)$, and 
the norm of the second sum on the left side of the last display can be bounded by
$$
\frac{1}{n} \sum_{i=1}^n \frac{\|\Psi(Z_i,Z_{i}^{\{r\}};\theta_0,\eta_0)\|^3 \| \lambda(\theta_0,\eta_0)\|^2}{1 +  \lambda(\theta_0,\eta_0)^\top \Psi(Z_i,Z_{i}^{\{r\}};\theta_0,\eta_0)} = o_\mathbb{P}(n^{1/2})O_\mathbb{P}(n^{-1})O_\mathbb{P}(1)=o_\mathbb{P}(n^{-1/2}).
$$
Thus
$$
 \lambda(\theta_0,\eta_0)=S(\theta_0,\eta_0)^{-1}\frac{1}{n} \sum_{i=1}^n \Psi(Z_i,Z_{i}^{\{r\}};\theta_0,\eta_0)+ o_\mathbb{P}(n^{-1/2}).
$$
\end{proof}

\quad

\begin{proof}[Proof of Lemma \ref{lemma:A_X}]
We decompose $
A_X= A_{X,0} +A_{X,a} + A_{X,b} +A_{X,c} + A_{X,d}+ A_{X,1} + A_{X,2}\\
$
where 
$$
A_{X,0}=\frac{1}{n} \sum_{i=1}^n \varepsilon_i \left[  \nabla_{\gamma_1} l(X_i;\gamma_{0,1}) -  \frac{ \eta_{\gamma_0,X}(\Upsilon_i)}{\eta_{\gamma_0,f}(\Upsilon_i)} \right]\eta^4_{\gamma_0,f}(\Upsilon_i) ,
$$
$$
A_{X,a} = \frac{1}{n} \sum_{i=1}^n \varepsilon_i  \left[\widehat\eta_{\gamma_0,f}(\Upsilon_i)  \nabla_{\gamma_1} l(X_i;\gamma_{0,1}) -   \eta_{\gamma_0,X}(\Upsilon_i)  \right]\widehat\eta^3_{\gamma_0,f}(\Upsilon_i),
$$
$$
A_{X,b}=\frac{1}{n} \sum_{i=1}^n  \varepsilon_i \left[ \eta_{\gamma_0,X}(\Upsilon_i)  -  \widehat \eta_{\gamma_0,X}(\Upsilon_i)  \right] \widehat\eta^3_{\gamma_0,f}(\Upsilon_i),
$$
\begin{multline*}
A_{X,c}= \frac{1}{n} \sum_{i=1}^n  \left[ m_0(\Upsilon_i) \eta_{\gamma_0,f}(\Upsilon_i)- \widehat \eta_{\gamma_0,m}(\Upsilon_i)\right]  \times \\
\left[\eta_{\gamma_0,f}(\Upsilon_i)  \nabla_{\gamma_1} l(X_i;\gamma_{0,1})  - \widehat  \eta_{\gamma_0,X}(\Upsilon_i)  \right] \widehat\eta^2_{\gamma_0,f}(\Upsilon_i),
\end{multline*}
\begin{multline*}
A_{X,d}= \frac{1}{n} \sum_{i=1}^n m_0(\Upsilon_i)  \left[  \widehat\eta_{\gamma_0,f}(\Upsilon_i)-  \eta_{\gamma_0,f}(\Upsilon_i)\right]  \times \\ \left[\eta_{\gamma_0,f}(\Upsilon_i)  \nabla_{\gamma_1} l(X_i;\gamma_{0,1})  - \widehat  \eta_{\gamma_0,X}(\Upsilon_i)  \right] \widehat\eta^2_{\gamma_0,f}(\Upsilon_i),
\end{multline*}
and 
\begin{multline*}
A_{X,1}= \frac{1}{n} \sum_{i=1}^n  \left[ m_0(\Upsilon_i) \eta_{\gamma_0,f}(\Upsilon_i)- \widehat \eta_{\gamma_0,m}(\Upsilon_i)\right] \times \\  \left[\widehat\eta_{\gamma_0,f}(\Upsilon_i)  -   \eta_{\gamma_0,f}(\Upsilon_i) \right]  \nabla_{\gamma_1} l(X_i;\gamma_{0,1}) \widehat\eta^2_{\gamma_0,f}(\Upsilon_i),
\end{multline*}
\begin{multline*}
A_{X,2}=\frac{1}{n} \sum_{i=1}^n   m_0(\Upsilon_i) \left[ \widehat\eta_{\gamma_0,f}(\Upsilon_i)- \eta_{\gamma_0,f}(\Upsilon_i)\right] \times \\\left[\widehat\eta_{\gamma_0,f}(\Upsilon_i)  - \eta_{\gamma_0,f}(\Upsilon_i) \right] \nabla_{\gamma_1} l(X_i;\gamma_{0,1}) \widehat\eta^2_{\gamma_0,f}(\Upsilon_i).
\end{multline*}
After taking the norm of the sums, by triangle inequality and \eqref{eq:boundsnp}, 
$$
 \| A_{X,1}\| + \| A_{X,2}\| = \left[ O_\mathbb{P}\left( \left(\frac{\ln n}{nh} \right)^{1/2} \right) + O_\mathbb{P}\left( h^2 \right)  \right]^2 \times O_\mathbb{P}(1) = o_\mathbb{P}(n^{-1/2}).
 $$ 
The terms $A_{X,a}$ to $A_{X,d}$ require a more refined treatment. We can write 
\begin{multline}\label{replace_h}
A_{X,a} = \frac{1}{n} \sum_{i=1}^n \varepsilon_i  \left[\widehat\eta_{\gamma_0,f}(\Upsilon_i)  \nabla_{\gamma_1} l(X_i;\gamma_{0,1}) -   \eta_{\gamma_0,X}(\Upsilon_i)  \right]\widehat\eta^3_{\gamma_0,f}(\Upsilon_i)\\
=  \frac{1}{n} \sum_{i=1}^n \varepsilon_i  \left[\widehat\eta_{\gamma_0,f}(\Upsilon_i)  \nabla_{\gamma_1} l(X_i;\gamma_{0,1}) -   \eta_{\gamma_0,X}(\Upsilon_i)  \right]\eta^3_{\gamma_0,f}(\Upsilon_i) + r_{X,a},
\end{multline}
with $\|r_{X,a}\|=o_\mathbb{P}(n^{-1/2})$. The rate of the negligible reminder $r_{A,a}$ is again obtained after taking the norm of the sums, by \eqref{eq:boundsnp}. Finally, using the same arguments like for bounding $A_{W,b}$ in Lemma \ref{lemma:A_W}, we also obtain $ \| A_{X,a}\| + \| A_{X,b}\| = o_\mathbb{P}(n^{-1/2}).$ 
To bound $A_{X,c}$ and  $A_{X,d}$, first we replace $\widehat\eta^3_{\gamma_0,f}(\Upsilon_i)$ by $\eta^3_{\gamma_0,f}(\Upsilon_i)$, as in \eqref{replace_h}. Next, we use the  the same arguments like for bounding $A_{W,c}$ and $A_{W,d}$ in Lemma \ref{lemma:A_W},  and we obtain $ \| A_{X,c}\| + \| A_{X,d}\| = o_\mathbb{P}(n^{-1/2}).$  \end{proof}

\section{Additional empirical evidence}\label{supp:sec_B}

\subsection{Additional results on the real data application} \label{supp:appli}

This section presents additional results on the real data application. Table~\ref{tab:autocorrelations} presents the autocorrelations of the different variables. Figures~\ref{fig:ozone}-\ref{fig:pm10}  present the original series and the series obtained by removing the seasonality.  Figure~\ref{fig:density} and Figure~\ref{fig:m} present the estimated density of the index and the estimated function $\hat m(\cdot)$, obtained with the $Lag(2)$ setup. 

 \begin{table}
\caption{Empirical autocorrelations of the variables computed for the learning and testing sample.\label{tab:autocorrelations}}
\centering
  \begin{tabular}{ccc}
 \hline 
 Variable & learning sample & testing sample\\
 \hline 
 $o3$ & 0.469 & 0.450 \\ 
 $rhum$ & 0.410 & 0.344 \\ 
 $temp$ & 0.735 & 0.705 \\ 
 $dptp$ & 0.609 & 0.594 \\ 
 $pm$ & 0.388 & 0.427 \\ 
 \hline 
 \end{tabular}
\end{table}

\begin{figure}[ht!]
  \begin{center}
    \subfloat[Original data]{
      \includegraphics[width=0.35\textwidth]{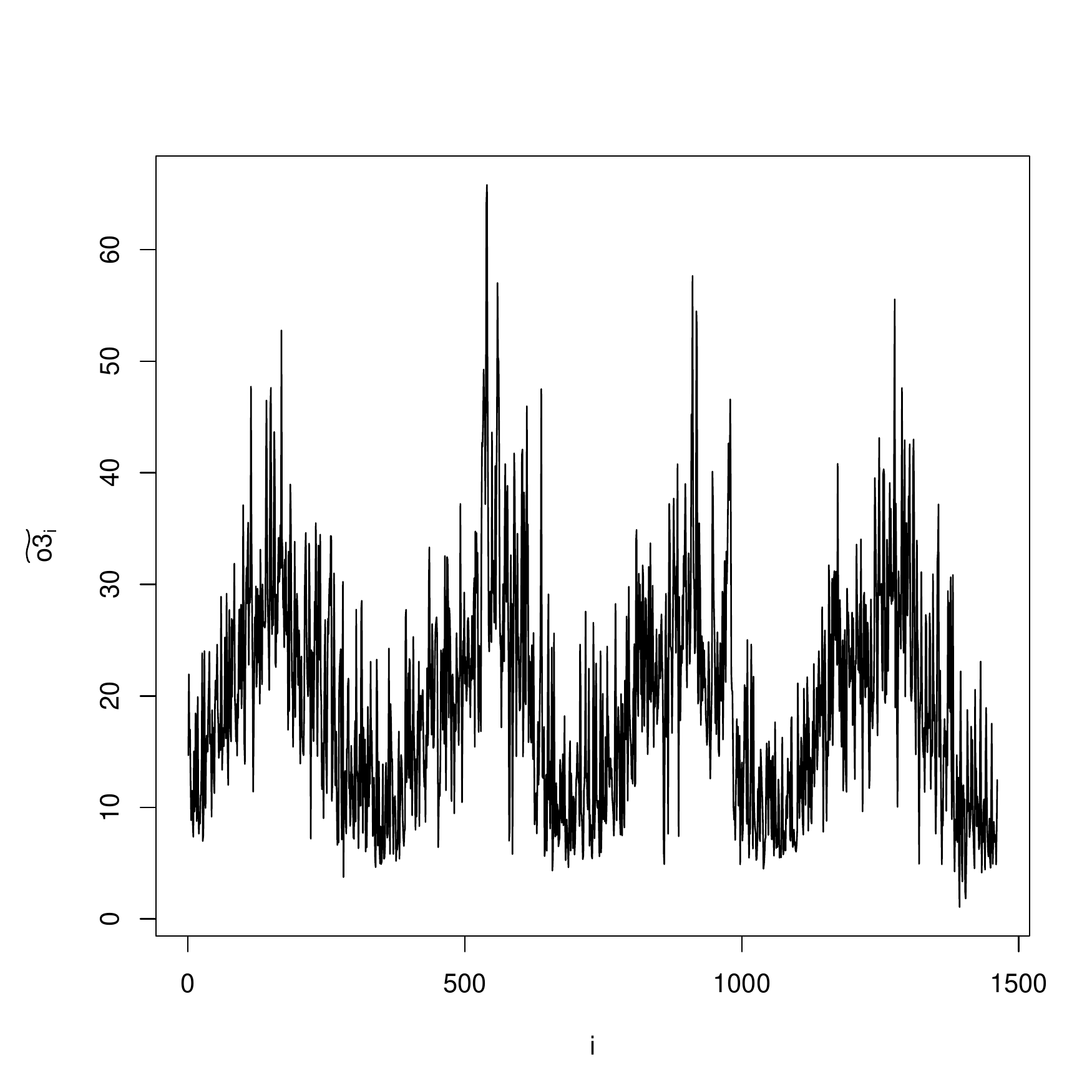}
     
                         }
    \subfloat[Deseasonalized data]{
      \includegraphics[width=0.35\textwidth]{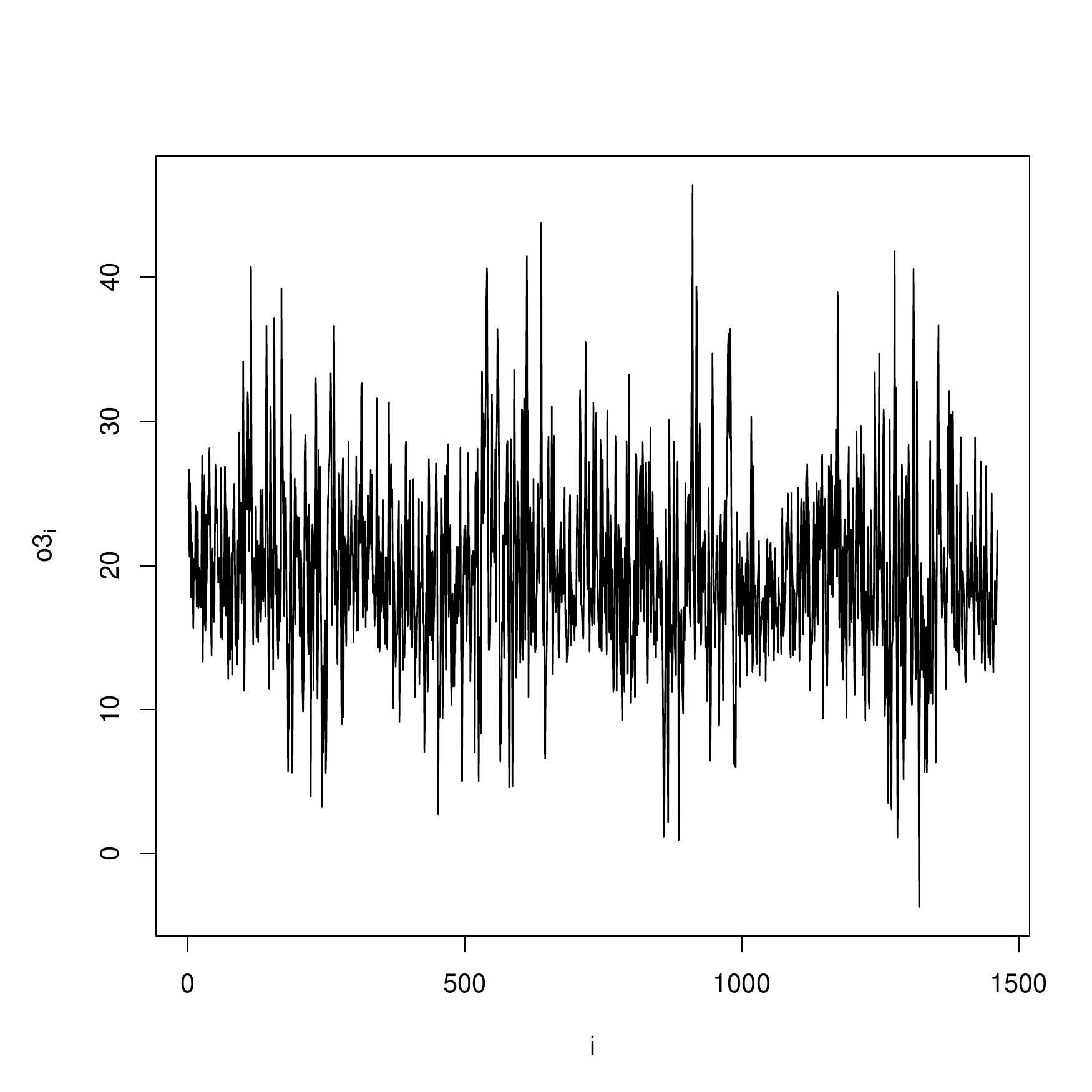}
       
                         }
    \caption{Series of the daily mean ozone level collected in Chicago in 1994-1997.}
    \label{fig:ozone}
  \end{center}
\end{figure}

\begin{figure}[ht!]
  \begin{center}
    \subfloat[Original data]{
      \includegraphics[width=0.35\textwidth]{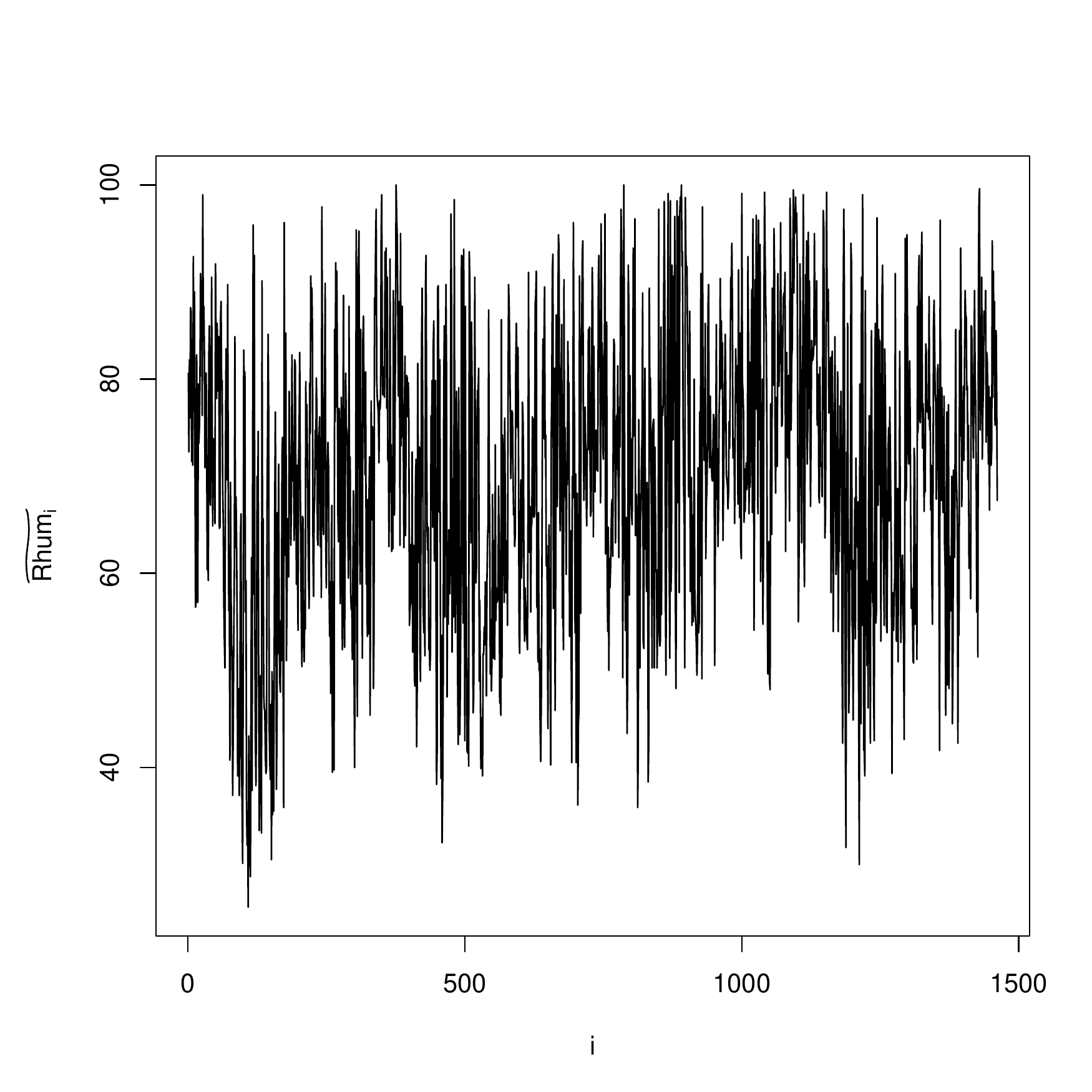}
     
                         }
    \subfloat[Deseasonalized data]{
      \includegraphics[width=0.35\textwidth]{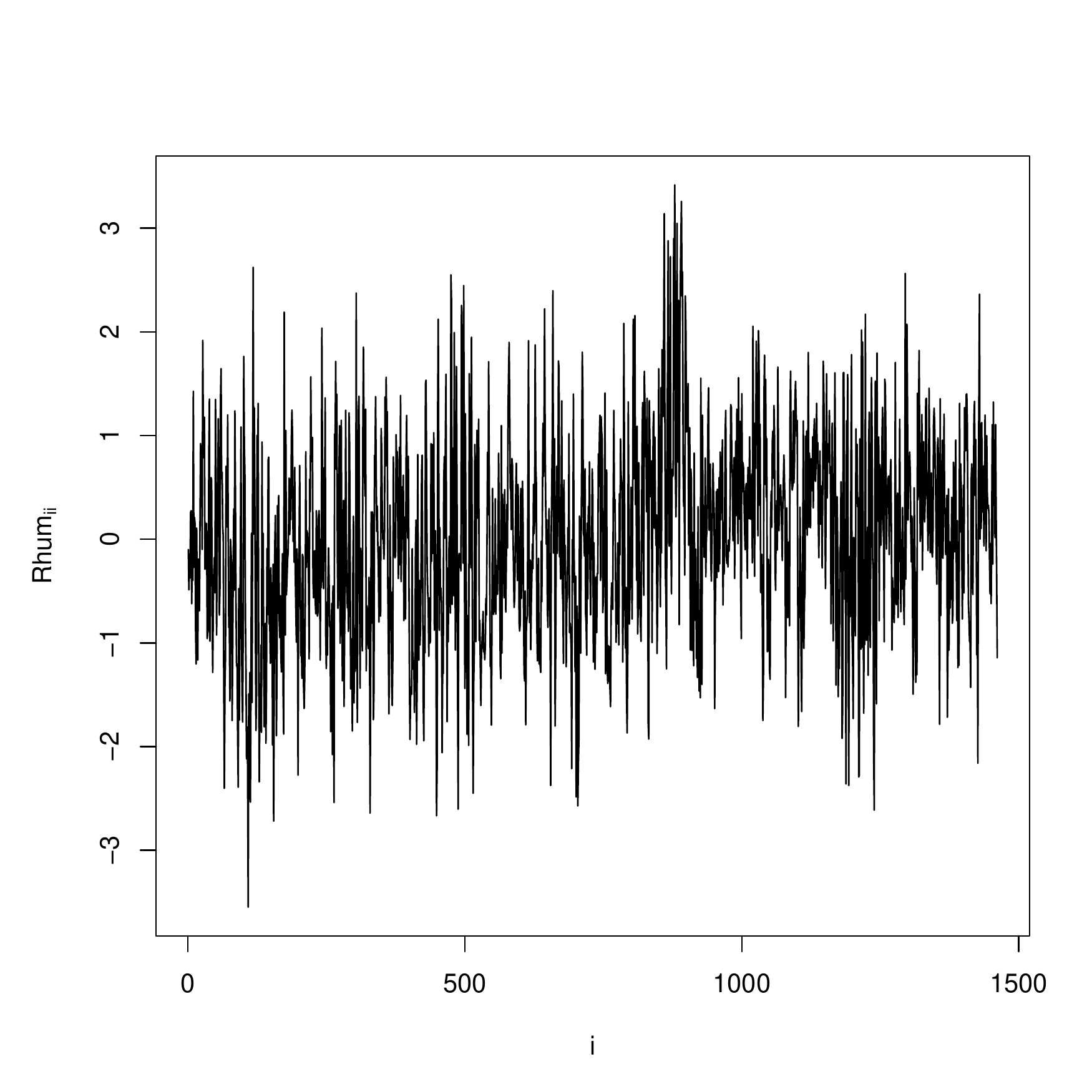}
       
                         }
    \caption{Series of the daily relative humidity level collected in Chicago in 1994-1997.}
    \label{fig:rhum}
  \end{center}
\end{figure}

 \begin{figure}[ht!]
  \begin{center}
    \subfloat[Original data]{
      \includegraphics[width=0.35\textwidth]{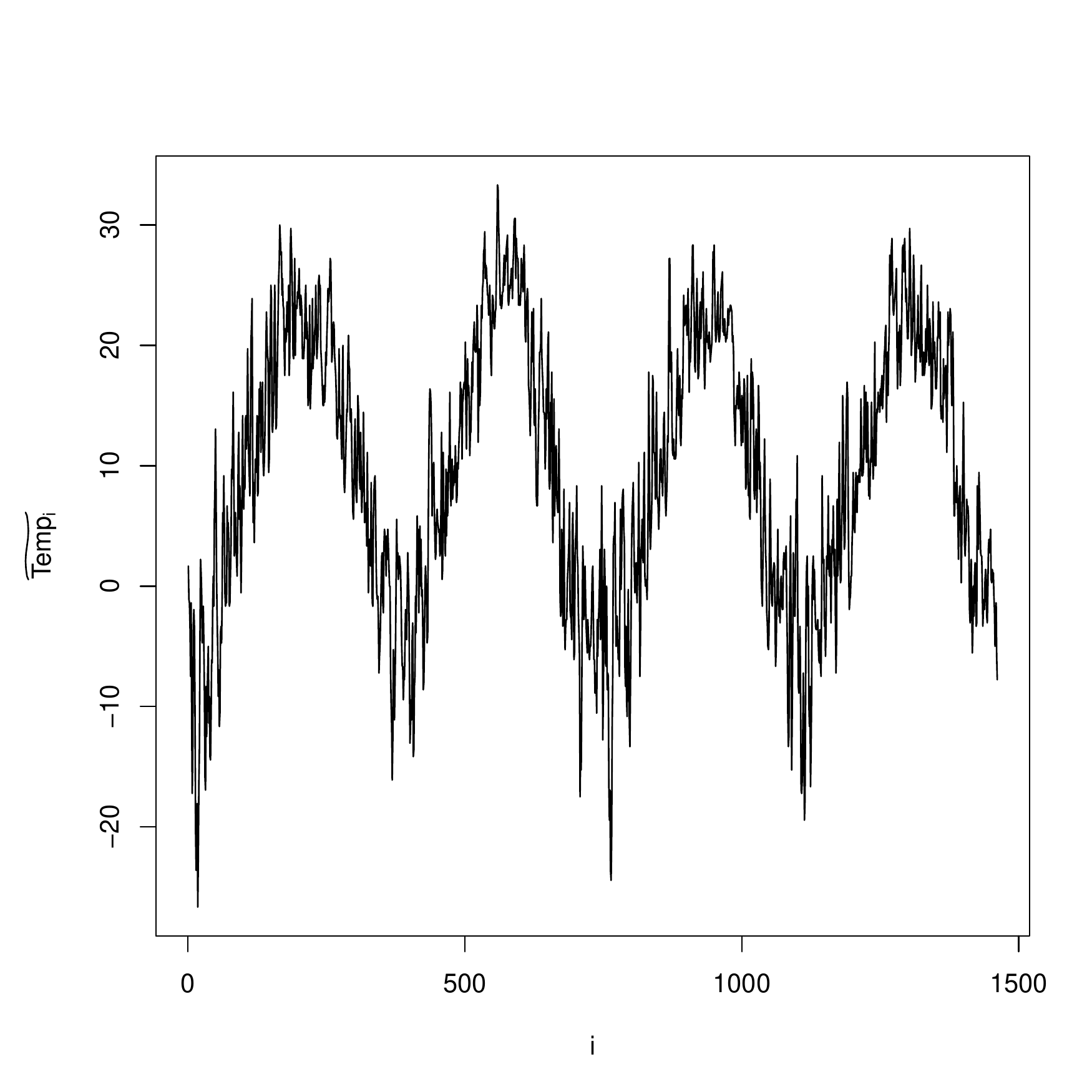}
     
                         }
    \subfloat[Deseasonalized data]{
      \includegraphics[width=0.35\textwidth]{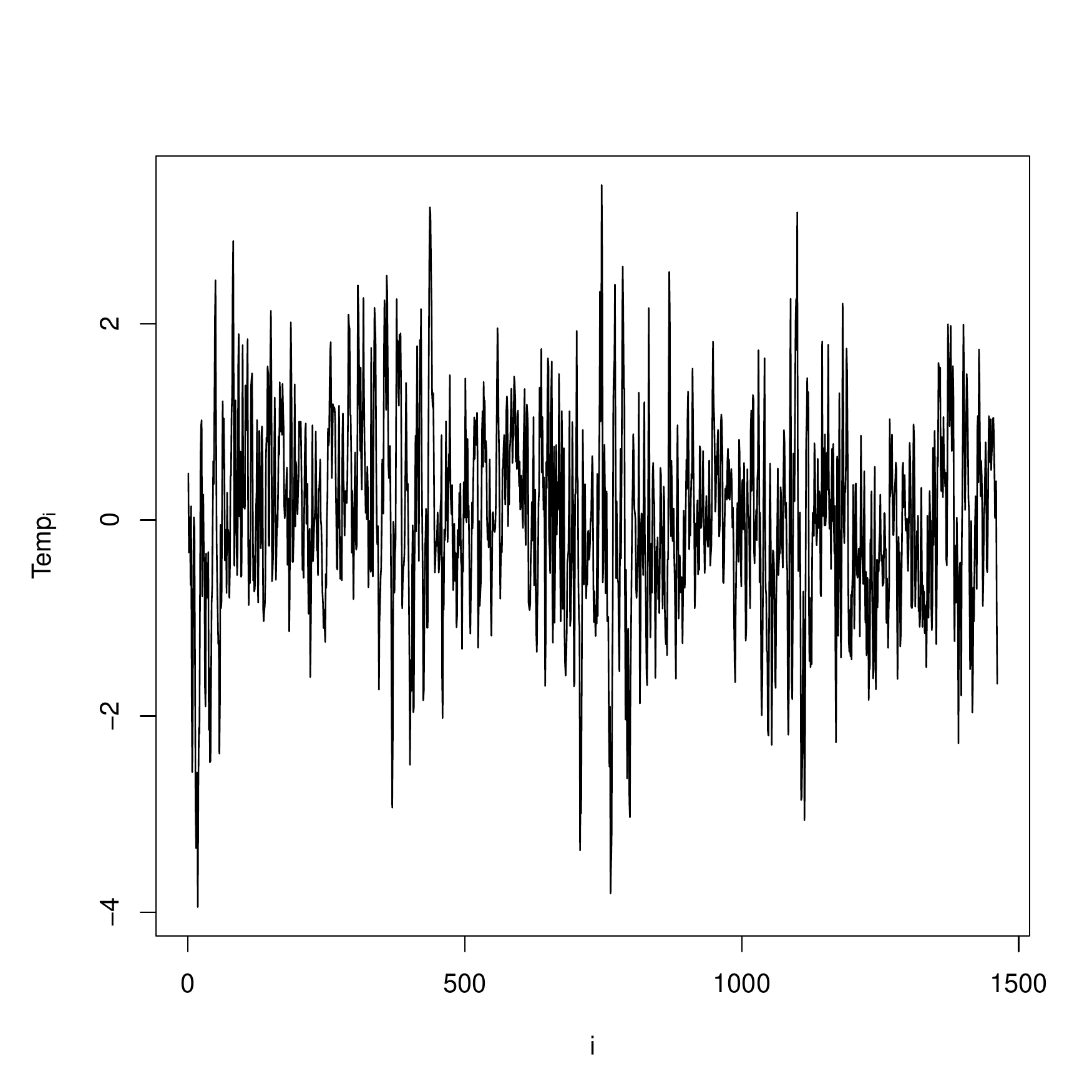}
       
                         }
    \caption{Series of the daily mean temperature collected in Chicago in 1994-1997.}
    \label{fig:temp}
  \end{center}
\end{figure}

 \begin{figure}[ht!]
  \begin{center}
    \subfloat[Original data]{
      \includegraphics[width=0.35\textwidth]{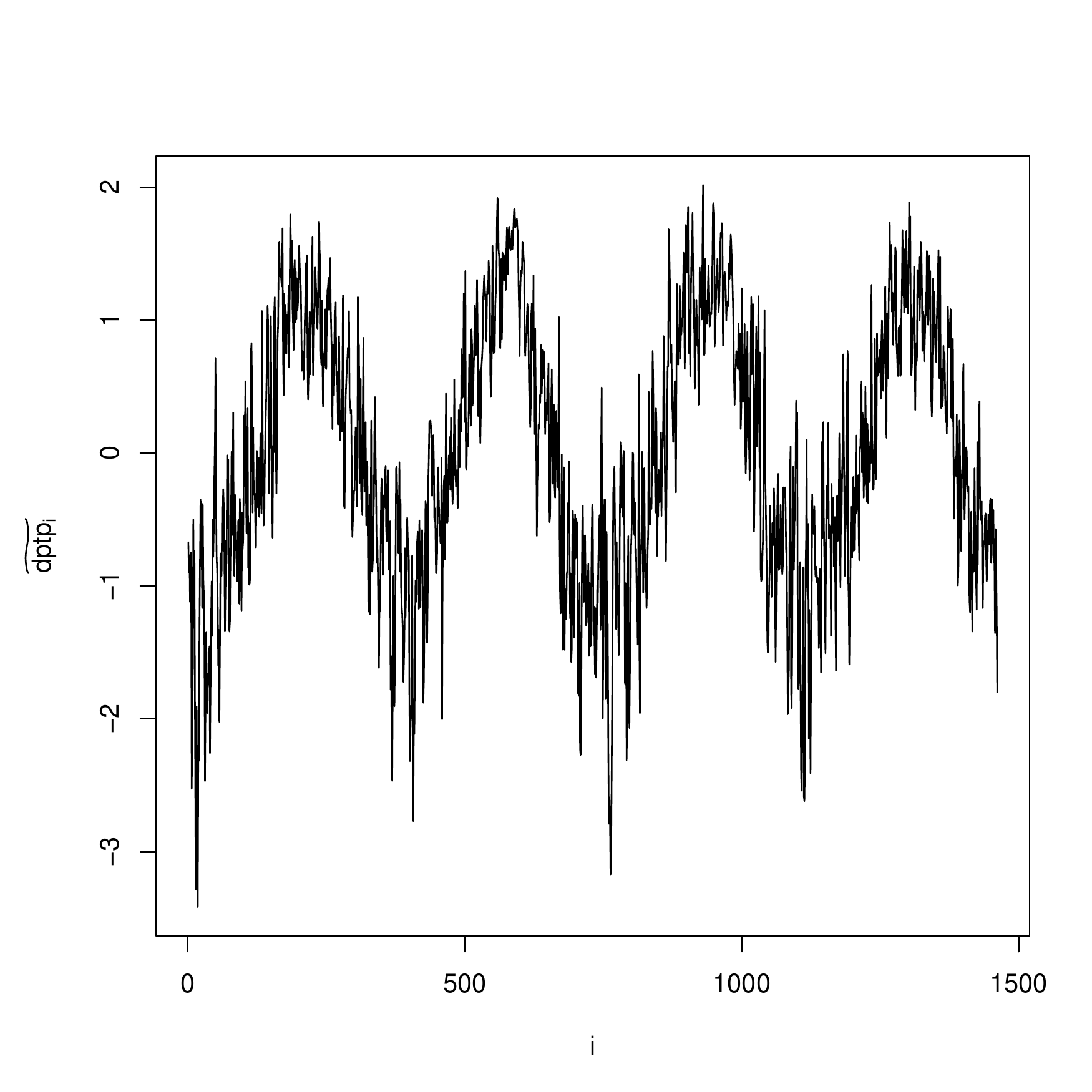}
     
                         }
    \subfloat[Deseasonalized data]{
      \includegraphics[width=0.35\textwidth]{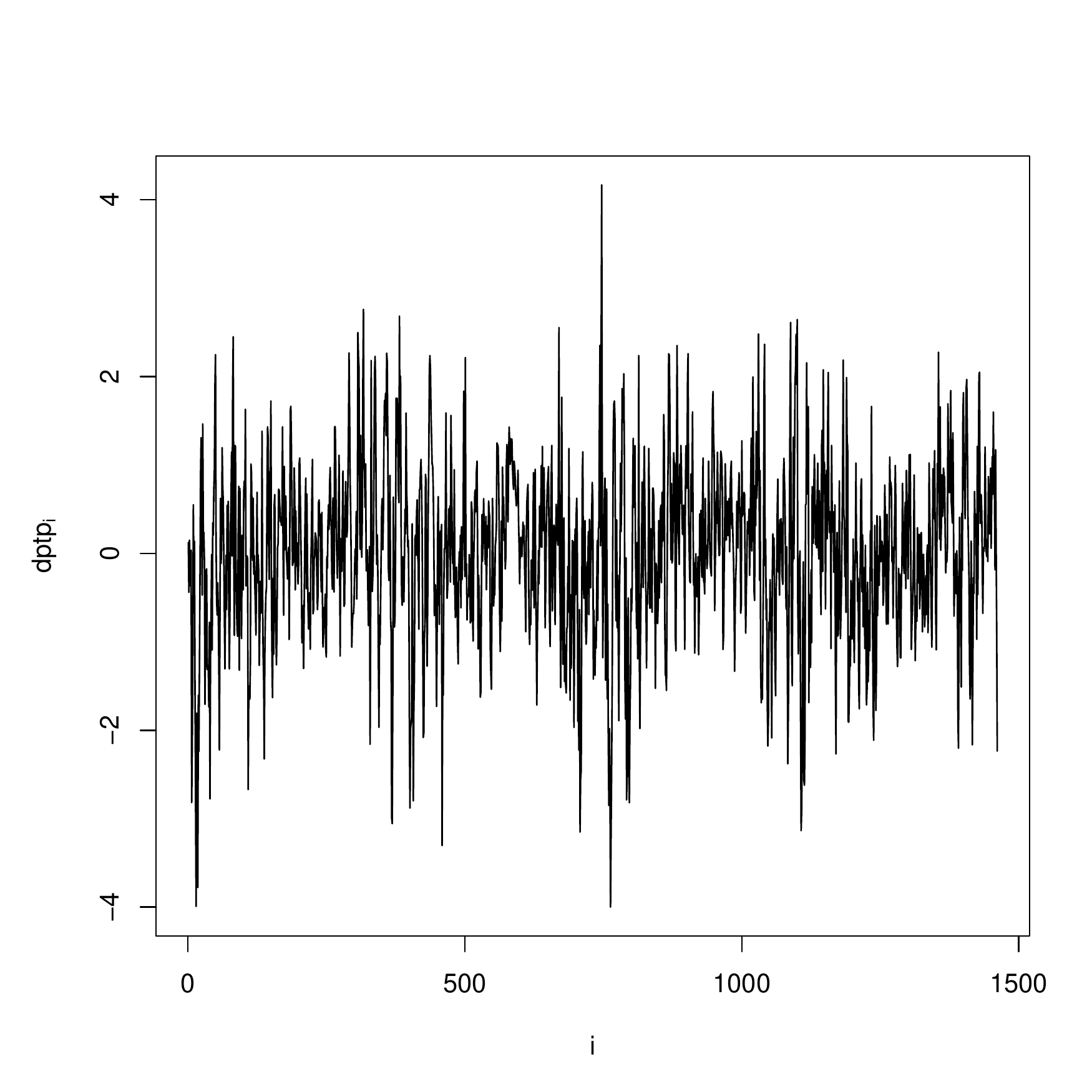}
       
                         }
    \caption{Series of the daily dew point temperature collected in Chicago in 1994-1997.}
    \label{fig:dptp}
  \end{center}
\end{figure}

 \begin{figure}[ht!]
  \begin{center}
    \subfloat[Original data]{
      \includegraphics[width=0.35\textwidth]{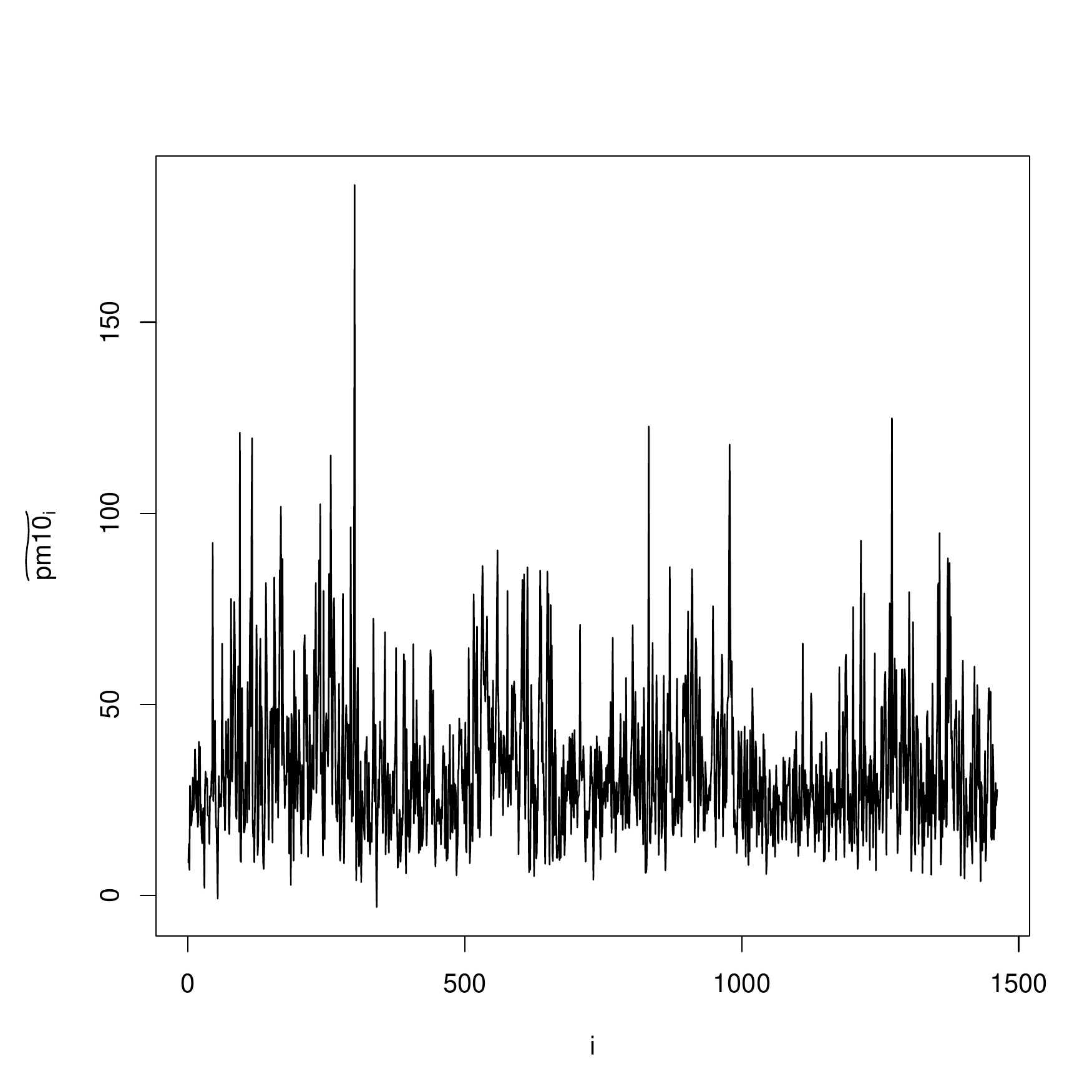}
     
                         }
    \subfloat[Deseasonalized data]{
      \includegraphics[width=0.35\textwidth]{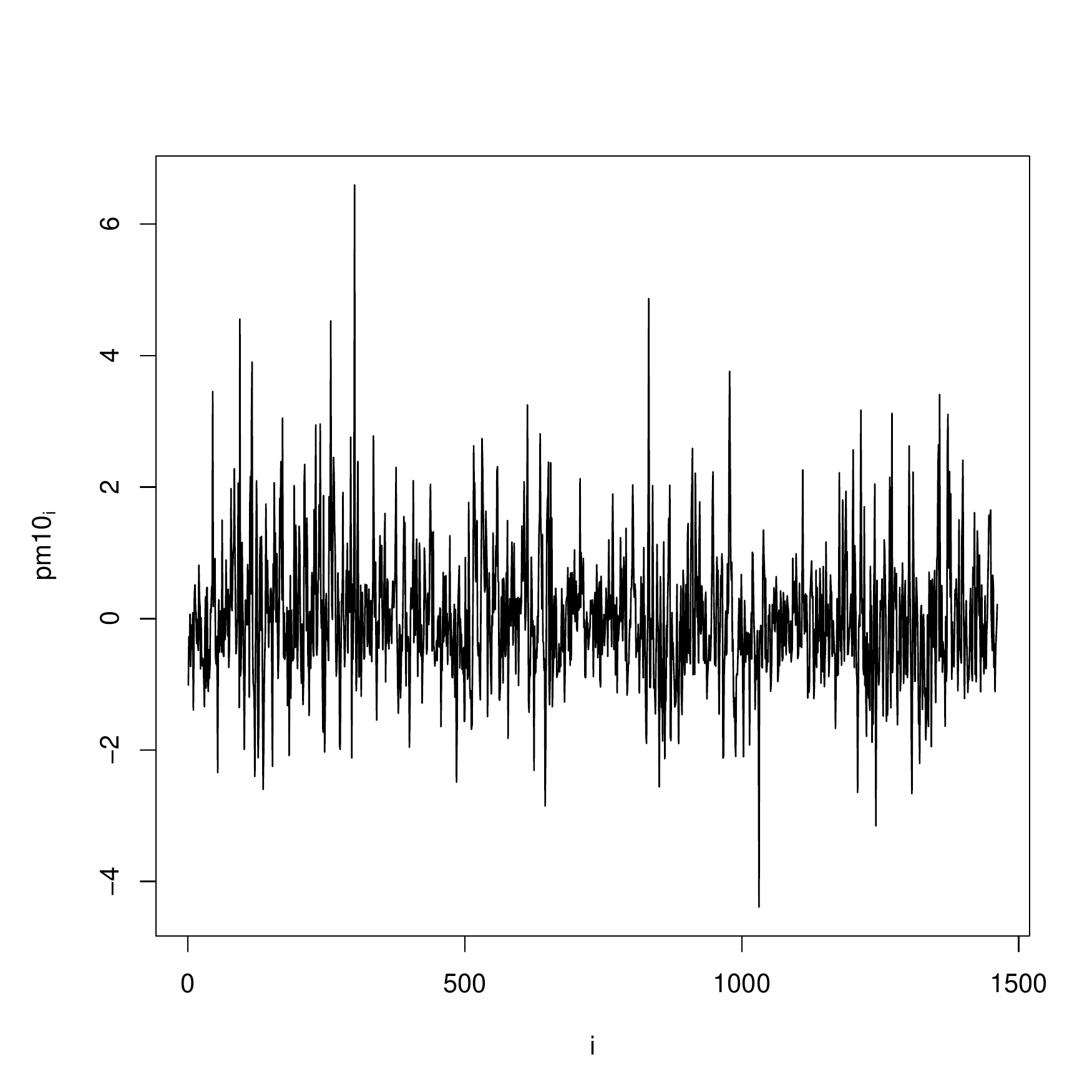}
       
                         }
    \caption{Series of the daily PM10-level collected in Chicago in 1994-1997.}
    \label{fig:pm10}
  \end{center}
\end{figure}

\begin{figure}[ht!]
\centering \includegraphics[scale=0.35]{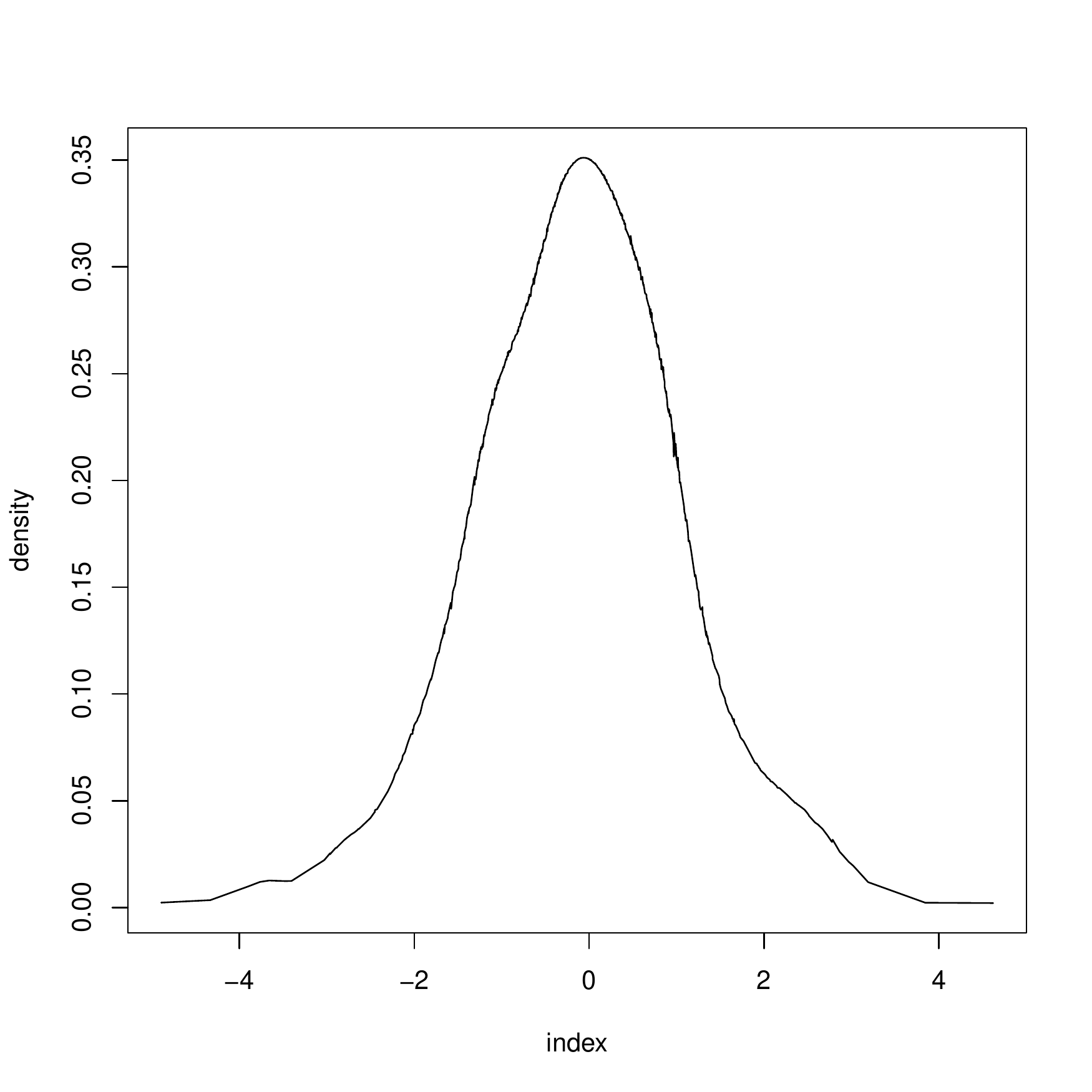}
\caption{Density of the index obtained on the testing sample.} \label{fig:density}
\end{figure} 
 
\begin{figure}[ht!]
\centering \includegraphics[scale=0.35]{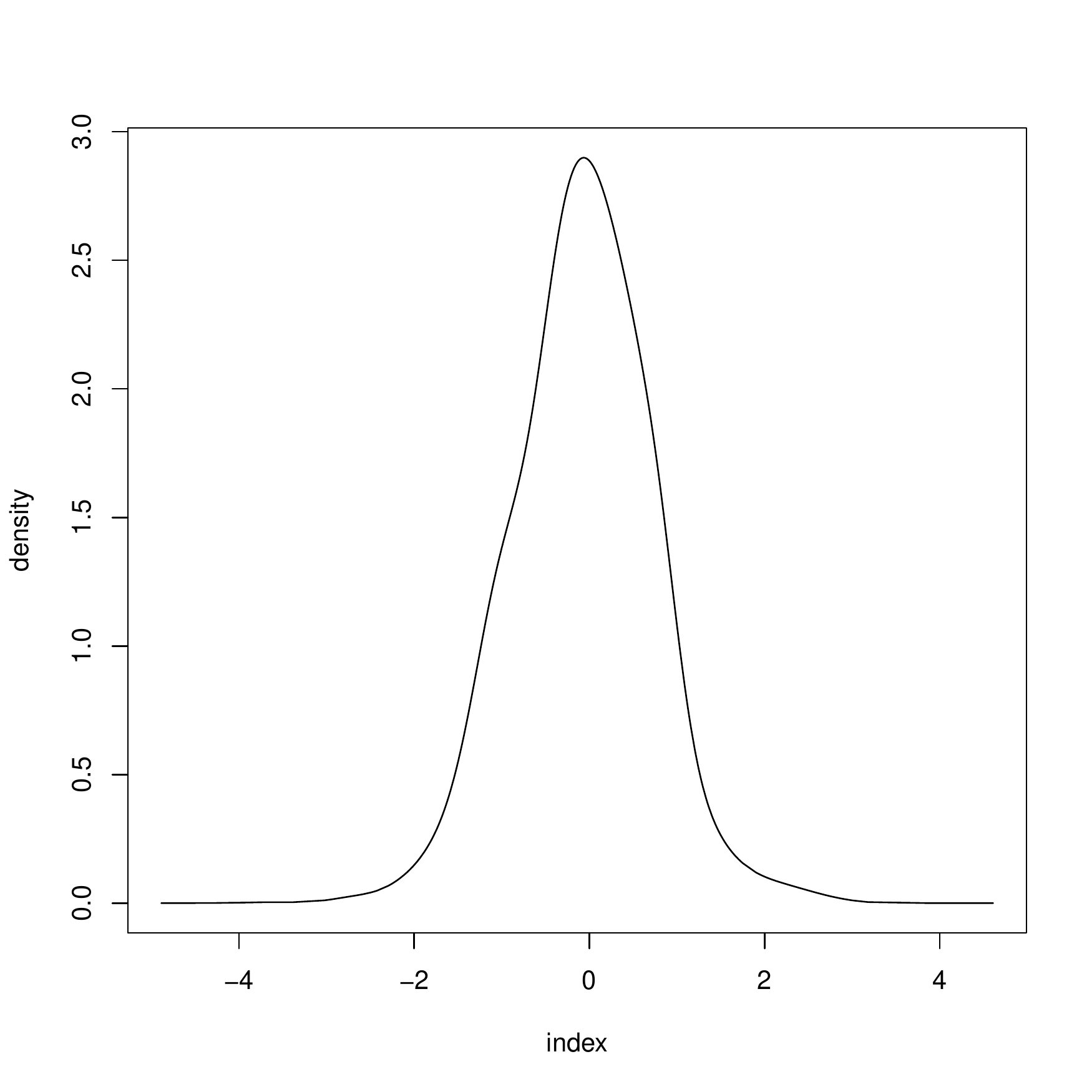}
\caption{Drawing of the estimator $\hat m (\cdot)$.} \label{fig:m}
\end{figure}

\color{black} 
\subsection{Mixing property for  the simulation setting in Section \ref{sec:num}}
We recall that we have generated data from model \eqref{eq:modelbase}-\eqref{eq:modelbase_c} with $\varepsilon_i=\sigma(V_i,Z_i^{\{r\}};\beta)\zeta_i$ and $$\sigma^2(V_i,Z_i^{\{r\}};\beta)=\beta_1 + \beta_2 Y_{i-1}^2,$$
where the $\zeta_i$ are independently drawn from a distribution such that $\mathbb{E}(\zeta_i)=0$ and $\text{Var}(\zeta_i)=1$. That means, we allow for conditional heteroscedasticity in the mean regression error term.     The covariates $X_i=(Y_{i-1},Y_{i-2})^\top$ are two lagged values of the target variable $Y_i$ and the covariates $W_i=(W_{i1},W_{i2},W_{i3})^\top$ are generated from a multivariate Gaussian distribution with  mean  $\rho W_{i-1}$  and covariance $S$ given by $\textrm{cov}(W_{ik},W_{i\ell})=0.5^{|k-\ell|}$. Thus, $$W_i = \rho W_{i-1} + \zeta_i^{*},$$ where the $\zeta_i^{*}$'s are i.i.d multivariate centered Gaussian $\mathcal{N}(0,S)$. We set 
\begin{equation*}
\ell(X_i;\gamma_1)=\gamma_{11}Y_{i-1} + \gamma_{12}Y_{i-2} \quad \text{ and }  \quad m(u)= \frac{3}{4}\sin^2(u\pi) ,
\end{equation*}
with $\gamma_1=(0.1,0)^\top$, $\gamma_2=(1,1,1)^\top$, $\rho= 1/4$ and $\beta=(0.9,0.1)^\top$.
We want to show that the process $(Z_i)_{i\in\mathbb Z}$, where 
$$
Z_i=(X_i^\top,W_i^\top,\varepsilon_i)^\top = (Y_{i-1}, Y_{i-2},W_i^\top,\varepsilon_i)^\top \in\mathbb{R} \times\mathbb{R}\times\mathbb{R}^{3}\times\mathbb{R}, 
$$
is stationnary and strongly mixing. For this purpose, we will show that the process $(\mathcal{Z}_i)_i$ defined by 
$$
\mathcal{Z}_{i} 
= ( Y_{i-1},Y_{i-2}, Y_{i-3}, W_{i-1}^\top, W_i^\top)^\top \in \mathbb{R} \times \mathbb{R}\times\mathbb{R} \times \mathbb{R}^{3}\times\mathbb{R}^{3},
$$
is geometrically ergodic thus strongly mixing with exponentially decaying mixing coefficients.  Since $Z_i$ is a measurable function of a subvector of $\mathcal{Z}_{i+1}$,  this will imply the result for $Z_i$.
 By our model, 
$\mathcal{Z}_{i}$ can be rewritten as follows
\begin{equation}\label{eqEZ}
\mathcal{Z}_{i}= \begin{pmatrix}
Y_{i-1}\\
Y_{i-2}\\
Y_{i-3}\\
W_{i-1}\\
W_{i}
\end{pmatrix} = F \left( \mathcal{Z}_{i-1}, \mathcal{Z}_{i-2} \right) +  H \left( \mathcal{Z}_{i-1}\right) \upsilon_i, 
\end{equation}
where 
$
\upsilon_i = 
(\zeta_{i-1}, \zeta_{i-2}, \zeta_{i-3},\zeta_{i-1}^{*\top}, \zeta_{i}^{*\top})^\top,
$
and
$F$ is a function from   $\mathbb{R}^9 \times \mathbb{R}^9$ to $\mathbb{R}^9$   \color{black}   given by
$$
F \left( \mathcal{Z}_{i-1}, \mathcal{Z}_{i-2} \right) =\begin{pmatrix}
\gamma_{11} Y_{i-2} + \gamma_{12} Y_{i-3} + m(W_{i-1}^\top \gamma_2 )\\
\gamma_{11} Y_{i-3} + \gamma_{12} Y_{i-4} +m(W_{i-2}^\top \gamma_2 )\\
\gamma_{11} Y_{i-4} + \gamma_{12} Y_{i-5} + m(W_{i-3}^\top \gamma_2 )\\
\rho W_{i-2} \\
\rho W_{i-1}
\end{pmatrix} ,$$
and 
$$
 H(\mathcal{Z}_{i-1})= diag \left(\sqrt{\beta_1 + \beta_2 Y_{i-2}^2}, \sqrt{\beta_1 + \beta_2 Y_{i-3}^2}, \sqrt{\beta_1 + \beta_2 Y_{i-4}^2}, 1, 1, 1, 1, 1, 1\right).$$
We apply the   Theorem 1 from \citet{LuJiang}   to show that there exists a   stationary solution of the equation \eqref{eqEZ} which is geometrically ergodic, and   thus $\alpha$-mixing with exponentially decaying coefficients.  For this purpose, we show that conditions (B2) and (B4) in \citet{LuJiang} are satisfied.  
In what follows, for any matrix $A$,  $\|A\|$ denotes the sum of the absolute value of all its coefficients.
 Since the function $m(\cdot)$ is non negative and bounded by 3/4, we have 
\begin{multline*}
\|F\left( \mathcal{Z}_{i-1}, \mathcal{Z}_{i-2}\right)\| \leq |\gamma_{11}| \sum_{k=1}^3 |\mathcal{Z}_{i-1,k}| + |\gamma_{12}| \sum_{k=1}^3|\mathcal{Z}_{i-2,k}| + |\rho| \sum_{k=4}^9 |\mathcal{Z}_{i-1,k}|+ 3 \times 3/4 \\
= \sum_{k=1}^9 \gamma_k^{*} |\mathcal{Z}_{i-1,k}| +  \sum_{k=1}^9 \gamma_k^{**} |\mathcal{Z}_{i-1,k}| + o\left(\left\| (\mathcal{Z}_{i-1}, \mathcal{Z}_{i-2})\right\| \right),
\end{multline*}
as $\left\| (\mathcal{Z}_{i-1}, \mathcal{Z}_{i-2})\right\|$ tends to infinity, where $\mathcal{Z}_{i-1,k}$ and $\mathcal{Z}_{i-2,k}$ stand for the components of $\mathcal{Z}_{i-1}$ and $\mathcal{Z}_{i-2}$ respectively, and
$$
\gamma_k^{*} = \left\lbrace \begin{array}{ll}
|\gamma_{11}| & \text{ if } k=1,\ldots, 3 \\
|\rho| & \text{ if } k=4, \ldots, 9,
\end{array}\right. \quad \text{ and } \quad 
\gamma_k^{**} = \left\lbrace \begin{array}{ll}
|\gamma_{12}| & \text{ if } k=1,\ldots, 3 \\
0 & \text{ if } k=4, \ldots, 9,
\end{array} \right.
$$
Moreover, 
$$\|H(\mathcal{Z}_{i-1})\| = \sum_{k=1}^9 \beta_k^{*} |\mathcal{Z}_{i-1,k}| + o(\|\mathcal{Z}_{i-1}\|),\qquad \text{ as } \; \left\| \mathcal{Z}_{i-1}\right\|\rightarrow \infty ,$$
where $\beta_k^{*} =|\beta_2|^{1/2}$ for $k=1, \ldots, 3$ and zero otherwise.
Thus, from  condition (B4) in \citet{LuJiang},  a sufficient condition to guarantee the geometric ergodicity is
\begin{equation}\label{cd_ergo}
\max_{1\leq k\leq 9} \left\lbrace \gamma_k^{*} + \gamma_k^{**} +  \beta_k^{*} \mathbb{E}[|\upsilon_{i,k}|]\right\rbrace < 1,
\end{equation}
where the $\upsilon_{i,k}$ stand for the components of $\upsilon_{i}$. 

With our simulation design, for $k=4, \ldots, 9$, 
$$\gamma_k^{*} + \gamma_k^{**} +  \beta_k^{*} \mathbb{E}[|\upsilon_{i,k}|] = |\rho| + 0 + 0 = \frac{1}{4} <1.$$
Hence, condition \eqref{cd_ergo} rewrites as
$$ 
\max_{1\leq k\leq 3} \left\lbrace |\gamma_{11}| + |\gamma_{12}| + |\beta_2|^{1/2} \mathbb{E}[|\upsilon_{i,k}|]\right\rbrace < 1
$$ 
Besides, noting that for each $1\leq k \leq 3$, $\mathbb{E}^2[|\upsilon_{i,k}|] \leq \mathbb{E}[\upsilon_{i,k}^2]= \mathbb{E}[\zeta_{i}^2]= 1$,
 we obtain that for any $k=1,\ldots,3$, $$|\gamma_{11}| + |\gamma_{12}| + |\beta_2|^{1/2}\mathbb{E}[|\upsilon_{i,k}|]  \leq \frac{1}{10} + 0 + \frac{1}{\sqrt{10}}.$$ Therefore, condition \eqref{cd_ergo} is fulfilled.

\color{black}

\subsection{A conditional variance semiparametric model}

In this section we extend the scope of our models. As mentioned in section \ref{sec:concl},  $Y_i$ could be observed with some error. For illustration, consider a time series $(R_i)$,  solution of the AR(1)  equation 
\begin{equation}
R_i = \rho_0 R_{i-1} + u_i,\qquad i\in\mathbb Z,\label{eq:newmodel}
\end{equation}  
Consider the PLSIM
\begin{equation}\label{apx_ar1}
Y_i = u_i^2 = \mu(V_i;\gamma,m) + \varepsilon_i \;\textrm{ with }\; \mu(V_i;\gamma,m)= l(X_i; \gamma_1) + m(W_i^\top \gamma_2)>0,
\end{equation}
with $\mathbb E[\varepsilon_i\mid V_i,\mathcal F_{i-1}]=0$ a.s. and $(V_i^\top,\varepsilon_i)^\top\in\mathbb{R}^{d_X+d_W}\times\mathbb{R}$ a strictly stationary and strongly mixing sequence. Moreover, $\mathbb E[u_i\mid V_i,\mathcal F_{i-1}]=0$ a.s.
A more common way to write model \eqref{apx_ar1} is
\begin{equation}\label{apx_ar2}
u_i = \sqrt{\mu_i}  \; \nu _i, 
\end{equation}
with $(\nu_i)$ a strong white noise process with unit variance, and  $\mu_i$ a positive function of the past values of $u_i$. The ARCH model is a typical example. When covariates are also allowed to enter the expression of $\mu_i$, one obtains a particular example of the so called GARCH-X models. See Han and Kristensen (2014). Here we allow for a flexible semiparametric form $\mu_i = \mu(V_i;\gamma,m)$ and our additive error term is 
$
\varepsilon_i = \mu(V_i;\gamma,m) (\nu_i^2-1).
$

Although the conditional variance of $Y_i$ is not constant and the $Y_i$ are not directly observed, the PLSIM is still applicable, as we will briefly justify in the following. Instead of $Y_i$, one has 
$$
\widetilde Y_i = \left( R_i - \widetilde \rho R_{i-1} \right)^2 = Y_i + R_{i-1}^2 (\widetilde \rho - \rho_0)^2 - 2u_iR_{i-1}(\widetilde \rho - \rho_0).
$$
Here, $\widetilde \rho$ is the least-squares estimator of $\rho_0$. Let $\widehat{\widetilde \eta}_\gamma$ be the vector of nonparametric estimators defined in 
\eqref{eq:eta_chap} obtained with $\widetilde Y_i$ instead of $Y_i$. Only the components $\widehat{\widetilde \eta}_{\gamma,m}$ and $\widehat{\widetilde \eta}_{\gamma,m^\prime}$ are affected by the fact that the $Y_i$'s are not available. Given the expression of $\widetilde Y_i - Y_i$ we deduce 
$$
\widehat{\widetilde \eta}_{\gamma,m}(t) = \widehat{ \eta}_{\gamma,m} (t)- \frac{2(\widetilde \rho - \rho_0)}{nh} \sum_{i=1}^n u_i R_{i-1}K\left(\frac{W_i^\top \gamma_2 - t}{h}\right) + (\widetilde \rho - \rho_0)^2 O_{\mathbb P }(1),
$$
uniformly with respect to $t$, and a similar representation holds true  for  $\widehat{\widetilde \eta}_{\gamma,m^\prime}$. 
 Using the fact that $\widetilde \rho - \rho_0= O_{\mathbb P }(n^{-1/2})$ and $\mathbb E[u_i\mid V_i,\mathcal F_{i-1}]=0$ a.s., the arguments in the proof of Theorem~\ref{thm:chi2etaunknown} remain valid and the limit of the ELR is still a chi-square distribution. 
The technical  details are quite straightforward and thus are omitted. Instead, we propose an illustration using simulation data. 

We generated data from model \eqref{eq:newmodel}-\eqref{apx_ar2}.  First, we generate the covariates $W_i=(W_{i1},W_{i2},W_{i3})^\top$ from a multivariate Gaussian distribution with  mean  $W_{i-1}/4$  and covariance given by $\textrm{cov}(W_{ik},W_{i\ell})=0.5^{|k-\ell|}$.  Then, we set $X_i=(U_{i-1}^2,U_{i-2}^2)^\top$ and we generate the $U_i^2$ from \eqref{apx_ar1}- \eqref{apx_ar2} with 
\begin{equation}\label{simu_design}
\ell(X_i;\gamma_1)=\gamma_{11}U_{i-1}^2 + \gamma_{12}U_{i-2}^2 \quad \text{ and }  \quad m(u)= \frac{1}{4}+\frac{3}{4}\sin^2(u\pi) ,
\end{equation}
with $\gamma_1=(0.1,0)^\top$, $\gamma_2=(1,1,1)^\top$. Finally, the variables $R_i$ can be computed from \eqref{eq:newmodel} with $\rho_0=0.1$. Based on the observed $R_i$'s, we can compute $\tilde Y_i$ and use the proposed EL procedure for testing the order of the lagged values to consider in the parametric
function.  Hypothesis testing is based on  Wilks' Theorem   in Section~\ref{subsec:wilks} (results related to this method are named \emph{estim}), along with the unfeasible EL approach  that previously learns the nonparametric estimators on a sample of size $10^4$ (this case mimic the situation where $m$,  $m'$ and the density of the index are known; results related to this method are named \emph{ref}).   The nonparametric elements are estimated by the Nadaraya-Watson method with Gaussian kernel and bandwidth $h=C^{-1} n^{-1/5}$ where $C$ is the standard deviation of the index. Thus, we consider the tests introduced in Section~\ref{sec:num}. The empirical probabilities of rejection are presented in Table~\ref{resARCH} for a nominal level of $0.05$.

\begin{table}[ht!]
\caption{Empirical probabilities of rejection obtained from 5000 replications using the PLSIM for testing the order for the lagged values of $Y_i$ in the parametric part $\ell(\cdot; \gamma_1)$.\label{resARCH}}

\centering
\begin{tabular}{rrrrrrrrr}
  \hline
Test & \multicolumn{2}{c}{$n=1000$} & \multicolumn{2}{c}{$n=2000$} & \multicolumn{2}{c}{$n=4000$} & \multicolumn{2}{c}{$n=8000$}\\ 
& ref. & estim. & ref. & estim. & ref. & estim. & ref. & estim. \\
  \hline
Lag(1) & 0.117 & 0.120 & 0.089 & 0.097 & 0.081 & 0.078 & 0.064 & 0.063 \\ 
Lag(0) & 0.422 & 0.402 & 0.678 & 0.654 & 0.942 & 0.932 & 1.000 & 0.999 \\ 
Lag(2) & 0.538 & 0.531 & 0.695 & 0.687 & 0.887 & 0.874 & 1.000 & 1.000 \\ 
   \hline
\end{tabular}
\end{table}

\qquad

\noindent \textbf{References}

\quad


\hangindent = 0.5cm
\hspace{-0.5cm}[1] Han, H. and D. Kristensen (2014). Asymptotic for the QMLE in
GARCH-X  Models With Stationary and Nonstationary Covariates. \emph{Journal of Business \& Economic Statistics 32}, 416--429.

\end{document}